\documentclass{article}
\usepackage{amsmath, amsthm, amssymb, thm-restate, authblk}
\usepackage[table,xcdraw]{xcolor}
\usepackage{style}
\usepackage[left=1in,bottom=1in,top=1in,right=1in]{geometry}
\hypersetup{%
	colorlinks   = true, %
	urlcolor     = blue, %
	linkcolor    = blue, %
	citecolor    = blue,  %
}

\graphicspath{{img/}}

\renewcommand{\epsilon}{\varepsilon}

\title{Embedding Probability Distributions into Low Dimensional $\ell_1$: Tree Ising Models via Truncated Metrics}
\author[1]{Moses Charikar}
\author[1]{Spencer Compton}
\author[1]{Chirag Pabbaraju}
\affil[1]{Computer Science Department, Stanford University \authorcr
  \{\tt moses, comptons, cpabbara\}@stanford.edu}

\date{\today}

\begin{document}

\maketitle

\thispagestyle{empty}
Given an arbitrary set of high dimensional points in $\ell_1$, there are known negative results that preclude the possibility of always mapping them to a low dimensional $\ell_1$ space while preserving distances with small multiplicative distortion. This is in stark contrast with dimension reduction in Euclidean space ($\ell_2$) where such mappings are always possible. While the first non-trivial lower bounds for $\ell_1$ dimension reduction were established almost 20 years ago, there has been limited progress in understanding what sets of points in $\ell_1$ are conducive to a low-dimensional mapping.

In this work, we study a new characterization of $\ell_1$ metrics that are conducive to dimension reduction in $\ell_1$. Our characterization focuses on metrics that are defined by the disagreement of binary variables over a probability distribution -- any $\ell_1$ metric can be represented in this form. We show that, for configurations of $n$ points in $\ell_1$ obtained from tree Ising models, we can reduce dimension to $\polylog(n)$ with constant distortion.
In doing so, we develop technical tools for embedding \emph{truncated metrics} which have been studied because of their applications in computer vision, and are objects of independent interest in metric geometry.
Among other tools, we show how any $\ell_1$ metric can be truncated with $O(1)$ distortion and $O(\log(n))$ blowup in dimension.

{
\clearpage
\vspace{1cm}
\renewcommand{\baselinestretch}{0.1}
\setcounter{tocdepth}{4}
\setcounter{secnumdepth}{4}
\tableofcontents{}
\thispagestyle{empty}
\clearpage
}

\setcounter{page}{1}

\section{Introduction}
\label{sec:intro}

Given an arbitrary set of $n$ high-dimensional points in $\ell_1$, we know that it is impossible to always map them to a low dimensional $\ell_1$ space while preserving distances with small multiplicative distortion.
Known lower bounds \cite{brinkman2005impossibility, lee2004embedding, andoni2011near, regev2013entropy} show that for any desired distortion $c$, at least $n^{\Omega(1/c^2)}$ dimensions are needed.
This is in stark contrast with dimension reduction in Euclidean space ($\ell_2$) where mappings to $O(\log(n))$ dimensions are always possible with $1+\epsilon$ distortion.
While these results rule out dimension reduction for arbitrary sets of point in $\ell_1$, a natural question is: {\em Are there broad families of structured sets of points in $\ell_1$ that admit low dimensional mappings with small distortion?}
Despite the fact that strong lower bounds for $\ell_1$ dimension reduction were established almost 20 years ago, there has been little progress in understanding this question.

In this work, we study a new characterization of $\ell_1$ metrics that admit low-dimensional $\ell_1$ representations. 
Our starting point is that any $\ell_1$ metric can be represented by the disagreement of binary variables over a probability distribution $D$ -- a simple consequence of the fact that any $\ell_1$ metric can be embedded into the (possibly infinite-dimensional) Hamming cube.
This distribution viewpoint of $\ell_1$ metrics leads to a natural reframing of the earlier question: {\em What classes of distributions $D$ lead to $\ell_1$ metrics that admit low-distortion, low-dimensional $\ell_1$ embeddings?}

We initiate a study of this question and show that, for configurations of $n$ points in $\ell_1$ obtained from {\em tree Ising models}, we can reduce dimension %
to $\polylog(n)$ with constant distortion.
Within the family of graphical models, tree Ising models are a natural first candidate for distributions to study in this context. On the one hand, they have been studied extensively in the recent theoretical computer science literature on testing and learning 
\cite{daskalakis2021sample, bhattacharyya2021near, boix2022chow, pmlr-v195-kandiros23a}.
On the other hand, it is well known that tree metrics embed isometrically into $\ell_1$, and the question of $\ell_1$ dimension reduction for tree metrics has been studied in the literature 
\cite{charikar2002dimension, lee2013dimension, lee2013lower}.
We should note, however, that the metrics that arise from tree Ising models are not tree metrics themselves -- establishing that they admit low dimensional $\ell_1$ embeddings is subtle and non-trivial.

In order to establish our results, we develop technical tools for embedding \emph{truncated metrics}. %
Truncated metrics arise in image segmentation in computer vision 
\cite{boykov1998markov, ishikawa1998segmentation, kumar2011improved, kumar2014rounding, kumar2016rounding, pansari2017truncated}
and have been studied in theoretical computer science in the context of the metric labeling problem motivated by this application \cite{gupta2000constant, chekuri2004linear, karloff2006earthmover}; they are also objects of independent interest in metric geometry, where they have been studied in the context of metric dimension
\cite{jannesari2012metric, beardon2016k, tillquist2021truncated, tillquist2023getting}, including very recent work on the metric dimension of truncated tree metrics \cite{bartha2023sharp, gutkovich2023computing}.
Our techniques for truncated metrics relate closely with recent developments in truncated metrics for embedding shortest path decompositions \cite{abraham2022metric,filtser2020face}. The work of \cite{abraham2022metric} provides embeddings of $O(\log(n))$ dimension and constant distortion for both ``truncated \textit{line} metrics'' and ``Lipschitz capped \textit{line} metrics'' (later defined in \Cref{sec:fixed-cap-metrics-prelims}, \ref{sec:lipschitz-cap-metrics-prelims}), whereas the work of \cite{filtser2020face} provides \textit{high-dimensional} and constant distortion embeddings for truncated tree metrics and Lipschitz capped tree metrics. Our work provides embeddings of $O(\log^2(n))$ dimension \textit{and} constant distortion for truncated tree metrics and Lipschitz capped tree metrics. Moreover, we provide a general result for truncating \textit{any} arbitrary $\ell_1$ metric that may be useful in future works dealing with truncation.

While our work shows that $\ell_1$ metrics represented by tree Ising models allow for low dimensional embeddings, it would be interesting to see how far this program can be extended -- namely, is there a more general class of graphical models (e.g., having bounded treewidth) or joint distributions that also allow for low dimensional embeddings in $\ell_1$? We discuss this in more detail in \Cref{sec:discussion}, where we show how we can express prior lower bounds for dimension reduction in $\ell_1$ in this distributional formulation. In particular, we show that even Bayesian networks with treewidth-3 capture the lower bound of \cite{brinkman2005impossibility}, precluding the possibility of embedding them in $n^{o(1)}$ dimensions. 
This suggests that tree Ising models are near the boundary of families of graphical models representing $\ell_1$ metrics that admit low dimensional embeddings.
We also discuss classes of graphical models for which this question remains open.

\paragraph*{Related Work:}
The study of dimension reduction in metric spaces has a rich history (see the survey~\cite{naor2018metric}).
Dimension reduction for $\ell_1$ has been studied in several contexts:
the performance of linear embeddings on structured subsets \cite{charikar2002dimension, berinde2008combining, arce2010reconstruction, rachkovskij2016real, krahmer2016unified, jacques2017small, lotz2019persistent} as well as nearest neighbor preserving embeddings -- a weaker requirement than low distortion (i.e. bilipschitz) embeddings \cite{indyk2007nearest, emiris2019near}.
The lower bounds for $\ell_1$ dimension reduction have been generalized to the nuclear norm 
\cite{naor2020impossibility, regev2020bounds}.

The question of characterizing metrics that embed into constant dimensional $\ell_2$ (and $\ell_p$ for $p>2$) in terms of the doubling dimension has received considerable attention in the literature. 
A major open problem in this space is to determine whether doubling metrics embed into $\ell_2$ with constant dimensions and constant distortion
\cite{lang2001bilipschitz, gupta2003bounded, lafforgue2014doubling, bartal2015impossibility, le2018isometric, tao2020embedding, baudier2021no}.

Similar in spirit to our study of structured subsets of $\ell_1$ that admit low dimensional, low distortion $\ell_1$ representations, is the study of graph metrics that embed into $\ell_1$ with constant distortion. A well-known open problem here is to determine whether graph metrics from minor-closed families (e.g. planar graphs) embed into $\ell_1$ with constant distortion
\cite{gupta2004cuts, rao1999small, chekuri2006embedding, lee2009geometry, lee2013pathwidth, sidiropoulos2013non, krauthgamer2019flow, filtser2020face, kumar2022approximate}.
Closer to our work, \cite{charikar2002dimension} studied classes of metrics that admit low-dimensional, low-distortion embeddings into $\ell_1$.
They focused on classes defined by metrics supported on families of graphs -- quite different from the new probability distribution based characterization initiated in our work.
They showed that metrics supported on outerplanar graphs, and more generally, $K_{2,3}$-free graphs, can be embedded into $\ell_1$ with $O(\log^2 n)$ dimensions and constant distortion.

Several other notions of dimension reduction are well studied in the literature. See
\cite{burges2010dimension, cunningham2015linear, bartal2019dimensionality, demaine2021multidimensional}.
Somewhat related to (but different from) the focus of this work is the question of approximating distance metrics on probability distributions with short vectors \cite{abdullah2016sketching, cicalese2016approximating, melucci2019brief}.

\section{Overview of Results}
\label{sec:overview}
\subsection{Embedding Probability Distributions}
In the context of understanding the classes of distributions that are amenable to low-dimensional $\ell_1$ embeddings, our main result shows that it is possible to embed distance metrics over Tree Ising models:

\begin{restatable}[Embedding Tree Ising Models]{theorem}{theoremtreeising}
    \label{thm:tree-ising}
    Any distance metric over tree Ising models (i.e, $d_D(i,j) = c \cdot \Pr_{X \sim D}[X_i \ne X_j]$, where $D$ is a tree Ising model distribution on $n$ random variables $X_1,\dots,X_n$) can be embedded into $(\R^{O(\log^2(n))}, \ell_1)$ with $\Theta(1)$ distortion.
\end{restatable}

We also show that this characterization of metrics with distributions is expressive enough to capture classical lower bounds of \cite{brinkman2005impossibility}. In particular, there exist Bayesian networks of treewidth 3 that are impossible to embed into low-dimensions with small distortion.

\begin{restatable}[Treewidth 3 lower bound]{claim}{claimlowerbound}
    \label{claim:treewidth-3-lower-bound}
    There exists a Bayesian network $D$ on $N$ binary random variables $X_1,\dots,X_N$ whose underlying undirected graph has treewidth 3, such that embedding the metric $d_D(X_i, X_j)=\Pr_D[X_i \neq X_j]$ into $\ell_1$ with $O(1)$ distortion requires $N^{\Omega(1)}$ dimensions.
\end{restatable}

The proof of our positive result \Cref{thm:tree-ising} follows two primary thrusts. First, in \Cref{sub:recution} we characterize that distance metrics corresponding to tree Ising models can approximately have distances between $X_i,X_j$ decomposed into three sources:

\begin{itemize}
    \item The difference in the marginals between $X_i,X_j$, or $|\Pr[X_i = 0] - \Pr[X_j=0]|$.
    \item The amount of independence between $X_i$ and $X_j$, which we formalize by a notion called ``Bernoulli randomness'' in \Cref{sub:bern}. This simply-defined quantity behaves intuitively similar to Shannon entropy in how it measures randomness, yet has smoothness that makes it more amenable to the $\ell_1$ setting.
    \item The strength of any negative correlation between $X_i$ and $X_j$.
\end{itemize}

Perhaps most interestingly, we sketch here how to embed distance originating from the ``independence'' between two variables $X_i$ and $X_j$. For simplicity, suppose that the other two sources of distance are zero for $X_i$ and $X_j$ (meaning, their marginals are the same and they are not negatively correlated). Then, one could roughly view the joint distribution over the $X_i,X_j$ by some process where first $X_i$ is realized, then with probability $p$ we sample $X_j$ independently, and with probability $1-p$ we set $X_j=X_i$. The parameter $p$ roughly corresponds to the independence between $X_i$ and $X_j$: if $p=0$ then they are completely dependent, and if $p=1$ they are completely independent. In this setting, we find that the distance metric $\Pr_D[X_i \ne X_j]$ is roughly equal to %
$\Theta(1) \cdot p \cdot M(x_i)$, where $M(x_i)=\min(\Pr[X_i=0], \Pr[X_i=1])$ is a quantity that measures how far $X_i$ is from being completely deterministic (i.e. biased towards one value).
We later non-trivially determine a weight $w_e$ for each edge in the tree (in terms of ``Bernoulli randomness''), such that, in this setting, if we were to think of the path connecting $X_i$ and $X_j$,
$\Pr_D[X_i \ne X_j] \approx \min(\sum_{e: i \rightarrow j} w_e, M(x_i))$.
One interpretation of this is that independence can accumulate along a path in the tree, but at some point the effect of independence is capped by the amount of randomness in the marginals of $X_i,X_j$. 
We are able to generalize the usage of this idea, to reduce the task of embedding tree Ising models into $\ell_1$ to the task of embedding a truncated tree metric, where the truncation parameter (we will call these ``caps'') is a function of the vertices. In particular, $\dlcap{M}(x_i,x_j) \triangleq \min(\sum_{e: i \rightarrow j} w_e, \max(M(x_i),M(x_j))$. Here, $M(x_i),M(x_j)$ are ``caps'' for each vertex, which, in the context of tree Ising models, ultimately end up being related to the biases of the variables. Crucially, our caps satisfy a Lipschitz property $|M(x_i) - M(x_j)| \le d(x_i,x_j)$. We next discuss our second thrust that studies how to embed this metric.

\subsection{Truncated Tree Metrics}

We saw that one of our main challenges in proving \Cref{thm:tree-ising} ends up being able to embed truncated tree metrics. As a warm-up (used in our proof for symmetric tree Ising models), we first study truncated tree metrics with a fixed cap. This metric is of independent interest; Lemma 1 of \cite{filtser2020face} provides a high-dimensional embedding using the results of \cite{chakrabarti2008embeddings}. We show a low-dimensional embedding:

\begin{restatable}[Fixed cap tree into $\ell_1$]{theorem}{theoremfixedcaptree}
    \label{thm:capped-tree}
    Let $T_n$ be an undirected edge-weighted tree on $n$ vertices $x_1,\dots,x_n$, equipped with the standard graph metric, i.e., $d(x_i, x_j)=\sum_{e:i \to j}e$, where $e:i\to j$ indexes the (weights) of all edges $e$ on the path from $x_i$ to $x_j$. Let $\dcap{M}$ be the distance function defined as $\dcap{M}(x_i,x_j) = \min(d(x_i, x_j), M)$.
    Then, $(T_n, \dcap{M})$ can be embedded into $(\R^{O(\log^2(n))}, \ell_1)$ with $\Theta(1)$ distortion.
\end{restatable}

Embedding this fixed cap tree metric will leverage three main ingredients: (i) truncated line metrics, (ii) the tree metric embedding of \cite{charikar2002dimension}, and (iii) a new truncated sparsity condition.

\emph{Truncated line metrics. } Truncated line metrics with a fixed cap are a special case of fixed cap tree metrics, and we will later leverage them as a subroutine. Lemma 2.1 of \cite{abraham2022metric} designs an embedding for truncated line metrics with a probabilistic Sawtooth function. While this is sufficient for the required subroutine for fixed cap tree metrics, in \Cref{sec:capped-line} we design an alternative embedding called \emph{lazy snaking} that is more conducive to our later result on truncating general $\ell_1$ metrics (\Cref{thm:general-cap}).

\emph{Tree metric embedding of \cite{charikar2002dimension}. } The work of \cite{charikar2002dimension} designs an embedding for tree metrics with $O(1)$ distortion and $O(\log^2(n))$ dimension. Summarizing their techniques, they first provide a simple isometric embedding of dimension $O(n)$. They then leverage the \emph{caterpillar decomposition} of trees to provide an isometric embedding of dimension $O(n)$ that is $O(\log(n))$-sparse (meaning each vertex embedding has at most $O(\log(n))$ nonzero entries). Moreover, they remark that any $d$-sparse embedding can be converted to a dimension $O(d \log(n))$ embedding with $O(1)$ distortion, by considering an embedding that hashes every coordinate into just $d$ coordinates, constructing a $d$-dimensional embedding where the $i$-th coordinate is a line metric corresponding to the total value hashed to the $i$-th coordinate, and then boosting the success probability $O(\log(n))$ times. Thus, as the caterpillar decomposition provided a $O(\log(n))$-sparse isometric embedding, they attain a $O(\log^2(n))$ dimension embedding with $O(1)$ distortion from their sparsity property.

\emph{Truncated sparsity condition. } Inspired by this approach, we introduce a new sparsity condition similar to the one leveraged for tree metrics that enables embedding a \textit{truncated} tree metric. Namely, for a cap value $M$ and target sparsity $d$, consider a modified sparsity condition such that: (i) each coordinate is bounded by $O(M/d)$ and (ii) for any pair $x_i,x_j$ there exists a subset of $O(d)$ coordinates such that the norm of the difference in embeddings restricted to just that subset of coordinates is within a constant factor of $\dcap{M}(x_i,x_j)$. Note how this condition is a generalization of the \cite{charikar2002dimension} sparsity condition when $M=\infty$. We show that, given this truncated sparsity condition, we can embed the truncated metric into $O(d \log(n))$ dimensions with $O(1)$ distortion. This is accomplished by a similar hashing method, where we replace the line metrics used in \cite{charikar2002dimension} with truncated line metrics of parameter $M/d$. Finally, we note that that the caterpillar decomposition embedding used in \cite{charikar2002dimension} does not naively satisfy our desired truncated sparsity condition, and we have to appropriately modify it to accomplish our requirements.

The above deals with fixed cap truncated tree metrics; however, to prove \Cref{thm:tree-ising}, we must embed a more nuanced truncated tree metric where the truncation is not fixed. We note that the techniques of Lemma 2 of \cite{filtser2020face} can enable a \textit{high-dimensional} embedding for this version of the problem. We show a low-dimensional embedding:

\begin{restatable}[Embedding Lipschitz-Capped Trees]{theorem}{theoremlipschitzcappedtree}
 \label{thm:lipschitz-capped-tree}
    Let $T_n$ be an undirected edge-weighted tree on $n$ vertices $x_1,\dots,x_n$, equipped with the standard graph metric, i.e., $d(x_i, x_j)=\sum_{e:i \to j}e$, where $e:i\to j$ indexes the (weights) of all edges $e$ on the path from $x_i$ to $x_j$. Consider a cap function $M(\cdot)$ that assigns a nonnegative cap to every vertex $x_i$ in the tree, such that cap values satisfy $|M(x_i)-M(x_j)| \le d(x_i, x_j)$ for all $i,j$. Let $\dlcap{M}$ be the distance function defined as
    \begin{align}
        \dlcap{M}(x_i, x_j) = \min(d(x_i, x_j), \max(M(x_i), M(x_j))).
    \end{align}
    Then, $(T_n, \dlcap{M})$ can be embedded into $(\R^{O(\log^2(n))}, \ell_1)$ with $\Theta(1)$ distortion.
\end{restatable}

Recall how our algorithm for the fixed cap tree metric involved taking a suitably modified version of the tree's caterpillar decomposition and reducing (by hashing) the embedding problem to a collection of fixed cap line metric problems. However, we cannot reduce the Lipschitz cap tree metric to a collection of Lipschitz cap line metrics in an obvious manner. Crucially, a direct modification of the prior approach would not work for Lipschitz cap tree metrics, primarily because segments that are hashed to a line metric instance are not contiguous segments of the tree, and thus \emph{the caps on these segments need no longer satisfy the necessary Lipschitz condition.} Consequently, we must design a new algorithm for this task. 

First, we observe that the primary issue in extending our prior techniques was how tree edges hashed to the same line metrics potentially have very different caps. If we view the algorithm hierarchically (considering how the embedding evolves as we progress downwards in the tree), it appears to be desirable to ``clean''/zero out the embedding somehow. If we clean the embedding frequently enough, we may expect that the only edges affecting some $\embed[x_i]$ would be parts of the tree very close to node $x_i$, and thus by the Lipschitz property, all the relevant edges would have roughly the same cap. Ultimately, we obtain an algorithm that is less directly similar to the previous approaches, but is conducive towards this notion of periodic cleaning. This result is ultimately achieved by a novel algorithm that we call the \textit{Build-Clean} approach (\Cref{sec:lipschitz-tree}).

\subsection{General Truncated $\ell_1$ Metrics}

Finally, we study how to truncate arbitrary $\ell_1$ metrics. This is likely a result of independent interest, as truncation has been leveraged in the $\ell_1$ embeddings literature for a variety of applications. For example, \cite{gupta2000constant} provides a custom construction for embedding the truncated line metric that attains $O(1)$ distortion \emph{only in expectation}. Later, \cite{abraham2022metric} proved a truncated line metric with $O(1)$ distortion and leveraged this for embedding weighted paths of bounded pathwidth. Truncation has been further leveraged for embedding planar graphs \cite{filtser2020face}. We show how to truncate any $\ell_1$ metric with $O(1)$ distortion and a $O(\log(n))$ dimension blowup:

\begin{restatable}[General Truncated $\ell_1$ Metrics]{theorem}{theoremgeneralcap}
    \label{thm:general-cap}
    Let $S=\{x_1,\dots,x_n\} \subseteq \R^d$ be $n$ points. Let $d(x_i,x_j)=\|x_i-x_j\|_1$, and for any fixed cap $M > 0$, consider the truncated $\ell_1$ metric $\dcap{M}(x_i,x_j)=\min(d(x_i, x_j), M)$. Then $(S, \dcap{M})$ can be embedded into $(\R^{O(d\log(n))}, \ell_1)$ with $\Theta(1)$ distortion.
\end{restatable}

We note that this $O(\log(n))$ dimension blowup is optimal upto constant factors. For example, consider how a line metric (by definition) can be embedded in $\ell_1$ using 1 coordinate. Observe how a uniform metric (where, say, every pair of points has distance 1) can be represented by a truncated line metric with a cap of $1$ and all points sufficiently far apart on the line. Moreoever, this uniform metric requires $\Omega(\log(n))$ dimension to obtain $O(1)$ distortion by a simple folklore volume argument. Thus, truncating the line metric requires a $\Omega(\log(n))$ dimension blowup and thus our general truncation result is optimal within constant factors.

Our result is attained by leveraging similar techniques to those used for embedding Lipschitz-Capped trees. Conceptually, we show that using these techniques enables one to map any metric in $\ell_1^d$ to a metric in $\ell_1^{d^2}$ that almost has similar distances to the truncated metric over the original points, yet satisfies the property that all coordinates have bounded magnitude of $O(\frac{M}{d})$, and that it almost resembles an embedding on the Hamming cube insofar as all but $O(d)$ coordinates are one of two distinct values. This satisfies our truncated sparsity condition, enabling us to finally reduce the dimension from $O(d^2)$ to $O(d \log(n))$ and obtain the desired truncated metric.

\section{Preliminaries and Notation}
\label{sec:preliminaries}

\subsection{Notation}
\label{sec:notation}
The notation $[n]$ denotes the integers $\{1,\dots, n\}$. $\R$ denotes the set of real numbers, and $\R_{\ge 0}, \R_{> 0}$ denote the set of nonnegative and strictly positive reals respectively. Random variables are denoted by capital letters (e.g., $X_i$), and points in metric spaces are denoted by lower case letters (e.g., $x_i$). We use array-indexing notation to index into coordinates of vectors. For example, if $x \in \R^d$, $x[k]$ denotes the $k^\text{th}$ coordinate of $x$. Points in a metric space will generally be denoted by $x_1,\dots,x_n$. We will often seek to obtain embeddings of these points into a lower dimensional space. We will denote the embedding of a point $x_i$ by $\embed[i]$, and sometimes by $\embed[x_i]$---both of these are supposed to stand for the same thing. Similarly, when talking about distances $d(x_i, x_j)$ between two points $x_i$ and $x_j$ in the metric space, we will interchangeably use both $d(x_i, x_j)$ and $d(i,j)$. We extensively make use of asymptotic notation $(\Omega(1), O(1))$ in the place of absolute constants. These constants will not depend on the problem parameters, unless explicitly specified. For example, at multiple places, we lower bound the probability of an event of interest by an absolute constant larger than 0, and to avoid tracking this constant, we use the $\Omega(1)$ notation. We emphasize that we do not require $n$ to be sufficiently large when we use this notation, where $n$ might be the number of joint random variables or points in the metric space.

\subsection{Tree Ising Models}
\label{sec:tim-prelims}
Given an underlying undirected tree on vertices $1,\dots, n$, let $i \sim j$ denote the existence of an edge between $i$ and $j$. Corresponding to this tree, a \textit{Tree Ising Model} defines a joint probability distribution $D$ over $\{0,1\}^n$ as follows:
\begin{equation}
    \label{eqn:tim-def}
    \Pr_D(X_1,\dots,X_n) \propto \exp\left(\sum_{i \sim j}\beta_{ij}\cdot(-1)^{X_i \oplus X_j}+\sum_{i}\gamma_i\cdot(-1)^{X_i}\right).
\end{equation}
Here, $\oplus$ denotes the XOR operation. The parameters $\beta_{ij}$ capture the \textit{pairwise interactions} between adjacent variables in the tree, while the parameters $\gamma_i$ specify the \textit{external field} on the individual variables. The Tree Ising Model given by \eqref{eqn:tim-def} for a given set of parameters $\{\beta_{ij}\}, \{\gamma_i\}$ can also be completely characterized by specifying the appropriate marginal probabilities for every variable $X_i$, and joint probabilities for every pair $(X_i, X_j)$ corresponding to an edge $i \sim j$ in the underlying tree. Under this specification, we can generate a sample from the model as follows: we arbitrarily root the tree at node 1, and direct all edges in the tree away from it. We realize a value for $X_1$ by sampling from its marginal distribution. Thereafter, for every directed edge $i \to j$, if we have realized the value of $X_i$ and haven't yet realized the value of $X_j$, we do so by sampling from the \textit{conditional} distribution $\Pr_D[X_j|X_i]$. In other words, tree Ising models are equivalent to arbitrary tree-structured Bayesian networks (with all edges directed away from the root) over binary alphabets. This equivalent characterization has been described and used in other prior works, e.g., \cite{daskalakis2021sample,pmlr-v195-kandiros23a}.

\subsection{Symmetric Tree Ising Models}
\label{sec:timwef-prelims}
If all the $\gamma_i$'s are 0 in the expression for the joint probability distribution in \Cref{eqn:tim-def}, we obtain a \textit{symmetric} tree Ising model, or a tree Ising model with \textit{no external field}. The joint distribution $D$ for these models is simply given by
\begin{equation}
    \label{eqn:timwef-def}
    \Pr_D(X_1,\dots,X_n) \propto \exp\left(\sum_{i \sim j}\beta_{ij}(-1)^{X_i \oplus X_j}\right).
\end{equation}
In a symmetric tree Ising model, $\Pr_D[X_i=0]=\Pr_D[X_i=1]=0.5$ for all $i \in [n]$. These models are widely studied in the literature---the assumption of no external fields makes the analysis of these tree-structured models significantly easier, while also capturing the central aspects of a number of problems. For example, \cite{bresler2020learning}, and more recently \cite{boix2022chow}, both study the Chow-Liu algorithm \cite{chow1968approximating} and its variants for \textit{learning} tree Ising models under the assumption of no external field---this assumption is necessary to make their analysis tractable. As we shall see, while we are able to obtain $\ell_1$ dimension reduction results more generally for tree Ising models even with an external field, our analysis for the symmetric case ends up being much simpler. %

We can identify a \timwef~uniquely given the underlying tree, and a single parameter $\theta_{ij}$ corresponding to every edge $i \sim j$ that specifies ``flip'' probabilities, where
\begin{equation}
    \label{eqn:timwef-edge-flip-map}
    \theta_{ij} = \Pr_D[X_i \neq X_j] = \frac{1}{1+\exp(2\beta_{ij})}.
\end{equation}
Given this characterization, we can generate a sample from the \timwef~as follows: we arbitrarily root the tree at node 1, and direct all edges away from the root. We first draw a uniformly random value in $\{0,1\}$ for $X_1$. Thereafter, for every directed edge $i \to j$, if we have realized the value of $X_i$ and haven't yet realized the value of $X_j$, we do so as follows: independently, with probability $\theta_{ij}$, we set $X_j = 1-X_i$, and with probability $1-\theta_{ij}$, we set $X_j=X_i$. A sample generated via this process has the same distribution as \Cref{eqn:timwef-def}, given the assignment to the $\theta_{ij}$'s as in \Cref{eqn:timwef-edge-flip-map}. In other words, symmetric tree Ising models are equivalent to tree-structured Bayesian networks (with all edges directed away from the root) over binary alphabets, where the conditional distributions for every edge $i \sim j$ satisfy the symmetries that $\Pr_D[X_j = 0 | X_i = 0] = \Pr_D[X_j = 1 | X_i = 1] = 1 - \Pr_D[X_j = 1 | X_i = 0] = 1 - \Pr_D[X_j = 0 | X_i = 1]$.

\subsection{Metric Spaces of Interest}
\label{sec:metrics-prelims}
We define explicitly the metric spaces that we study in our paper here. Recall that a metric space $(S, d)$ is defined by a set of points $S$, together with a distance function $d$ that maps pairs of points to nonnegative numbers, and satisfies (i) symmetry, (ii) triangle inequality and (iii) the property that every point has zero distance only to itself.

\begin{enumerate}
    \item $\ell_1$ metric: This metric space comprises of a set of $n$ points $\{x_1,\dots,x_n\} \subseteq \R^d$, where the distance function $d(x_i, x_j)=\|x_i-x_j\|_1=\sum_{k=1}^d|x_i[k]-x_j[k]|$.
    \item Tree metric: Consider an undirected edge-weighted tree $T_n$ on $n$ vertices $x_1,\dots,x_n$. Let $d(x_i, x_j)$ denote the standard graph distance on $T_n$, namely $d(x_i, x_j)=\sum_{e \in x_i \to x_j}e$, where the notation $e \in i \to j$ indexes the (weights on the) edges on the shortest path from $x_i$ to $x_j$.
    \item Line metric: This is simply a special case of the tree metric, where the tree is a line $L_n$ of vertices $x_1\to x_2 \to \dots \to x_n$. We can then instead think of the vertices as being points on the real line spaced according to the edge weights (with say $x_1$ at the origin).
    \item Tree Ising model metric: Given a tree Ising model $D$ on $n$ random variables $X_1,\dots,X_n$, the distance function $d_D$ is given by $d_D(X_i, X_j)=\Pr_D[X_i \neq X_j]$.\footnote{This is technically a \textit{pseudo-metric} space, because it may not satisfy property (iii) from above---two different variables $X_i$ and $X_j$ may both have $\Pr_D[X_i=1]=\Pr_D[X_j=1]=1$, in which case they have zero distance. We will not care too much about this  distinction.}
\end{enumerate}

We make a simple observation here: any $\ell_1$ metric can be alternatively viewed as $d(x_i,x_j) = c \cdot \Pr_D[X_i \neq X_j]$ where $D$ is \textit{some} distribution over $\{0,1\}^n$ (not necessarily a tree Ising model). To see this, let $x_1,\dots,x_n$ be $d$-dimensional points in $(\R^d, \ell_1)$. Now, consider shifting and scaling all the points, so that all coordinates of every point are in $[0,1]$. This preserves distances between points upto the constant scaling factor. %
Consider a joint distribution $D$ on $n$ binary-valued random variables $X_1,\dots,X_n$. A sample from this joint distribution is obtained as follows: first, we choose a coordinate $k \in [d]$ uniformly at random. Then, we sample a real number $p$ uniformly at random from the interval $[0,1]$. For each $i \in [n]$, we set $X_i = 1$ if $x_i[k] \le p$, and $0$ otherwise. Here, $x_i[k]$ is the $k^\text{th}$ coordinate of $x_i$. Then, observe that for any $i,j$,
\begin{align*}
    \Pr_D[X_i \neq X_j] &= \frac{1}{d}\sum_{k=1}^d\left|x_i[k]-x_j[k]\right| = \frac{1}{d}\cdot\|x_i-x_j\|_1.
\end{align*}
This observation motivates the main question in our work: is there a rich enough class of distributions $D$ such that $\ell_1$ metrics corresponding to that class embed well into low-dimensional spaces? In fact, while this distributional view in principle has enough expressive power to capture any $\ell_1$ metric, we also show how we can explicitly view the lower bound of \cite{brinkman2005impossibility} in this framework.

Finally, when we say that a metric space $(S_1, d_1)$ embeds into the metric space $(S_2, d_2)$ with distortion $\alpha > 0$, this means that there exists a function $\sigma:S_1 \to S_2$ and $\beta > 0$, which satisfies, for all pairs $x_i, x_j$ in $S_1$, the relation
\begin{align}
    \label{eqn:metric-embedding-def}
    \beta \cdot d_1(x_i, x_j) \le d_2(\sigma(x_i), \sigma(x_j)) \le \alpha\beta \cdot d_1(x_i, x_j).
\end{align}
\subsubsection{Fixed Cap Metrics}
\label{sec:fixed-cap-metrics-prelims}
 In our work, we will primarily be interested in studying \textit{capped} versions of metric spaces. Consider a metric space on a set $S$ with distance function $d$. As a start, consider a fixed cap $M \in \R_{\ge 0}$. The fixed cap metric space on $S$ corresponding to the fixed cap $M$ is given by the distance function $\dcap{M}$ defined as
\begin{equation}
    \label{def:fixed-cap-distance}
    \dcap{M}(x_i, x_j) = \min(d(x_i, x_j), M).
\end{equation}
The metric given by \Cref{def:fixed-cap-distance} has also been referred to as the \textit{truncated} metric in the literature, as mentioned in the introduction.

\subsubsection{Lipschitz Cap Metrics}
\label{sec:lipschitz-cap-metrics-prelims}
We can also consider metric spaces on $S$ where the cap varies across the different points in the space, albeit smoothly with the distance $d$. Concretely, let $M:S \to \R_{\ge 0}$ be a nonnegative cap function which satisfies the following \textit{Lipschitz} property: for any $x_i, x_j$\footnote{When talking about the cap at a point $x_i$, we will sometimes be loose and denote it interchangeably by $M(x_i)$ and $M(i)$.} in $S$,
\begin{equation}
    \label{eqn:lipschitz-property}
    |M(x_i)-M(x_j)| \le d(x_i, x_j).
\end{equation}
The Lipschitz cap metric space on $S$ corresponding to the Lipschitz cap function $M$ is given by the distance function $\dlcap{M}$ defined as
\begin{equation}
    \label{def:lipschitz-cap-distance}
    \dlcap{M}(x_i, x_j) = \min(d(x_i,x_j), \max(M(x_i), M(x_j))).
\end{equation}

\subsection{Caterpillar Tree Decomposition}
\label{sec:caterpillar-prelims}
Our low-dimension embedding for capped tree metrics crucially involves using a particular tree decomposition technique, known as the ``caterpillar" decomposition, or also the ``heavy-light" decomposition. This decomposition has been used in the past \cite{linial1998low,gupta1999embedding,gupta2001improved,charikar2002dimension} for embedding tree metrics into $\ell_1$. We briefly describe the caterpillar decomposition here. Given an undirected tree $T_n$ on $n$ vertices $x_1,\dots,x_n$, let us arbitrarily root the tree at $x_1$. The caterpillar decomposition then decomposes the edges of the tree into several disjoint vertical\footnote{A path is vertical if for every pair of two points on the path, one is an ancestor of the other.} paths called ``caterpillars", such that any root-to-leaf path touches at most $\log n$ caterpillars. Equivalently, we can walk up to the root from any node in the tree by traversing at most $\log n$ caterpillars. Every tree allows such a caterpillar decomposition, which can be computed in a fairly straightforward manner using a single depth-first search.

\section{Symmetric Tree Ising Models}
\label{sec:timwef}
We begin our analysis of embedding tree Ising models into $\ell_1$ with the simpler case of tree Ising models with no external field (\Cref{eqn:timwef-def}). The techniques we develop for embedding these models will end up being building blocks for embedding general tree Ising models that also have an external field.

\subsection{Symmetric Tree Ising Models Reduce to Fixed Cap Tree Metrics}
\label{sec:timwef-to-fixed-cap-tree}

Given a \timwef~$D$, the metric of interest that we want to embed into $\ell_1$ with few dimensions is
\begin{align*}
    d_D(i, j) &= \Pr_D[X_i \neq X_j].
\end{align*}

First, let us consider the \timwef~$D'$ with the same tree structure, but where the flip probabilities on the edges (\Cref{eqn:timwef-edge-flip-map}) are replaced from $\theta_{ij}$ to $\min(\theta_{ij}, 1-\theta_{ij})$, so that all the flip probabilities are in $[0,0.5]$. We will show later how we are able to reduce to this case without loss of generality---for now, let us assume this is possible. We have the following lemma:

\begin{lemma}[\timwef~with no ``bad" edges $\to$ fixed cap tree]
    \label{lem:timwef-no-bad-edges-to-capped-tree}
    Let $D'$ be a \timwef~where all the edge-flip probabilities are at most $0.5$. Define the following capped embedding:
    \begin{align*}
        \dcap{0.5}(i, j) &= \min\left(\sum_{e \in (i\to j)}\theta_{e},\quad 0.5\right),
    \end{align*}
    where the notation $e \in (i \to j)$ indexes the edges on the (undirected) path from $i$ to $j$ in the tree, and $\theta_e$ denotes the weight on the edge. Then, for all $i \neq j$,
    \begin{align*}
        0.5 \cdot d_{D'}(i, j) \le \dcap{0.5}(i, j) \le 8 \cdot d_{D'}(i, j).
    \end{align*}
\end{lemma}
\begin{proof}
    Fix $i,j$, %
    and let us think of the generative process that realizes $X_i$ and $X_j$. %
    Observe that $X_i \neq X_j$ if and only if we decide to flip realizations on an odd number of edges on the path from $i$ to $j$. Then, by the union bound,
    \begin{align*}
        &\Pr_{D'}[\text{at least one edge flip on $i \to j$}] \le \sum_{e \in (i \to j)} \theta_{e} \\
        \implies \qquad &\Pr_{D'}[\text{not a single edge flip on $i \to j$}] \ge 1-\sum_{e \in (i \to j)} \theta_{e}.
    \end{align*}
    Further, recall that we decide whether to flip on an edge independently of the others. We can then say that
    \begin{align*}
        \Pr_{D'}[\text{exactly one edge flip on $i \to j$}] &= \sum_{e \in (i \to j)} \Pr_{D'}\left[e \text{ flips}\right]\cdot \Pr_{p'}\left[\text{edges on $i \to j$ other than $e$ don't flip}\right] \\
        &\ge \sum_{e \in (i \to j)} \Pr_{D'}\left[e \text{ flips}\right]\cdot \Pr_{D'}\left[\text{all edges on $i \to j$ don't flip}\right] \\
        &\ge \left(1-\sum_{e \in (i \to j)} \theta_{e}\right) \sum_{e \in (i \to j)} \theta_{e}.
    \end{align*}
    We have two cases:\\
    \noindent Case 1: $\sum_{e \in (i \to j)} \theta_{e} \le 0.5$. \\
    In this case, $\dcap{0.5}(i, j) = \sum_{e \in (i \to j)}\theta_{e}$. We have that
    \begin{align*}
        d_{D'}(i, j) &= \Pr_{D'}[X_i \neq X_j] = \Pr_{D'}[\text{odd number of flips on $i \to j$}] \\
        &\ge \Pr_{D'}[\text{exactly one flip on $i \to j$}] \ge \underbrace{\left(1-\sum_{e \in (i \to j)} \theta_{e}\right)}_{\ge 0.5} \underbrace{\sum_{e \in (i \to j)} \theta_{e}}_{=\dcap{0.5}(i, j)} \\
        &\ge 0.5 \cdot \dcap{0.5}(i, j).
    \end{align*}
    Furthermore, observe that
    \begin{align*}
        d_{D'}(i, j) &= \Pr_{D'}[\text{odd number of flips on $i \to j$}] \le \Pr_{D'}[\text{at least one flip on $i \to j$}] \le \sum_{e \in (i \to j)} \theta_{e} = \dcap{0.5}(i, j).
    \end{align*}
    Together, we get that
    \begin{align*}
        d_{D'}(i, j) \le \dcap{0.5}(i, j) \le 2\cdot d_{D'}(i, j).
    \end{align*}
    \noindent Case 2: $\sum_{e \in (i \to j)} \theta_{e} > 0.5$. \\
    In this case, $\dcap{0.5}(i, j) = 0.5$. One direction of what we want to show is easy:
    \begin{align*}
        &d_{D'}(i, j) = \Pr_{D'}[X_i \neq X_j] \le 1 \\
        \implies \qquad & 0.5 \cdot d_{D'}(i, j) \le 0.5 = \dcap{0.5}(i, j).
    \end{align*}
    For the other direction, we want to lower bound $d_{D'}(i, j)$ by a constant. Observe that since $\sum_{e \in (i \to j)} \theta_{e} > 0.5$ and all $\theta_{e} \in [0,0.5]$, there must exist a prefix of edges on the path from $i$ to $j$ such that $0.25 \le \sum_{e \in \prefix} \theta_{e} < 0.75$. For this prefix, using the same arguments as above, we have that
    \begin{align*}
        \Pr_{D'}[\text{even flips on prefix}] &\ge \Pr_{D'}[\text{no flips on prefix}] 
        \ge 1-\sum_{e \in \prefix} \theta_{e} 
        > 1-0.75=0.25.
    \end{align*}
    \begin{align*}
        \Pr_{D'}[\text{odd flips on prefix}] \ge \Pr_{D'}[\text{exactly one flip on prefix}] &\ge \left(1-\sum_{e \in \prefix} \theta_{e}\right)\sum_{e \in \prefix} \theta_{e} \\
        &> (0.25)^2 = 2^{-4}.
    \end{align*}
    Therefore, we get that
    \begin{align*}
    d_{D'}(i, j) &= \Pr_{D'}[X_i \neq X_j] = \Pr_{D'}[\text{odd number of flips on $i \to j$}] \\
    &\hspace{-1cm}= \Pr_{D'}[\text{odd flips on prefix}]\Pr_{D'}[\text{even flips on suffix}] + \Pr_{D'}[\text{even flips on prefix}]\Pr_{D'}[\text{odd flips on suffix}] \\
    &\hspace{-1cm}\ge 2^{-4}\left(\Pr_{D'}[\text{even flips on suffix}] + \Pr_{D'}[\text{odd flips on suffix}]\right) = 2^{-4} = 2^{-3} \cdot \dcap{0.5}(i, j).
    \end{align*}
    Putting the two bounds together, we have in all
    \begin{align*}
        0.5\cdot d_{D'}(i, j) \le  \dcap{0.5}(i, j) \le 8 \cdot d_{D'}(i, j).
    \end{align*}
    This completes the proof of the lemma.
\end{proof}  

We will now see how we are able to reduce from an arbitrary \timwef~$D$ with edge-flip probabilities $\theta_{ij}\in [0,1]$ to a \timwef~$D'$ with edge-flip probabilities $\min(\theta_{ij}, 1-\theta_{ij}) \in [0,0.5]$.
Let us call an edge $e$ ``bad'' if $\theta_e > 0.5$, and ``good'' if $\theta_e \le 0.5$. Then, we have the following claim:
\begin{claim}
    \label{claim:bad-edges-timwef}
    For any pair of nodes $i,j$,
    \begin{align*}
    &\text{(1)  If the number of bad edges between $i$ and $j$ is even, then } d_{D}(i, j) = d_{D'}(i, j). \\
    &\text{(2)  If the number of bad edges between $i$ and $j$ is odd, then } d_{D}(i, j) \ge 0.5.
    \end{align*}
\end{claim}
\begin{proof}
    For the first part, observe that
    \begin{align*}
        d_{D}(i, j) &= \Pr_{D}[\text{odd number of flips on $i \to  j$}] \\
        &= \Pr_{D}[\text{odd flips on bad edges}]\Pr_{D}[\text{even flips on good edges}] \\ &\qquad+ \Pr_{D}[\text{even flips on bad edges}]\Pr_{D}[\text{odd flips on good edges}] \\
        &= \Pr_{D}[\text{odd flips on bad edges}]\Pr_{D'}[\text{even flips on good edges}] \\
        &\qquad+ \Pr_{D}[\text{even flips on bad edges}]\Pr_{D'}[\text{odd flips on good edges}],
    \end{align*}
    where the preceding equality follows because weights on good edges don't change from $D$ to $D'$. For the first part, observe crucially that if the number of bad edges between $i$ and $j$ is even, the parity of flips and non-flips on the bad edges is always the same. Hence, we can replace non-flips to flips from $D$ to $D'$, to obtain
    \begin{align*}
        &\Pr_{D}[\text{odd flips on bad edges}] = \Pr_{D}[\text{odd non-flips on bad edges}]= \Pr_{D'}[\text{odd flips on bad edges}] \\
        &\Pr_{D}[\text{even flips on bad edges}] = \Pr_{D}[\text{even non-flips on bad edges}] = \Pr_{D'}[\text{even flips on bad edges}].
    \end{align*}
    This gives us that
    \begin{align*}
        d_{D}(i, j) &= \Pr_{D'}[\text{odd flips on bad edges}]\Pr_{D'}[\text{even flips on good edges}] \\
        &\qquad+ \Pr_{D'}[\text{even flips on bad edges}]\Pr_{D'}[\text{odd flips on good edges}] \\
        &= \Pr_{D'}[\text{odd number of flips on $i \to  j$}] = d_{D'}(i,j).
    \end{align*}
    For the second part, let $Y_{e}$ be a $\pm1$ random variable, which is $-1$ if we flipped and $1$ if we did not flip on edge $e$ on the path from $i$ to $j$. Let $Y = \prod_{e \in (i \to j)}Y_{e}$. Then, we have that $Y=-1$ if and only if $X_i \neq X_j$. Observe that $d_D(i, j) = \Pr_D[X_i \neq X_j] = \Pr[Y=-1]$. Then, we have that
    \begin{align*}
        \E[Y] &= \prod_{e \in (i \to j)}\E\left[Y_{e}\right] = -d_D(i, j)+(1-d_D(i, j)) = 1-2d_D(i, j).
    \end{align*}
    For each good edge $e$, since $\Pr[Y_{e}=-1]\le 0.5$, we have that
    \begin{align*}
        \E\left[Y_{e}\right] &\ge 0.
    \end{align*}
    Similarly, for each bad edge, we have that
    \begin{align*}
        \E\left[Y_{e}\right] &< 0.
    \end{align*}
    Since we are under the case that the number of bad edges from $i$ to $j$ is odd, we get that
    \begin{align*}
        \E[Y] &= \prod_{e \in (i \to j)}\E\left[Y_{e}\right] \le 0 \implies d_D(i, j) \ge 0.5.
    \end{align*}
\end{proof}
Thus, if we are able to embed $\dcap{0.5}$ into $\ell_1$, with an additional coordinate which simply indicates if the number of bad edges on the root-to-node path is odd, we can achieve our goal of embedding $d_D$ into $\ell_1$.
\begin{claim}
    If $E$ is a constant distortion embedding of $\dcap{0.5}$ into $\ell_1$, then $E'$, defined as the concatenation
    $$E'[i] = (E[i], \Ind[\#\text{bad edges from root $X_1$ to $X_i$ is odd}])$$ is a constant distortion embedding of $d_D(i, j).$
\end{claim}
\begin{proof}
    Given $D$, we obtain $D'$ where the edge-flip probabilities are at most $0.5$. If the number of bad edges on the path between $i$ and $j$ is even, then
    \begin{align*}
        &\Ind[\#\text{bad edges from root $X_1$ to $X_i$ is odd}] - \Ind[\#\text{bad edges from root $X_1$ to $X_j$ is odd}] = 0 \\
        \implies \qquad &\|E'[i]-E'[j]\|_1 = \|E[i]-E[j]\|_1.
    \end{align*}
    Furthermore, from part (1) in \Cref{claim:bad-edges-timwef}, we know that $d_{D}(i,j) = d_{D'}(i,j)$. Also, from \Cref{lem:timwef-no-bad-edges-to-capped-tree}, we know that $\dcap{0.5}(i, j)$ is constant distortion embedding of $d_{D'}(i, j)$. Thus, we are good if $E$ is a constant distortion embedding of $\dcap{0.5}$.

    If the number of bad edges on the path between $i$ and $j$ is odd, observe that
    \begin{align*}
        &\left|\Ind[\#\text{bad edges from root $X_1$ to $X_i$ is odd}] - \Ind[\#\text{bad edges from root $X_1$ to $X_j$ is odd}]\right| = 1 \\
        \implies \qquad &\|E'[i]-E'[j]\|_1 = \|E[i]-E[j]\|_1 + 1 \le \Theta(1) \cdot \dcap{0.5}(i,j) + 1 \le \Theta(1) \cdot 0.5 + 1 \le \Theta(1) \cdot 0.5 \\
        &\qquad\qquad\qquad \le \Theta(1) \cdot d_D(i,j),
    \end{align*}
    where the last inequality follows from part (2) in \Cref{claim:bad-edges-timwef}.
    Also, we have that
    \begin{align*}
        d_D(i, j) = \Pr_D[X_i \neq X_j] \le 1 \le 1 + \|E[i]-E[j]\|_1  = \|E'[i]-E'[j]\|_1,
    \end{align*}
    giving us that
    \begin{align*}
        d_D(i, j) \le \|E'[i]-E'[j]\|_1 \le \Theta(1) \cdot d_D(i, j).
    \end{align*}
    This completes the proof.
\end{proof}

Thus, we have argued that the crux of embedding the metric $d_D$ defined by the \timwef~into $\ell_1$ is obtaining a constant distortion embedding of the metric $\dcap{0.5}$ into $\ell_1$. %

\subsection{Fixed Cap Metrics}
\label{sec:fixed-cap}
We now turn our attention towards generally embedding fixed cap tree metrics into $\ell_1$. We begin with the special case of line graphs, and build up towards arbitrary tree metrics.
\subsubsection{Fixed Cap Line Metrics}
\label{sec:capped-line}
Recall that a line metric space $(L_n, d)$ simply corresponds to $n$ vertices (or rather locations on the real line) $x_1,\dots,x_n$, where each pair of consecutive vertices $x_i,x_{i+1}$ is connected by an edge of length $e_i \in \R_{>0}$, and $d(x_i, x_j)=\sum_{k=i}^{j-1}e_i$. For a strictly positive fixed cap $M \in \Z_{> 0}$, the corresponding fixed cap metric is
\begin{align}
    \dcap{M}(x_i, x_j) = \min(d(x_i, x_j), M).
\end{align}
The work of \cite{abraham2022metric} shows that $(L_n, \dcap{M})$ embeds into $\ell_1$ using $O(\log n)$ dimensions with constant distortion. We will present an alternative approach that also attains this results, both for sake of completeness, and because it lends itself better towards our later result of \Cref{thm:general-cap}. Let $loc[i]=d(x_i, x_1)$. Throughout what follows, we will identify every vertex $x_i$ with $loc[i]$ instead. The main technique involved to construct the embedding is a ``lazy snaking'' procedure, described in \Cref{algo:lazy-snake}.

\begin{algorithm}[t]
    \caption{Lazy Snaking for fixed cap line metric} \label{algo:lazy-snake}
    \hspace*{\algorithmicindent} 
    \begin{flushleft}
      {\bf Input:} List of $n$ node locations $loc[]$, cap $M$ \\
      {\bf Output:} List of embeddings for the nodes $embedding[]$ \\
    \end{flushleft}
    \begin{algorithmic}[1]
    \Procedure{LazySnake}{$loc, M$}:
    \State $t \gets 0$
    \While{$t \leq loc[n]$}
        \State $\Delta \sim \mathrm{Uniform}\{0,1\}$
        \If{$\Delta = 0$} \Comment{Rest for a duration $\frac{M}{4}$}
            \State $snake[t'] \gets 0$ for $t' \in \left[t, t+\frac{M}{4}\right]$
            \State $t \gets t + \frac{M}{4}$
        \EndIf
        \If{$\Delta = 1$} \Comment{Snake with a width $M$}
            \State $snake[t'] = t'-t$ for $t' \in \left[t, t+M\right]$ 
            \State $snake[t'] = 2M-(t'-t)$ for $t' \in \left[t+M, t+2M\right]$
            \State $t \gets t + 2M$
        \EndIf    
    \EndWhile
    \For{$i \gets 1$ to $n$}
        \State $embedding[i] = snake[loc[i]]$
    \EndFor
    \State {\bf return} $embedding[]$
    \EndProcedure
    \end{algorithmic}
\end{algorithm}

\begin{figure}[H]
    \centering
    \includegraphics[scale=0.42]{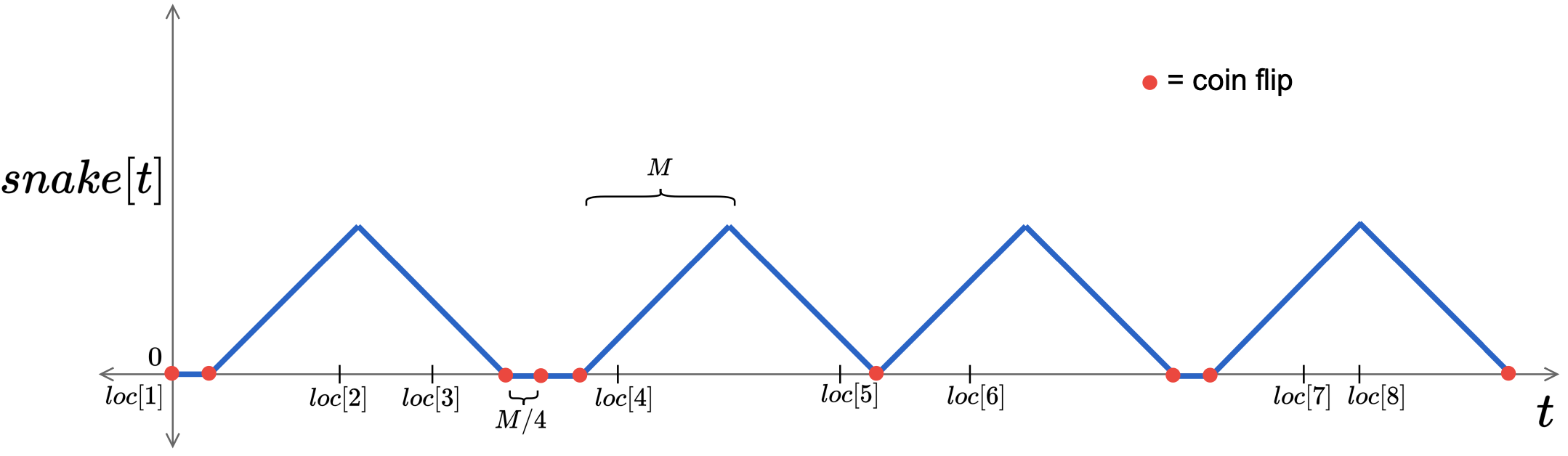}
    \caption{Lazy snaking with fixed cap.}
    \label{fig:lazy-snaking-fixed-cap}
\end{figure}

Similar techniques have previously appeared in the literature: for example, in \cite{abraham2022metric}, the ``sawtooth'' function is exactly a means of non-lazy snaking with a fixed period (i.e., $\Delta=1$ always in step 4 above), which similarly allows embedding fixed-cap metrics into $\ell_1$. Our version of \textit{lazy} snaking is more conducive towards deriving our general result (\Cref{thm:general-cap}) on truncated $\ell_1$ metrics.

We will require the following claim, which says that for every pair of nodes, lazy snaking does not overestimate their distance, and also covers at least a constant fraction of their distance in expectation. As a point of comparison, this claim is qualitatively similar to \cite[Lemma 2]{abraham2022metric}, but our proof is a different case analysis.

\begin{claim}
    \label{claim:lazy-snake-line-fixed-cap-lb}
    Let $embedding=\textsc{LazySnake}(loc, M)$. For every fixed pair of nodes $x_i$ and $x_j$,
    \begin{align*}
    &\text{(1)  } \Pr\left[|embedding[i]-embedding[j]| \le \dcap{M}(x_i, x_j)\right]=1. \\
    &\text{(2)  } \E\left[|embedding[i]-embedding[j]|\right]  \ge \Omega(1)\cdot \dcap{M}(x_i, x_j).
    \end{align*}
\end{claim}
\begin{proof}
    The first part of the claim is straightforward, since snaking can only reduce distance between two nodes. \\
    For the second part, let $i < j$ without loss of generality. Let $E_1$ be the event that $snake=0$ for some $t \in [\max(0, loc[i]-2M), loc[i]]$. Since the width that we snake over (if at all we are not at rest) is $M$, observe that $\Pr[E_1]=1$.
    We have three cases: \\
    \noindent Case 1: $1 \le |loc[i]-loc[j]|\le M.$ \\
    In this case, $\dcap{M}(x_i, x_j) = |loc[i]-loc[j]|$.  Let $E_2$ be the event that  $snake=0$ for some $t'=[loc[i]-M/4, loc[i]]$. Then, $\Pr[E_2|E_1]\ge \left(\frac{1}{2}\right)^8 = \Omega(1)$. This is because, the time to reach $loc[i]$ from $t$ is at most $2M$, and in this duration, we flip the ``rest" coin at most $\frac{2M}{M/4}=8=O(1)$ times to arrive at such a $t'$. Now, conditioned on $E_2$, we have that with probability $1/2$, we choose to move and snake for a width $M$, and if that happens, we have that that $0 \le embedding[i] \le M/4$. Then, since $|loc[i]-loc[j]|\le M$, we will have $embedding[j] - embedding[i] \ge \frac{1}{2} (loc[j]-loc[i])$. This is because, in the worst case, $embedding[i]=M/4$, and $loc[j]=loc[i]+M$, in which case $embedding[j]=3M/4$ after snaking around once. Since we only had to condition on a constant number of coin flips to obtain this,
    \begin{align*}
        \E\left[|embedding[i]-embedding[j]|\right] &\ge \Pr[E_1,E_2]\cdot \E\left[|embedding[i]-embedding[j]| ~|~E_1,E_2\right] \\
        &\ge \Omega(1) \cdot |loc[i]-loc[j]|.
    \end{align*}
    \noindent Case 2: $M < |loc[i]-loc[j]|\le 4M.$ \\
    In this case, $\dcap{M}(x_i, x_j) = M$. Let $E_2$ be the event that $embedding[i]=0$ and $snake=0$ for some $t' \in [loc[j]-M, loc[j]-3M/4]$. Then, we have that $\Pr[E_2|E_1] \ge \Omega(1)$. This is because, to arrive at such a $t'$ starting from $t$ while also ensuring that $embedding[i]=0$, we may flip at most $\frac{2M+3M}{M/4}=O(1)$ ``rest" coins from time $t$ --- this happens with probability at least $\Omega(1)$. Thereafter, conditioned on $E_2$, the coin flips ``snake" with probability at least $1/2$, and this ensures that $embedding[j] \ge 3M/4$, which in turn gives $|embedding[i]-embedding[j]| \ge 3M/4$. Again, since we only had conditioned on $O(1)$, we get that
    \begin{align*}
        \E\left[|embedding[i]-embedding[j]|\right] &\ge \Pr[E_1,E_2]\cdot \E\left[|embedding[i]-embedding[j]| ~|~E_1,E_2\right] \\
        &\ge \Omega(1)\cdot M.
    \end{align*}
    \newline\noindent Case 3: $|loc[i]-loc[j]| > 4M.$ \\
    Here as well, $\dcap{M}(x_i, x_j) = M$. First, just as in Case 1 above, let $E_2$ be the event that  $snake=0$ for some $t'=[loc[i]-M/4, loc[i]]$, for which we saw that $\Pr[E_2|E_1]\ge \Omega(1)$. Note that $E_1,E_2$ together ensure that $embedding[i] \le M/4$. Now, let $E_3$ be the event that $snake=0$ for some $t'' \in [loc[j]-4M, loc[j]-2M]$. Since $|loc[i]-loc[j]| > 4M$, there exists such a $t'' > t'$ with certainty, for which we have not yet assumed any conditioning, i.e., $\Pr[E_3|E_2,E_1]=1$. Now, following $E_3$, with $\frac{3M}{M/4} = O(1)$ ``rest" flips (happens with probability at least $\Omega(1)$), we ensure that $snake=0$ up until some $t''' \in [loc[j]-M, loc[j]-3M/4]$. Thereafter, we may flip a ``snake", ensuring that $embedding[j] \ge 3M/4$. Putting together, we get
    \begin{align*}
        \E\left[|embedding[i]-embedding[j]|\right] &\ge \Pr[E_1,E_2, E_3]\cdot \E\left[|embedding[i]-embedding[j]| ~|~E_1,E_2, E_3\right] \\
        &\ge \Omega(1)\cdot M.
    \end{align*}
\end{proof}

Using $O(\log n)$ dimensions, we can then boost the in-expectation guarantee of \Cref{claim:lazy-snake-line-fixed-cap-lb} to obtain the following theorem.
\begin{theorem}[Fixed cap line into $\ell_1$]
    \label{thm:capped-line}
    $(L_n, \dcap{M})$ can be embedded into $(\R^d, \ell_1)$ where $d=O(\log n)$ with $\Theta(1)$ distortion.
\end{theorem}
\begin{proof}
Let us invoke \textsc{LazySnake}$(loc, M)$ $N$ times to obtain $N$ independent copies $embedding_1,\dots,embedding_N$ for all the nodes. Let $X^{(l)}_{ij} = |embedding_l[i]-embedding_l[j]|$. Then, observe that $0 \le X^{(l)}_{ij} \le \dcap{M}(x_i, x_j)$. Furthermore, \Cref{claim:lazy-snake-line-fixed-cap-lb} gives that $\E\left[\sum_{k=1}^NX^{(l)}_{ij}\right] \ge cN\dcap{M}(x_i,x_j)$ for some constant $c > 0$. Hoeffding's inequality then gives us
\begin{align*}
    \Pr\left[\sum_{l=1}^N X^{(l)}_{ij} \le \frac{cN\dcap{M}(x_i,x_j)}{2}\right] \le 2\exp\left( -\Theta(N)\right),
\end{align*}
or rather, choosing $N = C\log n$ for an appropriately large constant $C$ gives
\begin{align*}
    \Pr\left[\frac{1}{N}\sum_{l=1}^N X^{(l)}_{ij} \le \frac{c}{2}\cdot\dcap{M}(x_i,x_j)\right] \le \frac{2}{n^3}.
\end{align*}
Furthermore, since the embedding never overestimates distances, observe that 
\begin{align*}
    \Pr\left[\frac{1}{N}\sum_{l=1}^N X^{(l)}_{ij} > \dcap{M}(x_i,x_j)\right] =0.
\end{align*}
We can therefore rescale each $embedding_k$ as $embedding_k \gets embedding_k/N$, and union bound over all the $< n^2$ pairs $i, j$ to obtain that
\begin{align*}
    \Pr\left[\forall i,j:~ \Omega(1)\cdot\dcap{M}(x_i,x_j) \le \sum_{l=1}^N|embedding_l[i]-embedding_l[j]| \le \dcap{M}(x_i, x_j)\right] > 0,
\end{align*}
which in particular, implies that there exists such an embedding that obtains $\Theta(1)$ distortion with $N=O(\log n)$ dimensions.
\end{proof}

\subsubsection{Fixed Cap Tree Metrics}
Consider now a tree $T_n$ on $n$ vertices $x_1,\dots,x_n$, for which recall that the distance is given by $d(x_i, x_j)=\sum_{e \in x_i \to x_j}e$. Again, we consider the (fixed) cap version of the tree metric distance, given by
\begin{align}
    \dcap{M}(x_i, x_j) = \min(d(x_i, x_j), M).
\end{align}
We will show that $(T_n, \dcap{M})$ also embeds into $\ell_1$ using $O(\log^2 n)$ dimensions and constant distortion. 

First, let us obtain the caterpillar decomposition (stated in \Cref{sec:caterpillar-prelims}) of $T_n$. Let $C_1,C_2,\dots,C_m$ denote the caterpillars in its caterpillar decomposition, where each $C_i=\{e_{i_1},\dots,e_{i_l}\}$ is a simple path in the tree. The length of $C_i$ is the sum of the weights of the edges that it is made up of. We will first quickly see how the caterpillar decomposition yields an isometric embedding of the tree metric into $\ell_1$, while possibly using a lot many dimensions. We will have a coordinate in the embedding dedicated for every caterpillar in the decomposition -- thus, the number of dimensions can be as large as the number of edges in the tree, i.e., $\Omega(n)$. For every node $x_i$, we walk up from $x_i$ to the root, and for every caterpillar that we touch (at most $\log n$ of them), we record the length traversed on that caterpillar at the coordinate dedicated for it. Thus, we will end up with a $\log n$-sparse embedding for every node. We can see that this is an isometric embedding, because for any pair of nodes, the lengths on the caterpillars from their least common ancestor up to root get canceled out when we take the difference of their embeddings, and only the lengths on caterpillars on the path joining them survive.

Now, we describe a slight modification to the above embedding, which possibly increases the dimensions even more, but still ensures that the embedding is isometric, and will be useful for constructing our final low-dimensional embedding for the capped tree metric with cap $M$. For any caterpillar $C_i$ whose length is at most $M$, we do nothing. For caterpillars having length larger than $M$, we split it up into snippets of size exactly $M/\log n$, and possibly one left-over snippet of size smaller than $M/\log n$. Having snipped each caterpillar like so, let $S_1,\dots,S_s$ be all the resulting snippets --- we will dedicate a single coordinate in the embedding to each $S_i$. As before, for every node, we walk up from the node to the root, and record the length of each snippet we traverse at the corresponding coordinate. Let this embedding be denoted by $\snipcat$ (short for Snipped Caterpillar) i.e., $\snipcat(x_i)$ is a vector of size $s$ for each node $x_i$. This embedding is still isometric, however, the number of snippets $s$ can now be huge, if a large number of the caterpillars had length $\gg M$.

We will now make a couple observations on the structure of the vector $\snipcat(x_i)-\snipcat(x_j)$ for any pair of nodes $x_i$ and $x_j$.

\begin{claim}
    \label{claim:sparsity-less-than-M}
    If $d(x_i, x_j) \le M$, then $\snipcat(x_i)-\snipcat(x_j)$ has at most $6\log n$ non-zero coordinates, and each of these coordinates is at most $M/\log n$ in magnitude.
\end{claim}
\begin{proof}
    Let $z=\snipcat(x_i)-\snipcat(x_j)$. The only non-zero coordinates in $z$ correspond to the caterpillar snippets on the path joining $x_i$ and $x_j$. First, let us count the number of full $M/\log n$ sized snippets on this path --- since $\|z\|_1=d(x_i, x_j)\le M$, we can have at most $\log n$ many such snippets. Now, let us count the number of non-zero coordinates contributed by snippets of length $< M/\log n$. These could either be due to caterpillars that were smaller than $M/\log n$, and never got snipped --- there could be at most $2\log n$ many of these, since the (unsnipped) caterpillar embeddings of both $x_i$ and $x_j$ are $\log n$ sparse. Or, these could be due to partial snippets on the path joining $x_i$ and $x_j$. Observe that we can have at most one partial snippet per whole caterpillar on this path (except possibly an extra at the least common ancestor), and thus, we can again only have $2\log n +1 \le 3\log n$ many of these in total. Together, we get $\log n + 2 \log n + 3\log n=6\log n$ many non-zero coordinates, as claimed. Finally, by construction, each coordinate in both $\snipcat(x_i)$ and $\snipcat(x_j)$ is at most $M/\log n$, and hence the difference can be at most $M/\log n$ in magnitude.
\end{proof}

\begin{claim}
    \label{claim:sparsity-greater-than-M}
    If $d(x_i, x_j) > M$, then $\snipcat(x_i)-\snipcat(x_j)$ contains a set of $\frac{3}{4}\log n$ non-zero coordinates, each of which is at least $\frac{M}{20\log n}$ and at most $\frac{M}{\log n}$ in magnitude, and the sum of these coordinates is at least $3M/4$.
\end{claim}
\begin{proof}
    Again, let $z=\snipcat(x_i)-\snipcat(x_j)$. First, let us count the number of snippets on the path joining $x_i$ and $x_j$ that are less than $\frac{M}{20\log n}$ in absolute value. Either these arise due to unsnipped-caterpillars of length smaller than $M/\log n$ in the embeddings of $x_i$ and $x_j$ --- there can be at most $2\log n$ many of these; or these arise due to partial snippets on the path joining them --- there are again at most $2\log n+1\le3\log n$ many of these, as argued above. Thus, there can only be at most $5\log n$ many coordinates having magnitude smaller than $\frac{M}{20 \log n}$ in $z$. These coordinates account for length at most $M/4$ out of $\|z\|_1$. All the other coordinates necessarily have magnitude at least $\frac{M}{20 \log n}$, but also at most $M/\log n$ by construction. Since $\|z\|=d(x_i, x_j)>M$, these other coordinates have to account for at least $3M/4$ of the length. Thus, there needs to be a set of at least $\frac{3M}{4} \cdot \frac{\log n}{M}=\frac{3}{4}\log n$ many coordinates, each of which is at least $\frac{M}{20 \log n}$, and whose sum is at least $3M/4$ as required.
\end{proof}

Now, for $k=6\log n$, consider choosing a uniformly random hash function $h:[s]\to [k]$, where recall that $s$ is the total number of caterpillar snippets. For each node $x_i$, let $\hsnipcat(x_i)$ (short for Hashed Snipped Caterpillar) be a vector of size $6\log n$, defined in the following. This will be the building block of our final $O(\log^2 n)$ sized embedding.
\begin{align}
    \hsnipcat(x_i)[p]=\sum_{q \in [s]:h(q)=p}\snipcat(x_i)[q] \qquad \text{for } p \in \{1,\dots,k\}.
\end{align}
Now, we interpret each coordinate of $\hsnipcat$ as defining a line metric, over which we will do lazy snaking. Concretely, for each $p \in [k]$, let $loc_p[i] = \hsnipcat(x_i)[p]$ --- note that for the root $x_1$, $loc_p[1]=0$ for all $p$, since $\snipcat(x_i)=0^s$. Let $sloc_p=\textsc{LazySnake}(loc_p, M/\log n)$, and consider the representation of each node $x_i$ as 
\begin{equation}
    \label{eqn:fixed-cap-tree-embedding-def}
    embedding[i] = (sloc_1[i], sloc_2[i], \dots, sloc_k[i]).
\end{equation}    
First, we make the following simple claim:
\begin{claim}[No overestimation]
    \label{claim:fixed-cap-tree-no-overestimate}
    Fix any pair of nodes $x_i$ and $x_j$. Then,
     with probability 1,
     \begin{align*}
        \|embedding[i]-embedding[j]\|_1 \le 6\dcap{M}(x_i, x_j).
     \end{align*}
\end{claim}
\begin{proof}
    We have two cases: \\
    \noindent Case 1: $d(x_i, x_j) \le M.$ \\
    In this case, $\dcap{M}(x_i, x_j) = d(x_i, x_j)$. Observe that
    \begin{align*}
        \|embedding[i]-embedding[j]\|_1 &= \sum_{p=1}^k\left|sloc_p[i]-sloc_p[j]\right| \\
        &\le \sum_{p=1}^k\left|loc_p[i]-loc_p[j]\right| \quad \text{(snaking never overestimates distances)}\\
        &= \sum_{p=1}^k\left|\sum_{q \in [s]:h(q)=p}(\snipcat(x_i)[q] - \snipcat(x_j)[q])\right| \\
        &\le \sum_{p=1}^k\sum_{q \in [s]:h(q)=p}\left|\snipcat(x_i)[q] - \snipcat(x_j)[q]\right| \\
        &= d(x_i, x_j).
    \end{align*}
    \noindent Case 2: $d(x_i, x_j) > M.$ \\
    In this case, $\dcap{M}(x_i, x_j) = M$. We have
    \begin{align*}
        \|embedding[i]-embedding[j]\|_1 &= \sum_{p=1}^k\left|sloc_p[i]-sloc_p[j]\right| \\
        &\le \sum_{p=1}^k \frac{M}{\log n} \qquad \left(\text{snaking width is }\frac{M}{\log n}\right)\\
        &\le 6M.
    \end{align*}
\end{proof}

\Cref{claim:fixed-cap-tree-no-overestimate} ensures that the embedding never overestimates distances beyond a constant factor. However, we also have the following nice property, which ensures that in expectation, the embedding captures at least a constant fraction of the distance between any fixed pair of nodes.
\begin{lemma}[No underestimation]
    \label{lemma:fixed-cap-tree-no-underestimate}
    Fix any pair of nodes $x_i$ and $x_j$, and fix $p \in [k]$. Then, we have that
    \begin{align*}
        \E\left[\left|sloc_p[i]-sloc_p[j]\right|\right] &= \E_{h}\E_{\snake}\left[\left|sloc_p[i]-sloc_p[j]\right|\right] \ge \frac{\Omega(1)}{k}\cdot \dcap{M}(x_i, x_j).
    \end{align*}
    Thus, by linearity of expectation,
    \begin{align*}
        \E\left[\|embedding[i]-embedding[j]\|_1\right] = \E\left[\sum_{p=1}^k\left|sloc_p[i]-sloc_p[j]\right|\right] \ge \Omega(1)\cdot \dcap{M}(x_i, x_j).
    \end{align*}
\end{lemma}
\begin{proof}
    We have two cases: \\
    \noindent Case 1: $d(x_i, x_j) \le M.$ \\
    In this case, $\dcap{M}(x_i, x_j) = d(x_i, x_j)$. Recall from \Cref{claim:sparsity-less-than-M} that $\snipcat(x_i)-\snipcat(x_j)$ has at most $6 \log n$ non-zero coordinates and each of these coordinates is at most $\frac{M}{\log n}$ in magnitude. Let us possibly include some zero coordinates, so that we have exactly $k=6\log n$ of these ``special'' coordinates $c_1,\dots, c_r, \dots,c_k$. Observe that $\sum_{r=1}^k |\snipcat(x_i)[c_r]-\snipcat(x_j)[c_r]|=d(x_i, x_j)$, because the snipped-caterpillar embedding is isometric. Let us condition on the realization of the hash function $h$, which is independent of the randomness in the snaking. Conditioned on this realization, for $loc_p[i]=\hsnipcat(x_i)[p]$ and $sloc_p=\textsc{LazySnake}(loc_p, M/\log n)$ we have from \Cref{claim:lazy-snake-line-fixed-cap-lb} that
    \begin{align*}
        \E_{\snake}\left[|sloc_p[i]-sloc_p[j]|\right] \ge \Omega(1) \cdot \dcap{M/\log n}(loc_p[i], loc_p[j]).
    \end{align*}
    Now, taking an expectation with respect to the choice of the hash function, we get
    \begin{align*}
        \E_h\E_{\snake}\left[|sloc_p[i]-sloc_p[j]|\right] \ge \Omega(1) \cdot \E_h\left[\dcap{M/\log n}(loc_p[i], loc_p[j])\right].
    \end{align*}
    Let $A$ be the event that only one of the $k$ special coordinates hashes to $p$. Then, we have that $\Pr[A]=\left(1-\frac1k\right)^{k-1} \ge \frac{1}{e}=\Omega(1)$, yielding
    \begin{align*}
        \E_h\left[\dcap{M/\log n}(loc_p[i], loc_p[j])\right] &\ge \Omega(1)\cdot\E_h\left[\dcap{M/\log n}(loc_p[i], loc_p[j]) ~|~ A\right]
    \end{align*}
    Furthermore, conditioned on $A$, we have
    \begin{align*}
        &\E_h\left[\dcap{M/\log n}(loc_p[i], loc_p[j]) ~|~ A\right] \\
        &\qquad=\sum_{r=1}^k \frac1k \cdot \E_h\left[\dcap{M/\log n}(loc_p[i], loc_p[j]) ~|~ \text{only $r^{\text{th}}$ special coordinate hashes to $p$}\right].
    \end{align*}
    But now, observe that if only the $r^{\text{th}}$ special coordinate hashes to $p$, we have $\dcap{M/\log n}(loc_p[i], loc_p[j]) =|\snipcat(x_i)[c_r]-\snipcat(x_j)[c_r]|$. This is because $|\snipcat(x_i)[c_r]-\snipcat(x_j)[c_r]|$ is at most $M/ \log n$, which is smaller than the cap. Finally, recalling that $\sum_{r=1}^k |\snipcat(x_i)[c_r]-\snipcat(x_j)[c_r]|=d(x_i, x_j)$, we get
    \begin{align*}
        \E_h\left[\dcap{M/\log n}(loc_p[i], loc_p[j]) ~|~ A\right] &= \sum_{r=1}^k\frac1k\left|\snipcat(x_i)[c_r]-\snipcat(x_j)[c_r]\right| =\frac{1}{k}\cdot d(x_i, x_j). 
    \end{align*}
    Putting everything together, we get
    \begin{align*}
        \E_h\E_{\snake}\left[|sloc_p[i]-sloc_p[j]|\right] &\ge \frac{\Omega(1)}{k}\cdot d(x_i, x_j) = \Omega(1) \cdot \frac{\dcap{M}(x_i, x_j)}{k}.
    \end{align*}
    \noindent Case 2: $d(x_i, x_j) > M.$ \\
    In this case, $\dcap{M}(x_i, x_j) = M$. Recall from \Cref{claim:sparsity-greater-than-M} that $\snipcat(x_i)-\snipcat(x_j)$ has a set of $k'=\frac34\log n$ non-zero ``special" coordinates $c_1,\dots,c_r,\dots,c_{k'}$, each of which is at least $\frac{M}{20\log n}$ and at most $\frac{M}{\log n}$ in magnitude. Let $A$ be the event that the non-special coordinates that get hashed to $p$ amount for a distance of at least $\frac{M}{100\log n}$. Concretely, under $A$, $\left|\sum_{q \in [s]:q \text{ not special},h(q)=p}(SC(x_i)[q]-SC(x_j)[q])\right| \ge \frac{M}{100\log n}$. 
    We have that
    \begin{align*}
        \E\left[|sloc_p[i]-sloc_p[j]|\right] &= \Pr[A]\cdot \E\left[|sloc_p[i]-sloc_p[j]| ~|~ A\right] \\
        &\quad + \Pr[\neg A]\cdot \E\left[|sloc_p[i]-sloc_p[j]| ~|~ \neg A\right]
    \end{align*}
    Conditioned on $A$, we are happy if none of the special coordinates hash to $p$, which happens with probability $\left(1-\frac{1}{k}\right)^{k'} \ge \exp\left(-\frac{k'}{\sqrt{k(k-1)}}\right)=\exp\left(-\frac18\sqrt{\frac{6\log n}{6\log n-1}}\right) \ge 1/e = \Omega(1)$ for $n \ge 2$. In this case, lazy snaking will capture at least a constant fraction of $\frac{M}{100 \log n}$, yielding
    \begin{align*}
        \E\left[|sloc_p[i]-sloc_p[j]| ~|~ A\right] &\ge \Omega(1) \cdot \frac{M}{100 \log n}.
    \end{align*}
    If $A$ does not occur, we have that $\left|\sum_{q \in [s]:q \text{ not special},h(q)=p}(SC(x_i)[q]-SC(x_j)[q])\right| < \frac{M}{100\log n}$. In this case, we are happy if exactly one of the special coordinates hashes to $p$, which happens with probability $\left(1-\frac{1}{k}\right)^{k'-1} \ge \left(1-\frac{1}{k}\right)^{k'} \ge 1/e = \Omega(1)$. This will ensure a distance of at least $\frac{M}{20\log n} - \frac{M}{100\log n}=\frac{M}{25\log n}$, of which lazy snaking will capture at least a constant fraction in expectation, yielding
    \begin{align*}
        \E\left[|sloc_p[i]-sloc_p[j]| ~|~ \neg A\right] &\ge \Omega(1) \cdot \frac{M}{25 \log n}.
    \end{align*}
    In total, we get that
    \begin{align*}
        \E\left[|sloc_p[i]-sloc_p[j]|\right] &\ge \Pr[A]\cdot \Omega(1) \cdot \frac{M}{100 \log n} 
        + \Pr[\neg A]\cdot \Omega(1)\cdot \frac{M}{25 \log n} \\
        &\ge \Omega(1) \cdot \frac{M}{100 \log n} \left(\Pr[A] + \Pr[\neg A] \right) \\
        &= \Omega(1)\cdot \frac{M}{100 \log n} = \Omega(1) \cdot \frac{M}{k} = \Omega(1)\cdot \frac{\dcap{M}(x_i, x_j)}{k}.
    \end{align*}
\end{proof}

As in the case of the capped line metric, using $O(\log n)$ independent repetitions of the above, we obtain the following theorem.

\theoremfixedcaptree*
\begin{proof}
    Let us obtain $N$ independent copies $embedding_1,\dots,embedding_N$ of the embedding in \Cref{eqn:fixed-cap-tree-embedding-def}. Note that each of these is an embedding of size $6\log n$. Let $X^{(l)}_{ij} = \|embedding_l[i]-embedding_l[j]\|_1$. Then, observe that from \Cref{claim:fixed-cap-tree-no-overestimate}, $0 \le X^{(l)}_{ij} \le 6\dcap{M}(x_i, x_j)$. Furthermore, \Cref{lemma:fixed-cap-tree-no-underestimate} gives that $\E\left[\sum_{l=1}^NX^{(l)}_{ij}\right] \ge cN\dcap{M}(x_i,x_j)$ for some absolute constant $c > 0$. Hoeffding's inequality then gives us
\begin{align*}
    \Pr\left[\sum_{l=1}^N X^{(l)}_{ij} \le \frac{cN\dcap{M}(x_i,x_j)}{2}\right] \le 2\exp\left(-\Theta(N)\right),
\end{align*}
or rather, choosing $N = C\log n$ for an appropriately large constant $C$ gives
\begin{align*}
    \Pr\left[\frac{1}{N}\sum_{l=1}^N X^{(l)}_{ij} \le \frac{c}{2}\cdot\dcap{M}(x_i,x_j)\right] \le \frac{2}{n^3}.
\end{align*}
Furthermore, since the embedding never overestimates distances (\Cref{claim:fixed-cap-tree-no-overestimate}), observe that 
\begin{align*}
    \Pr\left[\frac{1}{N}\sum_{l=1}^N X^{(l)}_{ij} > 6\dcap{M}(x_i,x_j)\right] =0.
\end{align*}
We can therefore rescale each coordinate in $embedding_k$ as $embedding_k \gets embedding_k/N$, and union bound over all the $< n^2$ pairs $i, j$ to obtain that
\begin{align*}
    \Pr\left[\forall i,j:~ \Omega(1)\cdot\dcap{M}(x_i,x_j) \le \sum_{l=1}^N\|embedding_l[i]-embedding_l[j]\|_1 \le O(1)\cdot\dcap{M}(x_i, x_j)\right] > 0,
\end{align*}
which in particular, implies that there exists such an embedding that obtains $\Theta(1)$ distortion with $N=O(\log^2n)$ dimensions.
\end{proof}

\section{General Tree Ising Models}
\label{sec:general-tim}
We now turn our attention to the more general and challenging case of tree Ising models on $n$ random variables $X_1,\dots,X_n$ where each variable can additionally have a nonzero external field (\Cref{eqn:tim-def}). Recall that these tree Ising models can be uniquely described by specifying marginal probabilities for every $X_i$, and joint probabilities for every pair $(X_i, X_j)$ that corresponds to an edge in the tree.

To exhaustively capture the probability of disagreement between two variables $X_i$ and $X_j$, we end up requiring to crucially capture contribution towards this by the ``independence" between the variables. In the typical notion of independence, this could be thought of in terms of being distance that grows with $H(X_i|X_j)$ or $H(X_j|X_i)$, where $H$ is the standard Shannon entropy. However, conditional Shannon entropy $H(X_i|X_j)$ seems at first glance an unwieldy quantity from the perspective of embedding into $\ell_1$ (for reasons like, e.g., the behavior of $H(X_i)$ as $X_i \rightarrow 0$). We define an alternative seemingly simpler quantity, which we call ``Bernoulli Randomness"\footnote{We expect this quantity (or equivalent versions of it) to have previously been defined and used in the literature.}, that essentially captures what we are hoping for.

\subsection{Bernoulli Randomness}\label{sub:bern}

\begin{definition}[Bernoulli Randomness]
    Any binary random variable $X$ can be written as the mixture of an unbiased random variable and a deterministic (i.e. single state of support) variable. Let $\Br(X)$ be the weight of the mixture on the unbiased variable. We can see that $\Br(X) = 2 \cdot \min(\Pr[X = 0],\Pr[X = 1])$. 
\end{definition}

A philosophical interpretation of Bernoulli randomness is that it corresponds to the probability that $X$ is determined by a purely random coin flip, instead of being deterministically chosen. Like Shannon entropy, Bernoulli randomness increases as $X$ becomes less biased.

We similarly define conditional Bernoulli randomness:

\begin{definition}[Conditional Bernoulli Randomness]
    $\Br(X|Y) = \sum_{y \in \{0,1\}} \Pr(Y = y) \cdot \Br(X | Y=y) = \sum_{y} 2 \cdot \min(\Pr(X=0,Y=y),\Pr(X=1,Y=y))$
\end{definition}

\begin{figure}[H]
    \centering
    \includegraphics[scale=0.35]{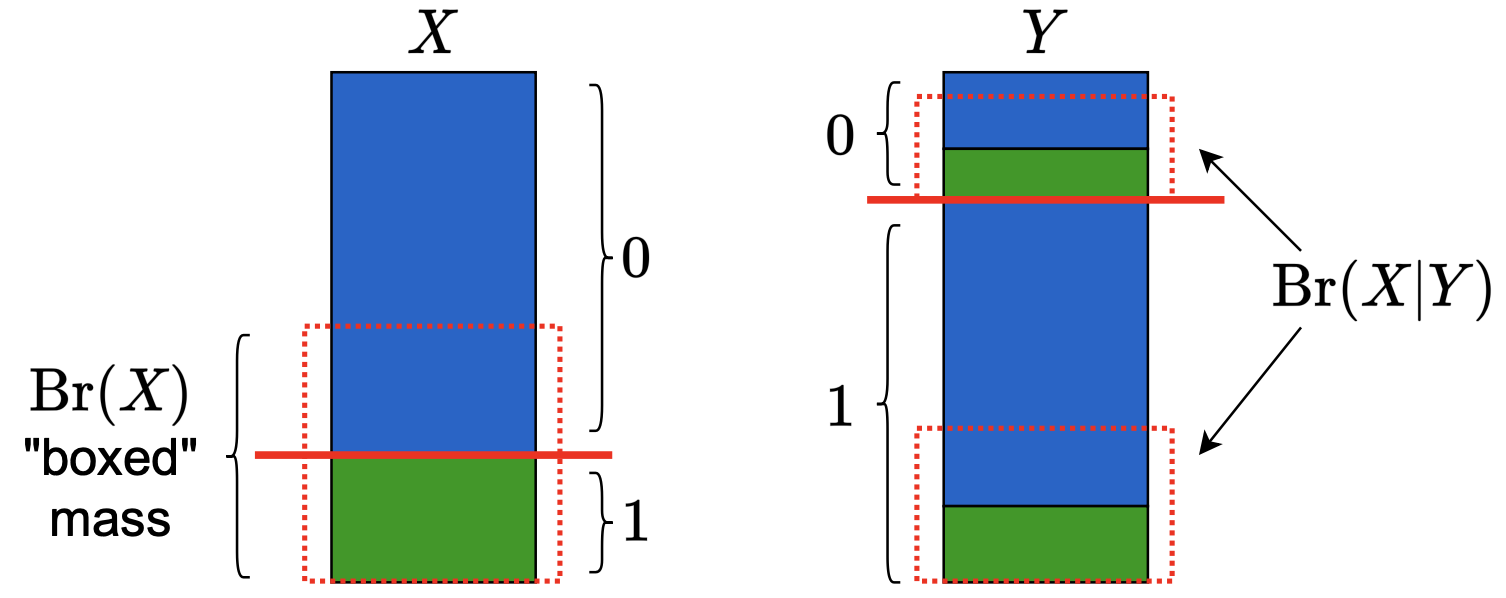}
    \caption{Bernoulli randomness and conditional Bernoulli randomness as ``boxed'' mass.}
    \label{fig:br-x-y}
\end{figure}
It helps a lot to keep \Cref{fig:br-x-y} in mind when thinking of Bernoulli randomness. The probability mass of $X$ has been divided by the red solid line into blue and green states corresponding to $\Pr[X=0]$ and $\Pr[X=1]$ respectively. Here, $\Br(X)$ is exactly the mass ``boxed'' by the red dashed line---there is an equal amount of green and blue mass in the box. Similarly, on the right, the blue and green mass of $X$ has been distributed in the states of $Y=0$ and $Y=1$ according to the joint distribution of $X$ and $Y$. The total boxed mass across both the states is precisely $\Br(X|Y)$.

We can lower bound the probability of disagreement between $X$ and $Y$ in terms of conditional Bernoulli randomness in a straightforward manner:
\begin{lemma}[$\Br$ lower-bounds distance] \label{lemma:bern-lb-dist}
    $\Pr[X \ne Y] \ge \frac{1}{2} \cdot \Br(Y|X), \frac{1}{2} \cdot \Br(X|Y)$.
\end{lemma}
\begin{proof}
    This is true because, for example, because if $Y$ is chosen via an unbiased coin from $X$ with probability $\Br(Y|X)$, then for each unbiased coin flip there is at least $1/2$ chance $X \ne Y$ and thus $\Pr[X \ne Y] \ge \frac{1}{2} \cdot \Br(Y | X)$. More formally,
    \begin{align*}
        \Pr[X \ne Y] &= \Pr[X=0, Y=1] + \Pr[X=1, Y=0] \\
        &\ge \min(\Pr[X=0, Y=0], \Pr[X=0, Y=1]) + \min(\Pr[X=1, Y=0], \Pr[X=1, Y=0]) \\
        &= \frac12 \cdot \Br(Y|X).
    \end{align*}
    The argument for $\Br(X|Y)$ is identical.
\end{proof}

The following result also nicely matches the intuition from Shannon entropy:
\begin{claim}[$\Br$ maximized at independence]
    Given two variables with fixed marginals $X,Y$, when optimizing over their joints it holds that $\Br(X|Y)$ and $\Br(Y|X)$ are both maximized when $X,Y$ are independent.
\end{claim}
\begin{proof}
    We state the proof for $\Br(X|Y)$---the proof for $\Br(Y|X)$ is identical.
    \begin{align*}
        \Br(X|Y) &= 2(\min(\Pr[X=0,Y=0],\Pr[X=1,Y=0]) + \min(\Pr[X=0,Y=1],\Pr[X=1,Y=1])) \\
        &\le 2(\Pr[X=0,Y=0]+ \Pr[X=0,Y=1]) = 2\Pr[X=0],
    \end{align*}
    and similarly,
    \begin{align*}
        \Br(X|Y) &= 2(\min(\Pr[X=0,Y=0],\Pr[X=1,Y=0]) + \min(\Pr[X=0,Y=1],\Pr[X=1,Y=1])) \\
        &\le 2(\Pr[X=1,Y=0]+ \Pr[X=1,Y=1]) = 2\Pr[X=1].
    \end{align*}
    This means $\Br(X|Y)\le 2\min(\Pr[X=0],\Pr[X=1])=\Br(X)$. Note that equality is attained when $X$ and $Y$ are independent.
\end{proof}

The niceness of the aforementioned claim is that the same statement is classically known to hold when $\Br$ terms are replaced with $H$ (Shannon entropy)  terms. A key difference, however, is that while the independent joint distribution is the \emph{unique} maximizer for Shannon entropy, there are often many maximizers for Bernoulli randomness. For example, consider random variables $X,Y$ with $p_{00}=0.1, p_{01}=0.2, p_{10}=0.2, p_{11}=0.5$. Here, $\Br(X|Y)=\Br(Y|X)=\Br(X)=\Br(Y)=0.6$, but $X$ and $Y$ are not independent. This indicates that although the notions share similar intuitions, there are key differences.

In our proof for tree Ising models, we will be studying the accumulation of conditional Bernoulli randomness. A simple lemma that may help see the niceness of this property (analogous to the data processing inequality for mutual information) is as follows:

\begin{lemma}[Data Processing Inequality]
    \label{lemma:data-processing-inequality}
    Consider a Bayesian network on a line specified by $X_1\to \dots\to X_n$. For any $i < j < k$, it must hold that $\Br(X_i | X_k) \ge \Br(X_i | X_j)$.
\end{lemma}
\begin{proof}
    \begin{figure}[H]
        \centering
        \includegraphics[scale=0.35]{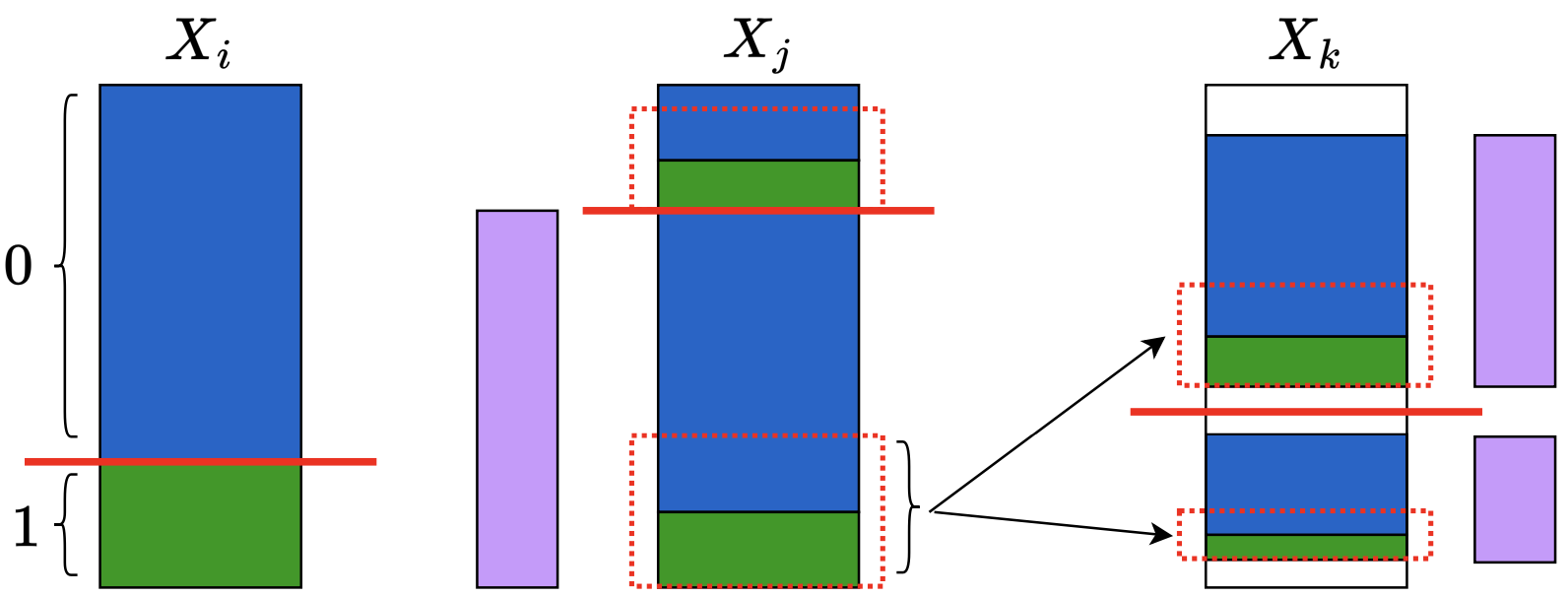}
        \caption{The ratio of blue and green mass within $X_j=0$ and $X_j=1$ is preserved when it gets distributed between $X_k=0$ and $X_k=1$. The picture above only shows this for the (purple) mass in state $X_j=1$. However, it is also true for the mass within state $X_j=0$ that would occupy the leftover white regions in $X_k$.}
        \label{fig:br-xi-xj-xk}
    \end{figure}
    The proof is aided by the pictoral definition of conditional Bernoulli randomness in terms of ``boxed'' mass that we stated above. In \Cref{fig:br-xi-xj-xk}, blue mass corresponds to $X_i=0$ and green mass corresponds to $X_i=1$. The total boxed mass in $X_j$ corresponds to $\Br(X_i|X_j)$. Observe that because of the Markov property, $\Pr[X_i|X_j, X_k]=\Pr[X_i|X_j]$. Namely, the mass in the bucket $X_j=1$ must get distributed between the buckets $X_k=0$ and $X_k=1$ in a such a way that the blue-to-green ratio in both these buckets is equal to the ratio in $X_j=1$. Put another way, some fraction of the total mass in $X_j=1$ goes to the bucket $X_k=0$ and the remaining fraction goes to the bucket $X_k=1$. The same also happens to the mass in $X_j=0$. But observe that this preserves all the boxing of mass in $X_j$. Namely, all the boxed mass in $X_j$ remains boxed in $X_k$, and we can only gain additional boxed mass because of redistribution. Thus, the boxed mass in $X_k$ is at least as much as the boxed mass in $X_j$, which means that $\Br(X_i|X_k) \ge \Br(X_i|X_j)$.
\end{proof}

\subsection{General Tree Ising Models Reduce to Lipschitz Cap Tree Metrics}\label{sub:recution}
In this section, we will transform the task of embedding general tree Ising models into $\ell_1$ to the task of embedding Lipschitz cap tree metrics into $\ell_1$. In what follows, we will reason about $\Pr[X_i \neq X_j]$ for nodes $X_i$ and $X_j$ in the tree Ising model. For this, we will imagine the tree to be rooted at $X_i$, with all edges pointing away from the root. This ensures that all the edges on the path from $X_i$ to $X_j$ are in the same direction, which will allow us to use convenient independence properties to reason about the evolution of conditional Bernoulli randomness as we traverse the path. Furthermore, this rooting assumption is without loss of generality, by the equivalence of tree Ising models and tree-structured Bayesian networks, as mentioned in \Cref{sec:tim-prelims}.

For a node $X_i$ in the tree Ising model, let us define its bias $b(i)=\min(\Pr[X_i=0], \Pr[X_i=1])$. The following inequality upper-bounds the absolute difference in the biases of $X_i$ and $X_j$ in terms of Bernoulli randomness and the probability that they disagree.
\begin{lemma}\label{lemma:br-bound}
    Let $b(i)=\min(\Pr[X_i=0], \Pr[X_i=1])$.
    \begin{enumerate}
        \item[(a)] $b(i)-b(j) \le \frac{1}{2}\cdot \Br(X_i|X_j).$
        \item[(b)]  $|b(i)-b(j)| \le \Pr[X_i \neq X_j].$
    \end{enumerate}
\end{lemma}
\begin{proof}
    Suppose $\Pr[X_i = 0], \Pr[X_j = 0] \ge \frac{1}{2}$ without loss of generality. Then, 
    \begin{align*}
    \Br(X_i | X_j) &\ge 2\cdot \min_x \Pr[X_i=x,X_j=0] = 2\cdot \left(\Pr[X_j = 0] - \max_x \Pr[X_i=x, X_j=0] \right)\\
    &\ge 2\cdot \left(\Pr[X_j = 0] - \max_x \Pr[X_i=x] \right) = 2 \cdot \left((1 - b(j)) - (1 - b(i)) \right) = 2\cdot \left( b(i) - b(j) \right)
    \end{align*}
    Repeating the calculation above with $\Br(X_j|X_i)$ yields $\Br(X_j|X_i) \ge 2\cdot \left( b(j) - b(i) \right)$. Thus, we get
    \begin{align*}
        |b(i)-b(j)| \le \frac{1}{2}\cdot\max\left(\Br(X_i, X_j), \Br(X_j|X_i)\right) \le \Pr[X_i \neq X_j]
    \end{align*}
    where we used \Cref{lemma:bern-lb-dist} in the last step.
\end{proof}

Now, let us define a strange but necessary quantity, that captures a ``surprise" event when we traverse an edge.

\begin{definition}[Cross]
We define an edge $e = (X_u,X_v)$ as being a ``cross'' if $\Pr[X_u = 0, X_v = 1] + \Pr[X_u = 1, X_v = 0] > \Pr[X_u = 0, X_v = 0] + \Pr[X_u = 1, X_v = 1]$, or equivalently     $\Pr[X_u = 0, X_v = 1] + \Pr[X_u = 1, X_v = 0] > \frac{1}{2}$.
\end{definition}

We will show how distance can be related to three sources and how we can embed these sources well. 

\begin{lemma}\label{claim:three}
Each of the following quantities is upper bounded by $O(\Pr[X_i \ne X_j])$:
    \begin{enumerate}
        \item[(a)] Difference in marginals: $d_{\marg}(i, j)=|\Pr[X_i = 0] - \Pr[X_j = 0]|$.
        \item[(b)] Bernoulli randomness: $d_{\Br}(i,j)=\Br(X_i | X_j)$.%
        \item[(c)] Negative correlation: $d_{\negcor}(i,j) = \begin{cases}
            \max(b(i),b(j)) & \text{if number of crosses between $i$ and $j$ is odd.} \\
            |b(i)-b(j)| & \text{if number of crosses between $i$ and $j$ is even.}
        \end{cases}$
    \end{enumerate}
\end{lemma}
\begin{proof}
    For part (a), observe that
    \begin{align*}
        |\Pr[X_i = 0] - \Pr[X_j = 0]| &= |\Pr[X_i = 0, X_j=1] - \Pr[X_i=1, X_j = 0] \\&\le \Pr[X_i = 0, X_j=1] + \Pr[X_i=1, X_j = 0] = \Pr[X_i \ne X_j].
    \end{align*}
    Part (b) follows from \Cref{lemma:bern-lb-dist} above.
    Part (c) is relatively most involved. Note that the ``even crosses'' case follows from \Cref{lemma:br-bound}. For the ``odd crosses" case, we have
    \begin{claim}\label{claim:odd-cross}
        If there are an odd number of crosses between $X_i$ and $X_j$, then $\min(b(i),b(j))=O(\Pr[X_i \ne X_j])$.
    \end{claim}
    \begin{proof}
        \textbf{Case 1:} $\Br(X_i | X_j) = 2 \cdot b(i)$. In this case, we use (2) to get
        \begin{align*}
            \min(b(i),b(j)) \le b(i) = \frac12\Br(X_i|X_j) = O(\Pr[X_i \ne X_j]).
        \end{align*}

        \textbf{Case 2:} $\Br(X_i | X_j) \ge \frac{1}{2}$. In this case, again by (2) we know 
        $$
        \min(b(i),b(j)) \le \frac12 \le \Br(X_i | X_j) = O(\Pr[X_i \ne X_j]).
        $$
        
        \textbf{Case 3:} $\Br(X_i | X_j) < 2 \cdot b(i), \frac{1}{2}$. 
        \begin{figure}[H]
            \centering
            \includegraphics[scale=0.35]{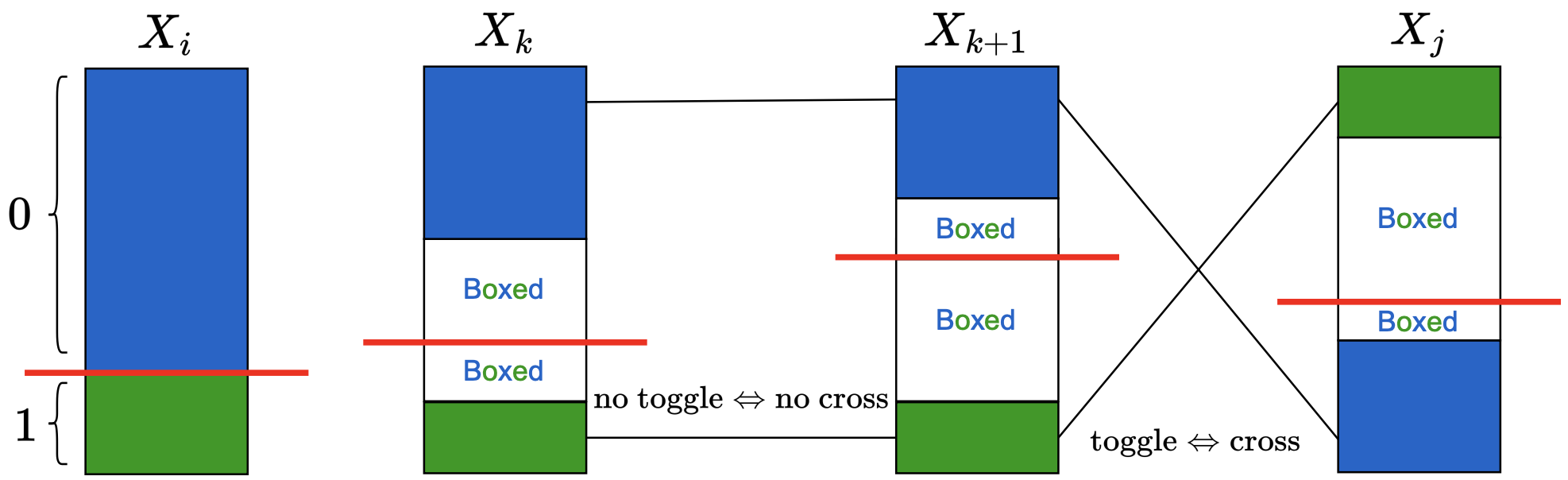}
            \caption{ The case where $\Br(X_i | X_j) < 2 \cdot b(i), \frac{1}{2}$. Each $X_{i+1},\dots,X_j$ has unmatched blue/green mass in each state, totalling more than $\frac12$. The boxed mass in each state has equal green and blue mass. Toggles and crosses are in one-to-one correspondence in this case. In the picture above, we are imagining there to be an odd number of toggles on the path $X_i \to X_j$. Observe how this results in all of the unmatched mass in $X_j$ contributing to $\Pr[X_i \neq X_j]$. In particular, for $X_j$ (for which \#toggles is odd), $\Pr[X_i \neq X_j]=1-\frac12\Br[X_i|X_j]$, while for $X_{k+1}$ (for which \#toggles is even), $\Pr[X_i \neq X_{k+1}]=\frac12\Br[X_i|X_{k+1}]$.}
            \label{fig:toggle-cross}
        \end{figure}
        In this case, observe that by \Cref{lemma:data-processing-inequality}, for any $X_k \in X_i \to X_j$, $\Br(X_i | X_k) < 2 \cdot b(i), \frac12$. In particular, the former inequality means that at every $k$, the two states of $X_k$ contain non-zero amounts of \textit{different} ``unmatched" states of $X_i$ (refer to \Cref{fig:toggle-cross}), and as we move from $X_k$ to $X_{k+1}$, the configuration of these unmatched states either toggles (blue-green $\to$ green-blue or green-blue $\to$ blue-green) or stays the same. The latter inequality means that these states further total to $>\frac12$ at every $k$. Observe that in this case, the configuration from $X_k$ to $X_{k+1}$ toggles if and only if the edge connecting them is a cross. This is because, since the unmatched states amount to a mass $>\frac12$, a toggle would mean that $\Pr[X_k=0, X_{k+1}=1] + \Pr[X_k=1, X_{k+1}=0] > \frac12$, which corresponds to a cross edge. Similarly, a non-toggle would mean that $\Pr[X_k=0, X_{k+1}=0] + \Pr[X_k=1, X_{k+1}=1] > \frac12$, which is a non-cross edge. Now, if we are given there are an odd number of crosses between $X_i$ and $X_j$, then that means there are an odd number of toggles, which means that the unmatched states contribute their mass towards $\Pr[X_i \neq X_j]$. Namely, we get $\Pr[X_i \neq X_j] > \frac12 \ge \min(b(i), b(j))$. 
    \end{proof}
    To conclude the proof of the ``odd crosses" case, we note that
    $$\max(b(i),b(j)) = \min(b(i),b(j)) + |b(i)-b(j)| = O(\Pr[X_i \ne X_j]),$$
    where we used \Cref{claim:odd-cross} and \Cref{lemma:br-bound}.
\end{proof}

We aim to show that the sum of the three metrics given in \Cref{claim:three} will be an approximation of $\Pr[X_i \ne X_j]$. This would follow from showing:
\begin{lemma} \label{lemma:cross-exist}
    Suppose $|\Pr[X_i = 0] - \Pr[X_j=0]|,\Br(X_i | X_j) \le \frac{1}{100} \cdot \Pr[X_i \ne X_j]$, then it must hold that $b(i),b(j) \in [0.485,0.5]$ and that there are an odd number of crosses between $X_i$ and $X_j$.
\end{lemma}
\begin{proof}
    \textbf{Case 1: $\Br(X_i | X_j) < 2 \cdot b(i)$.} In this case, as $\Br(X_i | X_j)\le \frac{1}{100}\cdot\Pr[X_i \neq X_j] < \frac{1}{2}$ as well, as depicted in \Cref{fig:toggle-cross}, it must hold that $\Pr[X_i \ne X_j] = \frac{1}{2} \cdot \Br(X_i | X_j)$ if there are an even number of crosses from $X_i$ to $X_j$, and $\Pr[X_i \ne X_j] = 1-\frac{1}{2} \cdot \Br(X_i | X_j)$ if there are an odd number of crosses. In the even case, it would then not hold that $\Br(X_i | X_j) \le \frac{1}{100} \cdot \Pr[X_i \ne X_j]$. Thus, we must be in the odd number of crosses case, and 
    $$\Pr[X_i \ne X_j] = 1-\frac{1}{2} \cdot \Br(X_i | X_j) \ge 1 - \frac{1}{200}.$$
    Without loss of generality, suppose $\Pr[X_i = 0] \ge \Pr[X_i = 1]$ (as in \Cref{fig:toggle-cross} too). Then, we can see that 
    $$\Pr[X_i = 0, X_j = 1] \ge \frac{1}{2} \cdot (1 - \Br(X_i | X_j)) \ge \frac{99}{200}.$$
    Accordingly, $\Pr[X_i = 0], \Pr[X_j =1] \ge \frac{99}{200}$. Moreover, as $|\Pr[X_i = 0] - \Pr[X_j = 0]| \le \frac{1}{100}$, then 
    $$\Pr[X_j = 0] \ge \frac{99}{200} - \frac{1}{100} = \frac{97}{200}.$$
    Since $\Pr[X_j = 0] \ge \frac{97}{200}$ and $\Pr[X_j = 1] \ge \frac{99}{200}$, it holds that $b(j) \in [0.485,0.5]$. Similarly,
    $$\Pr[X_i = 1] \ge \Pr[X_j = 1] - |\Pr[X_i = 0] - \Pr[X_j = 0]| \ge \frac{99}{200} - \frac{1}{100} = 0.485.$$
    So, $b(i) \in [0.485,0.5]$ too.
    
    \textbf{Case 2: $\Br(X_i | X_j) = 2\cdot b(i)$}.
    \begin{figure}[H]
    \centering
    \includegraphics[scale=0.35]{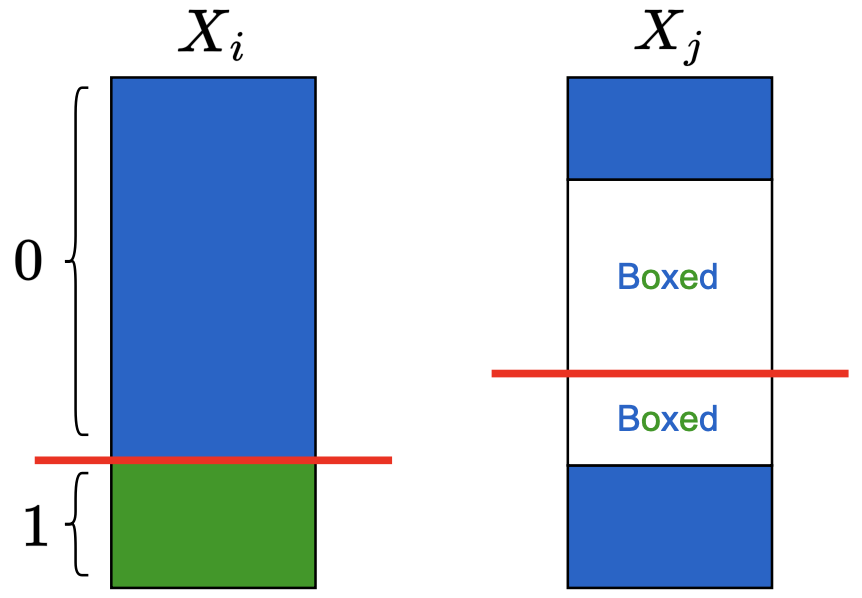}
    \caption{The case where $\Br(X_i | X_j) = 2 \cdot b(i).$ The mass corresponding to the smaller state in $X_i$ (green mass) should all be in a box in $X_j$.}
    \label{fig:br=2bi}
\end{figure}
        Without loss of generality assume $\Pr[X_i = 0] \ge \Pr[X_i = 1]$. Since $\Br(X_i | X_j) = 2\cdot b(i)$, the mass corresponding to the smaller state in $X_i$ (\Cref{fig:br=2bi}) must all be in a box in $X_j$. This means that 
    \begin{align}
        \label{eqn:a}
        \Pr[X_i \ne X_j] &\in [\Pr[X_i = 0, X_j = 1], \Pr[X_i = 1, X_j = 0]] \nonumber \\
        &\in\left[\Pr[X_i = 0, X_j = 1],\frac{1}{2}\cdot \Br(X_i | X_j) + \Pr[X_i = 0, X_j = 1]\right].
    \end{align}
    Since by assumption, $b(i)= \Pr[X_i=1]$,
    \begin{align*}
        \Pr[X_j=1] &\le \Pr[X_i=1] + |\Pr[X_i=1]-\Pr[X_j=1]| \\
        &=b(i) + |\Pr[X_i=1]-\Pr[X_j=1]| \\
        &=b(i) + |\Pr[X_i=0]-\Pr[X_j=0]| \\
        &=\frac12\cdot\Br(X_i|X_j) + |\Pr[X_i=0]-\Pr[X_j=0]| \\
        &\le \frac{2}{100} \cdot \Pr[X_i \ne X_j].
    \end{align*}
    But then combining with \eqref{eqn:a}, this gives
    \begin{align*}
        \Pr[X_i \ne X_j] &\le \frac{1}{2}\cdot \Br(X_i | X_j) + \Pr[X_i = 0, X_j = 1] \\
        &\le \frac{1}{2}\cdot \Br(X_i | X_j) + \Pr[X_j = 1] \\
        &\le \frac{1}{200} \Pr[X_i \ne X_j] + \frac{2}{200} \Pr[X_i \ne X_j] = \frac{3}{200} \Pr[X_i \ne X_j] < \Pr[X_i \ne X_j],
    \end{align*}
    which is a contradiction. 
\end{proof}

\begin{corollary}\label{corr:metrics-good}
The sum of the three metrics in \Cref{claim:three} i.e., $d_{\marg} + d_{\Br} + d_{\negcor}$ is a $\Theta(1)$-distortion embedding for $\Pr[X_i \ne X_j]$.
\end{corollary}
\begin{proof}
    By \Cref{claim:three}, we know that the metrics do not individually overestimate (upto constant factors). Now, we just want to show that their sum does not sufficiently underestimate. If the sum of the metrics was to underestimate under $\frac{1}{100} \cdot \Pr[X_i \ne X_j]$, then the conditions of \Cref{lemma:cross-exist} must hold. However, the result of \Cref{lemma:cross-exist} then implies that the third metric will yield distance $\ge 0.485$, causing a contradiction as the distance of any $\Pr[X_i \ne X_j] \le 1 < 100 \cdot 0.485$.
\end{proof}

Now all that remains is to show that each of the three metrics in \Cref{claim:three} can be embedded efficiently into $\ell_1$ with constant distortion. The first metric can be embedded exactly into $\ell_1$ simply with one coordinate having value $\Pr[X_i = 0]$. The third metric also easily embeds into $\ell_1$ by having a single coordinate for each $X_i$ which is $b(i)$ if there is an even number of crosses from the root to $X_i$, and $-b(i)$ if there is an odd number of crosses from the root to $X_i$. For any two nodes $X_i$ and $X_j$, note that this coordinate measures a distance $|b(i)-b(j)|$ iff there is an even number of crosses on the path between $X_i$ and $X_j$, and distance $|b(i)+b(j)| \in [\max(b(i), b(j)), 2\cdot\max(b(i)+b(j))]$ iff there is an odd number of crosses.

The second metric is the most complicated to embed into $\ell_1$. To do so, we once again seek to reduce the task at hand to embedding a capped-tree metric. Recall that in the case of symmetric tree Ising models, we were able to reduce the task to a tree metric with a \textit{fixed} cap. The more general asymmetric case requires us to be able to embed a tree metric with a \textit{Lipschitz} cap instead, where the Lipschitzness in the caps is with respect to the tree distance between nodes. More concretely, we consider a tree with the same tree structure as the tree Ising model. We specify edge weights for every edge $e=(X_u, X_v)$ in the tree, via a quantity that we call ``forced randomness".
\begin{definition}[Forced randomness]
    For an edge $e = (X_u,X_v)$, its forced randomness $\mathcal{F}_e$ is defined as follows: $\mathcal{F}_e = \max (\Br(X_u | X_v), \Br(X_v | X_u))$. 
\end{definition}
Now, we define the following Lipschitz capped-tree metric:
\begin{definition}[Forced randomness Lipschitz-capped tree metric]
    Every node $X_i$ in the tree is associated with a cap $b(i)$. The edge weights are given by $\mathcal{F}_e$. The distance is defined as 
    \begin{equation}
        d_{\mathcal{F}}(i,j) = \min\left(\sum_{e \in (i \to j)}\mcF_e, \max(b(i), b(j))\right).
    \end{equation}
\end{definition}
Note that the caps respect Lipschitzness with respect to the tree distance.
\begin{claim}
    The forced randomness Lipschitz-capped tree distance is a valid Lipschitz capping because for every $e = (X_u,X_v)$ it holds that $|b(u)-b(v)| \le \mathcal{F}_e$.
\end{claim}
\begin{proof}
    This follows directly from the proof of \Cref{lemma:br-bound}, which asserts that $|b(u)-b(v)| \le \frac{1}{2}\cdot\mathcal{F}_e$.
\end{proof}

We show that it suffices to be able to efficiently embed this Lipschitz-cap tree metric into $\ell_1$ in order to embed $\Br(X_i | X_j)$. First, we show that $d_{\mcF}(i, j)$ does not vastly underestimate the quantities that we are interested in.

\begin{lemma}[No underestimation]
    \label{lemma:forced-radomness-no-underestimate}
    Let $d_{\mathcal{F}}(i,j)$ denote the forced randomness Lipschitz-capped tree distance between $X_i$ and $X_j$. Then, $d_{\mathcal{F}}(i,j) \ge \min(\max(b(i),b(j)),\Br(X_i | X_j))$.
\end{lemma}
\begin{proof}
    Note how this immediately holds if we show that the \emph{uncapped} distance i.e., $\sum_{e \in (i \to j)}\mcF_e$ is at least $\Br(X_i | X_j)$. Recall that
    \begin{align*}
        \sum_{e=(X_k, X_{k+1}) \in (i \to j)}\mcF_e = \sum_{(X_k, X_{k+1}) \in (i \to j)}\max(\Br(X_k|X_{k+1}), \Br(X_{k+1}|X_{k})) \ge \sum_{(X_k, X_{k+1}) \in (i \to j)}\Br(X_k|X_{k+1}).
    \end{align*}
    Thus, it is sufficient to show:
    \begin{claim}
        \label{claim:br-sum-across-edges}
        $\sum_{k : i \rightarrow j} \Br(X_k | X_{k+1}) \ge \Br(X_i | X_j)$
    \end{claim}
    \begin{proof}
        \begin{figure}[H]
            \centering
            \includegraphics[scale=0.35]{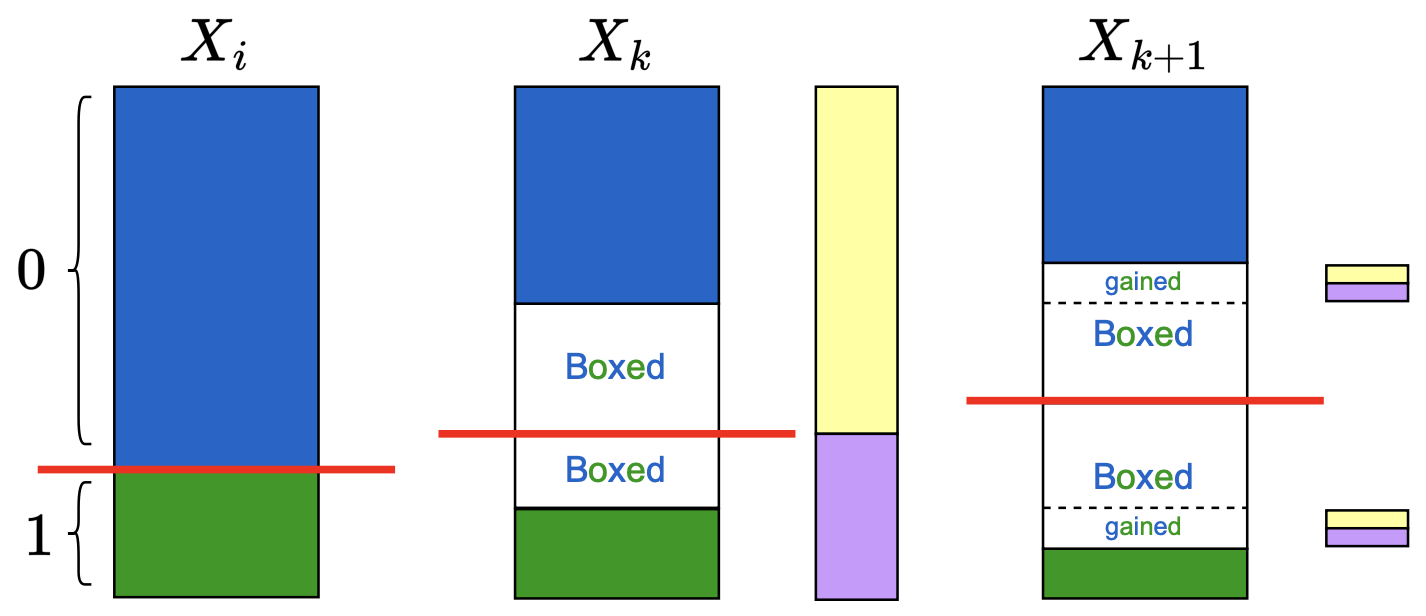}
            \caption{The ``blue-green" boxed mass in $X_k$ corresponds to $\Br[X_i|X_k]$. The gained blue-green boxed mass in $X_{k+1}$ must come from the unboxed mass from different states in $X_k$. Reinterpreting mass in the different states of $X_k$ as yellow and purple, we can see that the gained blue-green boxed mass certainly also constitutes yellow-purple boxed mass in $X_{k+1}$, meaning that it forms part of $\Br[X_k|X_{k+1}]$.}
            \label{fig:xk-xk+1}
        \end{figure}
        Let the mass in $X_i=0$ and $X_i=1$ be colored blue and green respectively, as in \Cref{fig:xk-xk+1}. For any $k:i \to j$, the ``blue-green" boxed mass is non-decreasing as we go from $X_k$ to $X_{k+1}$ (\Cref{lemma:data-processing-inequality}), and $\Br[X_i|X_{k+1}] - \Br[X_i|X_{k}]$ is exactly the increase in boxed mass. Observe that any increase in boxed mass at $X_{k+1}$ must come from mass that was unboxed and separate at $X_k$. Since this new boxed mass in either state of $X_{k+1}$ was from separate states in $X_{k}$, it also forms part of $\Br[X_k|X_{k+1}]$. This is made clear if we think of all the mass under $\Pr[X_k=0]$ as yellow, and all the mass under $\Pr[X_k=1]$ as purple, so that any additional blue-green boxed mass in $X_{k+1}$ is in fact also yellow-purple boxed mass. Thus, we have argued that $\Br[X_i|X_{k+1}]-\Br[X_i|X_{k}] \le \Br[X_{k}|X_{k+1}]$. Since this holds for every $k=i,i+1,\dots,j-1$, we can add up the inequalities to obtain a telescoping sum as follows:
        \begin{align*}
            \sum_{k:i \to j}\left(\Br[X_i|X_{k+1}]-\Br[X_i|X_{k}]\right)  = \Br[X_i|X_{j}]-\Br[X_i|X_{i}] = \Br[X_i|X_{j}] \le \sum_{k:i\to j}\Br[X_k|X_{k+1}].
        \end{align*}
    \end{proof}
\end{proof}

Second, and finally, we will show that this metric does not overestimate above $O(\Pr[X_i \ne X_j])$. 
 \begin{lemma}[No overestimation]
    \label{lemma:forced-radomness-no-overestimate}
    Let $d_{\mathcal{F}}(i,j)$ denote the forced randomness Lipschitz-capped tree distance between $X_i$ and $X_j$. Then, $d_{\mcF}(i,j) \le  200\cdot\Pr[X_i \ne X_j]$.
\end{lemma}
\begin{proof}
    Let $\mathcal{S}_{ij}$ be the set of edges $(X_k,X_{k+1})$ on the path from $i$ to $j$ for which $\Br(X_k | X_{k+1}) \ge \Br(X_{k+1} | X_k)$, and let $S_{ji}$ be the set of edges for which $\Br(X_k | X_{k+1}) < \Br(X_{k+1} | X_k)$. %
    
    Let us first process edges in $S_{ij}$. Suppose $\max(b(i),b(j)) > 100 \cdot \Pr[X_i \ne X_j]$ and $\sum_{e \in S_{ij}} \mathcal{F}_e > 100 \cdot \Pr[X_i \ne X_j]$. \\
    \noindent \textbf{Case 1:} $\Br(X_i | X_k) < \min\left(\frac{99}{100} b(i), \frac{1}{2} b(k)\right)$ for every $e=(X_k, X_{k+1})\in S_{ij}.$ \\
    
    \begin{figure}[H]
        \centering
        \includegraphics[scale=0.4]{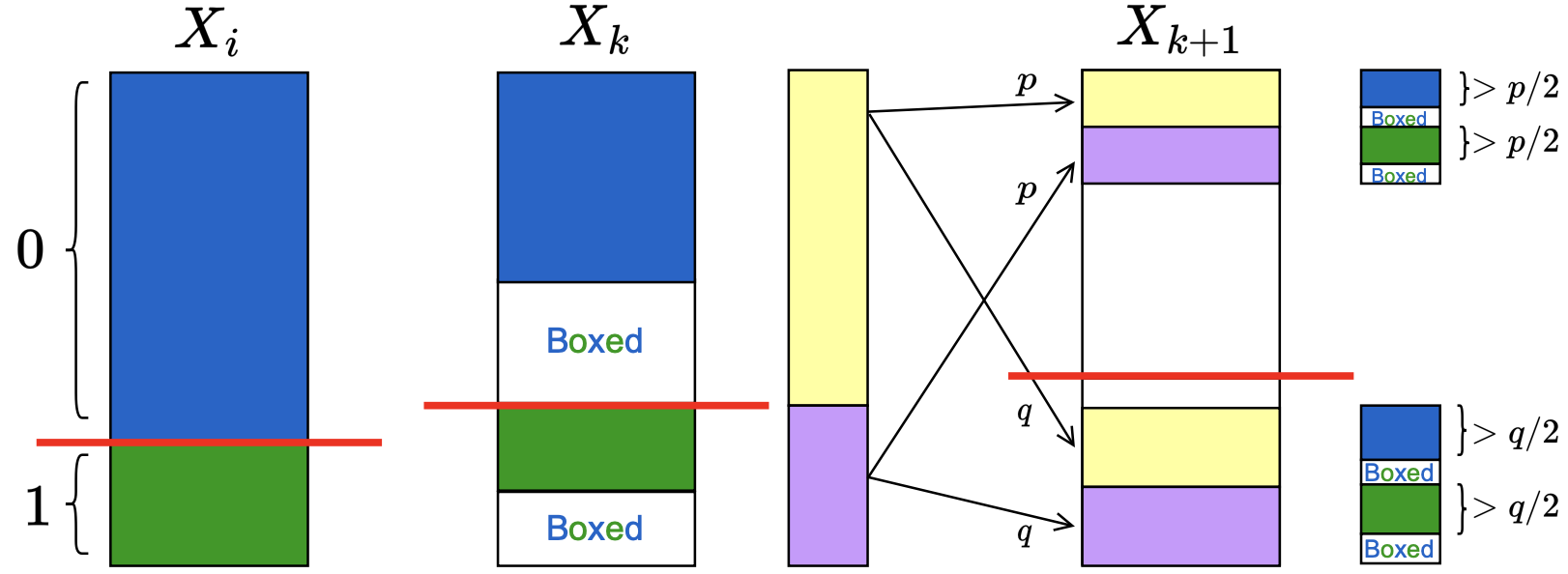}
        \caption{We reinterpret mass in the different states of $X_k$ as yellow and purple. The total ``yellow-purple" boxed mass in both states of $X_{k+1}$ constitutes $\Br[X_k|X_{k+1}]$. The gained blue-green boxed mass in $X_{k+1}$ must come from unboxed mass from different states in $X_k$. Under the assumption that $\Br(X_i | X_k) < \min\left(\frac{99}{100} b(i), \frac{1}{2} b(k)\right)$, at least half of $\Br[X_k|X_{k+1}]$ is derived from unboxed blue-green mass in $X_k$.}
        \label{fig:2br<bi}
    \end{figure}
    Consider any $e=(X_k, X_{k+1})\in S_{ij}$. Refer to \Cref{fig:2br<bi}. Since $\Br(X_i | X_k) < \frac{99}{100} b(i) < 2b(i)$, there must be different-colored unboxed mass overflowing in each state of $X_k$. As mentioned previously in the proof of \Cref{claim:br-sum-across-edges}, $\Br[X_i|X_{k+1}] - \Br[X_i|X_{k}]$ is exactly the increase in ``blue-green" boxed mass as we go from $X_k$ to $X_{k+1}$, and this increase in boxed mass \textit{must} come from mass that was unboxed and separate at $X_k$. With a view to reason about $\Br[X_k|X_{k+1}]$, let us reinterpret the mass in state 0 at $X_k$ as yellow and that in state 1 as purple. Then, $\Br[X_k|X_{k+1}]$ is precisely the total ``yellow-purple" boxed mass across both states of $X_{k+1}$. In the figure, $\Br[X_k|X_{k+1}]=2p+2q$, corresponding to a $p$ amount of mass coming from each state of $X_k$ into state 0 of $X_{k+1}$, and a $q$ amount of mass coming from each state of $X_k$ into state 1 of $X_{k+1}$. But recall that because of the Markov property, the proportion of blue and green mass in, say the $p$ amount of yellow mass coming from state $``X_k=0"$ into state $``X_{k+1}=0"$ is the same as it was in all of state $``X_{k}=0"$. Since we are also assuming that $\Br[X_i|X_k]<\frac12 b(k)$, this means that at least one half of the yellow mass constituting $\Br[X_k|X_{k+1}]$ in state $``X_{k+1}=0"$ corresponds to unboxed blue mass. Similarly, at least one half of the purple mass constituting $\Br[X_k|X_{k+1}]$ in state $``X_{k+1}=0"$ corresponds to unboxed green mass. This means that we have at least $\frac{p}{2}+\frac{p}{2}=p$ new boxed blue-green mass in state 0 of $X_{k+1}$ arising out of previously unboxed blue-green mass in $X_k$. The same holds true for the state $X_{k+1}$, where we have at least $\frac{q}{2}+\frac{q}{2}=q$ new boxed blue-green mass arising out of unboxed mass in $X_k$. In total, we have at least $p+q \ge \frac{1}{2}\Br[X_k|X_{k+1}]$ new boxed blue-green mass in $X_{k+1}$ arising out of unboxed mass in $X_k$, meaning, from our earlier reasoning and the definition of $\mcF_e$ for $e \in S_{ij}$, that
    $$
        \Br[X_i|X_{k+1}] - \Br[X_i|X_{k}] \ge \frac{1}{2}\Br[X_k|X_{k+1}] = \frac{1}{2}\mcF_e. %
    $$
    Under the current case, this logic holds for every $e=(X_k, X_{k+1}) \in S_{ij}$, which means that
    \begin{align*}
        \frac{1}{2} \sum_{e \in S_{ij}}\mathcal{F}_e &\le \sum_{e=(X_k,X_{k+1}) \in S_{ij}}\Br(X_i | X_{k+1})-\Br(X_i | X_k) \nonumber \\
        &\le \sum_{k:i\to j}\Br(X_i | X_{k+1})-\Br(X_i | X_k) \nonumber \\
        &= \Br[X_i|X_j] \\
        &\le 2\Pr[X_i \neq X_j] \\
        &< \frac{1}{50}\sum_{e \in S_{ij}}\mathcal{F}_e
    \end{align*}
    which is a contradiction. 
    
    Thus, it must be the case that $\Br(X_i | X_k) \ge \min\left(\frac{99}{100} b(i), \frac{1}{2} b(k)\right)$ for \textit{some} $e=(X_k, X_{k+1})\in S_{ij}.$ Now, observe that under our assumptions and by \Cref{lemma:br-bound},
    \begin{equation}
        \label{eqn:bi-bj-bound}
        |b(i)-b(j)| \le \Pr[X_i \neq X_j] < \frac{1}{100}\max(b(i),b(j)).
    \end{equation}
    Therefore,
    \begin{equation}
        \label{eqn:bi-bound}
        b(i) \ge \max(b(i),b(j)) - |b(i)-b(j)| > \frac{99}{100}\max(b(i),b(j)).
    \end{equation}

    \noindent \textbf{Case 2:} $\Br(X_i | X_k) \ge \frac{99}{100} b(i)$ for some $e=(X_k, X_{k+1}) \in S_{ij}$. \\
    
    If $\Br(X_i | X_k) \ge \frac{99}{100} b(i)$, we would have
    \begin{align*}
        \Pr[X_i\neq X_j] &\ge \frac{1}{2}\Br[X_i | X_j] \\
        &\ge \frac{1}{2}\Br[X_i | X_k] \qquad (\text{\Cref{lemma:data-processing-inequality}})\\
        &\ge \frac{1}{2}\cdot \frac{99}{100} b(i) \\
        &>\frac{1}{2}\cdot \left(\frac{99}{100}\right)^2\max(b(i),b(j))\qquad (\text{from }\eqref{eqn:bi-bound}) \\
        &> \Pr[X_i\neq X_j],
    \end{align*}
    which is a contradiction. %

    \noindent \textbf{Case 3:} $\Br(X_i | X_k) \ge \frac{1}{2} b(k)$ but $\Br(X_i | X_k) < \frac{99}{100} b(i)$ for some $e=(X_k, X_{k+1}) \in S_{ij}$. \\
    If $b(k) \ge \frac{1}{5}b(i)$, then we would have
    \begin{align*}
        \Pr[X_i \neq X_j] &\ge \frac{1}{2}\Br[X_i | X_k] \ge \frac{1}{4}b(k) \ge \frac{1}{20}b(i) > \frac{1}{20}\cdot\frac{99}{100}\max(b(i), b(j)) > \Pr[X_i \neq X_j],
    \end{align*}
    which is a contradiction. \\
    On the other hand, if $b(k) < \frac{1}{5}b(i)$, we would have
    \begin{align*}
        \Pr[X_i \neq X_j] &\ge \frac{1}{2}\Br[X_i | X_k] \\
        &\ge b(i)-b(k) \qquad (\text{\Cref{lemma:br-bound}}) \\
        &> \frac{4}{5}b(i) > \frac{4}{5}\cdot\frac{99}{100}\max(b(i), b(j)) > \Pr[X_i \neq X_j],
    \end{align*}
    which is again contradiction.

    Thus, we have shown that if both $\max(b(i),b(j)) > 100 \cdot \Pr[X_i \ne X_j]$ and $\sum_{e \in S_{ij}} \mathcal{F}_e > 100 \cdot \Pr[X_i \ne X_j]$, every possible case leads to a contradiction. Consequently, it must be true that 
    \begin{equation}
        \label{eqn:sij}
        \min\left(\max(b(i),b(j)), \sum_{e \in S_{ij}} \mathcal{F}_e\right)\le 100\cdot \Pr[X_i \ne X_j].
    \end{equation}
    Repeating the argument above for edges $e \in S_{ji}$ gives
    \begin{equation}
        \label{eqn:sji}
        \min\left(\max(b(i),b(j)), \sum_{e \in S_{ji}} \mathcal{F}_e\right)\le 100\cdot \Pr[X_i \ne X_j].
    \end{equation}
    The lemma follows by putting \eqref{eqn:sij} and \eqref{eqn:sji} together, and using
    \begin{align*}
        d_{\mcF}(i,j) &= \min\left(\max(b(i),b(j)), \sum_{e \in (i \to j)} \mathcal{F}_e\right) \\
        &=  \min\left(\max(b(i),b(j)), \sum_{e \in S_{ij}} \mathcal{F}_e + \sum_{e \in S_{ji}} \mathcal{F}_e\right) \\
        &\le \min\left(\max(b(i),b(j)), \sum_{e \in S_{ij}} \mathcal{F}_e\right) +  \min\left(\max(b(i),b(j)), \sum_{e \in S_{ji}} \mathcal{F}_e\right) \\
        &\le 200\cdot \Pr[X_i \ne X_j].
    \end{align*}

\end{proof}
Therefore, we have shown that the forced randomness Lipschitz cap tree metric neither grossly underestimates nor overestimates. This lets us substitute $d_{\mcF}$ for $d_{\Br}$ in \Cref{corr:metrics-good}.
\begin{corollary}\label{corr:modified-metrics-good}
The sum $d_{\marg} + d_{\mcF} + d_{\negcor}$ is a $\Theta(1)$-distortion embedding for $\Pr[X_i \ne X_j]$.
\end{corollary}
\begin{proof}
    Fix nodes $X_i$ and $X_j$. From \Cref{corr:metrics-good}, we know that
    \begin{equation}
        \label{eqn:good-metrics-to-modified-good-metrics}
        \Theta(1) \cdot \Pr[X_i \neq X_j] \le d_{\marg}(i,j) + d_{\Br}(i,j) + d_{\negcor}(i,j) \le \Theta(1) \cdot \Pr[X_i \neq X_j]. 
    \end{equation}
    This gives
    \begin{align*}
        d_{\marg}(i,j) + d_{\mcF}(i,j) + d_{\negcor}(i,j) &\le d_{\marg}(i,j) + 100\cdot\Pr[X_i \neq X_j] + d_{\negcor}(i,j) \qquad(\text{\Cref{lemma:forced-radomness-no-overestimate}}) \\
        &\le \Theta(1)\cdot(d_{\marg}(i,j) + d_{\Br}(i,j) + d_{\negcor}(i,j)) \qquad(\text{lower bound in \eqref{eqn:good-metrics-to-modified-good-metrics}}) \\
        &\le \Theta(1)\cdot \Pr[X_i \neq X_j] \qquad(\text{upper bound in \eqref{eqn:good-metrics-to-modified-good-metrics}}).
    \end{align*}
    Next, from \Cref{lemma:forced-radomness-no-underestimate}, we know that
    \begin{align*}
        d_{\mcF}(i,j) \ge \min(\max(b(i), b(j)), \Br[X_i|X_j]).
    \end{align*}
    But notice that 
    \begin{align*}
        \max(b(i), b(j)) &\ge b(i) \ge \frac{1}{2}\Br[X_i|X_j],
    \end{align*}
    and hence
    \begin{align*}
        d_{\mcF}(i,j) \ge \min(\max(b(i), b(j)), \Br[X_i|X_j]) \ge \frac{1}{2}\Br[X_i|X_j] = \frac{1}{2}d_{\Br}(i, j).
    \end{align*}
    Finally,
    \begin{align*}
        d_{\marg}(i,j) + d_{\mcF}(i,j) + d_{\negcor}(i,j) &\ge d_{\marg}(i,j) +  \frac{1}{2}d_{\Br}(i, j) + d_{\negcor}(i,j) \\
        &\ge \frac12\left(d_{\marg}(i,j) + d_{\Br}(i,j) + d_{\negcor}(i,j)\right) \\
        &\ge \Theta(1)\cdot\Pr[X_i \neq X_j] \qquad(\text{lower bound in \eqref{eqn:good-metrics-to-modified-good-metrics}}).
    \end{align*}
\end{proof}
In summary, we have shown that if we are able to efficiently embed the forced randomness Lipschitz cap tree metric $d_{\mcF}$ efficiently into $\ell_1$ with constant distortion, we will have achieved our goal of embedding $\Pr[X_i \neq X_j]$ into $\ell_1$ (recall that $d_\marg$ and $d_{\negcor}$ embed into $\ell_1$ in a very simple manner). In the following section, we show that Lipschitz cap tree metrics embed into $\ell_1$ with only $O(\log^2(n))$ dimensions.

\subsection{Lipschitz Cap Metrics}
\label{sec:lipschitz-cap}
As before, we begin our journey on Lipschitz cap metrics with the basic case of a line graph metric. As it turns out, the technique of lazy snaking that we used for the fixed cap line ends up being sufficient for the Lipschitz cap case, with a slight modification.

\subsubsection{Lipschitz Cap Line Metrics}
\label{sec:lipschitz-line}
The work of \cite{abraham2022metric} shows that Lipschitz cap line metrics embed into $\ell_1$ using $O(\log n)$ dimensions with constant distortion. We will present an alternative approach that also attains this results, both for sake of completeness, and because it lends itself better towards our later result of \Cref{thm:general-cap}. Let $L_n$ be the line graph on $n$ vertices $x_1,\dots,x_n$, where each pair of consecutive vertices $x_i,x_{i+1}$ is connected by an edge of length $e_i \in \Z_{>0}$, and $d(x_i, x_j)=\sum_{k=i}^{j-1}e_i$. Let $loc[i]=d(x_i, x_1)$. Throughout what follows, we will identify every vertex $x_i$ with $loc[i]$ instead. Every vertex has associated with it a cap, given by the cap function $M: \R \to \Z_{\ge 0}$, so that the cap at vertex $x_i$ is $M(loc[i])$. The cap function $M$ satisfies the Lipschitz property in the graph distance, i.e., for any $i,j$, $|M(loc[i])-M(loc[j])|\le |loc[i]-loc[j]|$. Let $loc[0]=-2M(0)$, and for any $i\ge0$, $t \in [loc[i], loc[i+1]]$, let us additionally define $M(t)$ to be the linear interpolation of $M(loc[i])$ and $M(loc[i+1])$ i.e., $M(t)=M(loc[i])+\left(\frac{M(loc[i+1])-M(loc[i])}{loc[i+1]-loc[i]}\right)\cdot(t-loc[i])$. Consider the metric space $(L_n, \dlcap{M})$ equipped with the Lipschitz cap line metric, defined as
\begin{align}
    \dlcap{M}(x_i, x_j) = \min(d(x_i, x_j), \max(M(loc[i]), M(loc[j])).
\end{align}

Given that we have a different cap value at every node, a natural idea, and intuitive generalization of the lazy snaking procedure in \Cref{algo:lazy-snake} would be to snake around with varying widths depending on the cap function. In fact, this modification is sufficient for our purposes, and is described in \Cref{algo:lipschitz-lazy-snake}\footnote{Observe that this algorithm might not terminate if $M(t)$ happens to be $0$ at some location $t$. We fix this issue in the general tree case.}. 

\begin{algorithm}[t]
    \caption{Lazy Snaking for Lipschitz-cap line metric} \label{algo:lipschitz-lazy-snake}
    \hspace*{\algorithmicindent} 
    \begin{flushleft}
      {\bf Input:} List of $n$ node locations $loc[]$, Lipschitz cap function $M$ \\
      {\bf Output:} List of embeddings for the nodes $embedding[]$ \\
    \end{flushleft}
    \begin{algorithmic}[1]
    \Procedure{LipschitzLazySnake}{$loc, M$}:
    \State \blue{$t \gets -2M(0)$}
    \While{$t \leq loc[n]$}
        \State $\Delta \sim \mathrm{Uniform}\{0,1\}$
        \If{$\Delta = 0$} \Comment{Rest for a duration $\frac{M(t)}{300}$}
            \State $snake[t'] \gets 0$ for $t' \in \left[t, t+\frac{M(t)}{300}\right]$
            \State $t \gets t + \frac{M(t)}{300}$
        \EndIf
        \If{$\Delta = 1$} \Comment{Snake with a width $\frac{M(t)}{100}$}
            \State $snake[t'] = t'-t$ for $t' \in \left[t, t+\frac{M(t)}{100}\right]$ 
            \State $snake[t'] = \frac{2M(t)}{100}-(t'-t)$ for $t' \in \left[t+\frac{M(t)}{100}, t+\frac{2M(t)}{100}\right]$
            \State $t \gets t + \frac{2M(t)}{100}$
        \EndIf    
    \EndWhile
    \For{$i \gets 1$ to $n$}
        \State $embedding[i] = snake[loc[i]]$
    \EndFor
    \State {\bf return} $embedding[]$
    \EndProcedure
    \end{algorithmic}
\end{algorithm}

\begin{figure}[H]
    \centering
    \includegraphics[scale=0.42]{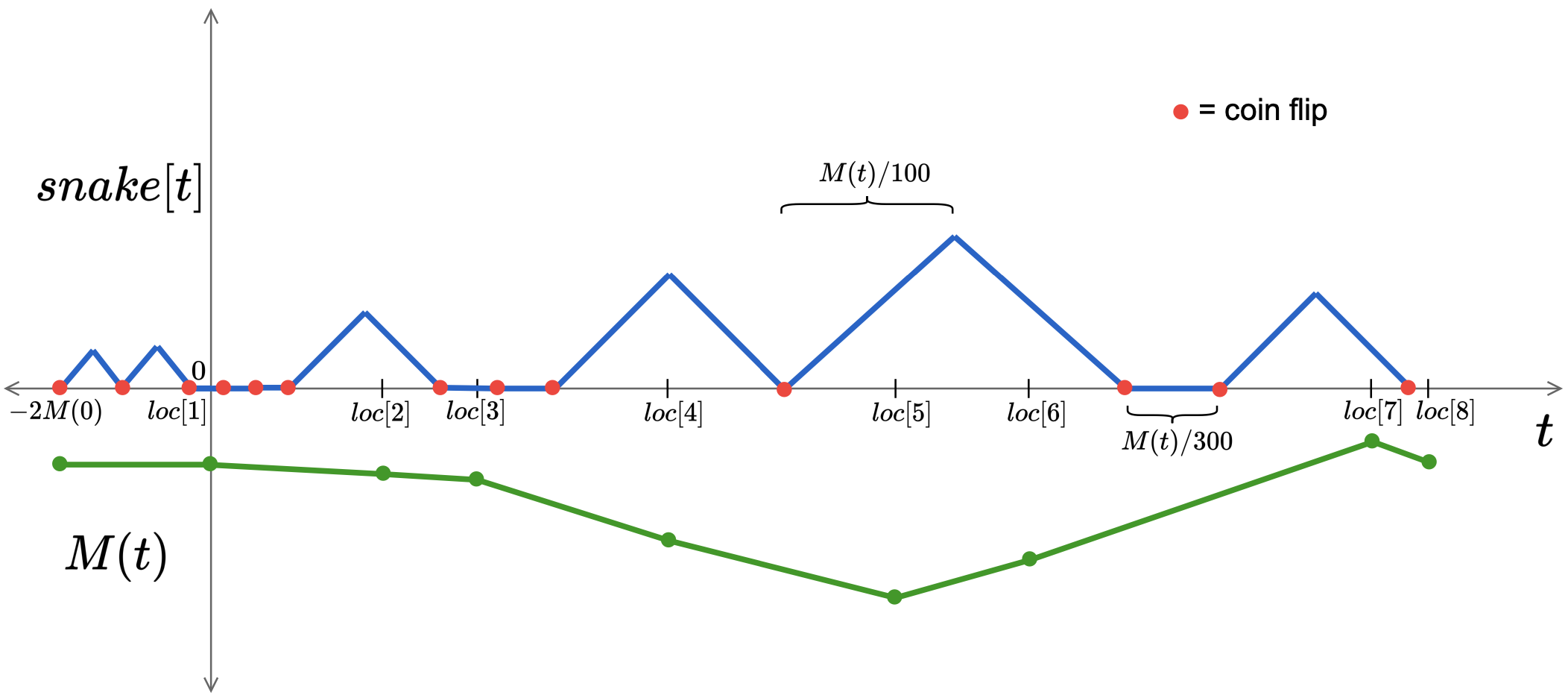}
    \caption{Lazy snaking with Lipschitz cap.}
    \label{fig:lazy-snaking-lipschitz-cap}
\end{figure}

We remark that \Cref{algo:lazy-snake} and \Cref{algo:lipschitz-lazy-snake} are nearly identical, except for constants and initialization. The modified initialization is necessary for algorithm correctness when the cap is non-constant, for a technical reason.

Just like \Cref{claim:lazy-snake-line-fixed-cap-lb}, we have the guarantee that for every pair of nodes, \Cref{algo:lipschitz-lazy-snake} does not overestimate their distance, and also covers at least a constant fraction of it in expectation.
\begin{restatable}{lemma}{lazysnakelipschitzcapline}
    \label{lemma:lazy-snake-line-lipschitz-cap-lb}
    Let $M$ be a strictly positive function, and let $embedding=\textsc{LipschitzLazySnake}(loc, M)$. For every fixed pair of nodes $x_i$ and $x_j$,
    \begin{align*}
    &\text{(1)  } \Pr\left[\|\embed[i]-\embed[j]\|_1 \le \dlcap{M}(x_i, x_j)\right] =1.\\
    &\text{(2)  } \E\left[\|\embed[i]-\embed[j]\|_1\right]  \ge \Omega(1)\cdot \dlcap{M}(x_i, x_j).
    \end{align*}
\end{restatable}
\begin{proof}
    The proof mirrors the proof of \Cref{claim:lazy-snake-line-fixed-cap-lb} stepwise (although with slightly messier calculations), and is given in \Cref{sec:lipschitz-line-proof}.
\end{proof}
Boosting the above in-expectation guarantee using logarithmically many dimensions then yields the following theorem:
\begin{theorem}[Lipschitz cap line into $\ell_1$]
    \label{thm:lipschitz-capped-line}
    $(L_n, \dlcap{M})$ can be embedded into $(\R^d, \ell_1)$ where $d=O(\log n)$ with $\Theta(1)$ distortion.
\end{theorem}

\subsubsection{Lipschitz Cap Tree Metrics}\label{sec:lipschitz-tree}

We will now consider how to obtain the analogous result for Lipschitz cap tree metrics. It is tempting to try and combine our techniques from the fixed cap tree metric and our Lipschitz cap line metric. Recall how our main technique for the fixed cap tree metric was to obtain a modified ``snipped" version of the tree's caterpillar decomposition and to reduce (by hashing) the embedding problem to a collection of fixed cap line metric problems. However, we cannot reduce the Lipschitz cap tree metric to a collection of Lipschitz cap line metrics in an obvious manner. Crucially, a direct modification of the prior approach would not work for Lipschitz cap tree metrics, primarily because segments that are hashed to a line metric instance are not contiguous segments of the tree, and thus \emph{the caps on these segments need no longer satisfy the necessary Lipschitz condition.} Consequently, we must design a new algorithm for this task. 

\textbf{Algorithm Intuition.} First, we observe that the primary issue in extending our prior techniques was how tree edges hashed to the same line metrics potentially have very different caps. If we view the algorithm hierarchically (considering how the embedding evolves as we progress downwards in the tree), it would be desirable to ``clean''/zero out the embedding somehow. If we clean the embedding frequently enough, we may expect that the only edges affecting some $\embed[x_i]$ would be parts of the tree very close to $x_i$, and thus by the Lipschitz property, all the relevant edges would have roughly the same cap. Ultimately, we obtain an algorithm that is less directly similar to the previous approaches, but is conducive towards this notion of periodic cleaning.

\textbf{Build-Clean-Tree Algorithm.} We now define our ``Build-Clean-Tree" algorithm.

\underline{\emph{Tree decomposition}}. We will again consider a modified caterpillar decomposition. Abusing notation slightly, we will refer to the caterpillars/segments given by the caterpillar decomposition as edges themselves. For any edge $e$ in the tree decomposition, we will denote $s(e)$ to be the location of the ``start'' (or top) of the edge, $\ell(e)$ the length of the edge, and $a(e)$ will store an auxiliary value for the edge we will later define. After obtaining the original caterpillar decomposition, we will split up any overly long edges. More concretely, we chop off any edge $e$ having length larger than $\frac{M(s(e))}{100\log(n)}$ at $\frac{M(s(e))}{100\log(n)}$%
, and then recursively continue chopping the remaining part of the edge.\footnote{The most natural manner of doing this splitting faces a nuanced issue that if one of the caps at the endpoints is zero and the other is nonzero, then the preprocessing would not terminate as it splits the edge into infinitely many edges. This is remedied by separately handling the case where all caps are strictly positive in \Cref{lemma:tree-lipschitz-cap-lb}, and then reducing the general problem to this special case in \Cref{lemma:general-tree-lipschitz-cap-lb}.} This process ensures that the length of every (potentially chopped) edge $e$ is at most $\frac{M(s(e))}{100\log(n)}$.
From this decomposition, we will create the embedding. \underline{All embeddings will be of length $|H| = 8 \log(n)$}. Consider building the embedding from the top to bottom, in a way such that we will only process an edge $e$ after its parent has been processed. We will then determine the embedding for everything along the edge $e$ as purely a function of: $embedding[s(e)]$ (this is already calculated because the embedding has been calculated for everything along the parent edge), the length $\ell(e)$, and auxiliary information $a(e)$. Our algorithm will work in stages that we call ``building'' and ``cleaning'', and the auxiliary information $a(e)$ will store what determines the current stage.

\underline{\emph{Building}}. One stage of our algorithm is \emph{building}. If we process an edge $e$ while the auxiliary information $a(e)$ indicates it is the building stage, we will then uniformly at random determine a hash for the edge $h(e) \in [8 \log(n)]$. Let $embedding[i](k)$ denote the $k^\text{th}$ coordinate of the embedding at location $i$. Then, we will proceed as follows:

\begin{itemize}
    \item If $embedding[s(e)](h(e)) = 0$, then we will walk in the positive direction for the coordinate $h(e)$ as traversing the edge. More formally, suppose $embedding[s(e)+v]$ denotes the embedding at $v$ below $s(e)$ for $v \in [0,\ell(e)]$. Then, we set $embedding[s(e)+v](h(e))=v$. All other coordinates remain the same as $embedding[s(e)]$.
    \item Otherwise, if $embedding[s(e)](h(e)) \ne 0$, then we do nothing and keep the embedding entirely the same as $embedding[s(e)]$.
\end{itemize}

Note how this process is Lipschitz in how the embedding changes while moving along the tree, and is designed in a way such that we are trying to keep coordinates of the embedding to be at most $O(\frac{M(s(e))}{\log(n)})$ by leveraging properties of our tree decomposition that limits the sizes of edges.

\underline{\emph{Cleaning}}. The other stage of our algorithm is \emph{cleaning}. When we are in a cleaning stage, it is our hope to try make the state closer to $0$, but we must do so in a Lipschitz manner. Accordingly, our algorithm will be to use the edge to walk a coordinate negatively towards $0$. We proceed by:

\begin{itemize}
    \item If $embedding[s(e)](k)=0$ for all $k$, then do nothing and keep the embedding the same as $embedding[s(e)]$.
    \item Otherwise, let $k^*$ denote the smallest $k$ such that $embedding[s(e)](k^*) > 0$. Then, we will walk negatively for a length of $\min(embedding[s(e)](k^*), \ell(e))$. More formally, we set $embedding[s(e)+v](k^*)=embedding[s(e)](k^*)-v$ for $v \in [0,\min(embedding[s(e)](k^*), \ell(e))]$
\end{itemize}

\underline{\emph{Controlling stages}}.  We will design our stages such that build stages process $\lb = 4 \log(n)$ edges, and clean stages process $\lc = 9 \log(n)$ edges, and we alternate between building and cleaning. Our auxiliary information $a(e)$ will actually be an integer counter $\in \{0, \dots,\lb + \lc-1\}$, such that the edge will be processed in the build stage if $a(e) \in \{0,\dots, \lb-1\}$ and otherwise in the clean stage if $a(e) \in \{\lb,\dots,\lb+\lc-1\}$. Moreover, $a(e)$ will simply be $a(par(e))+1 \mod (\lb + \lc)$. %

\underline{\emph{Initialization}}. For technical reasons similar to those of the Lipschitz cap line metric, we will modify the tree so that the root $r$ of the original tree actually has a parent that is a new node $r_0$, where $M(r_0)=M(r)$ and the length of the edge $(r_0,r)$ is $2M(r)$. We initialize $embedding[r_0] = 0^{|H|}$. The auxiliary information of the topmost edge in the tree which contains the extra node $r_0$ is initialized uniformly at random in $\{0,\dots,\lb+\lc-1\}$.  %

These are the main components of our algorithm.

\textbf{Analysis intuition. } Recall how our tree composition will limit the length of an edge to be $O(\frac{M(s(e))}{\log(n)})$. We will set the length of our stages $\lb,\lc = \Theta(\log(n))$. With correctly chosen parameters, we can ensure properties such as the following:

\begin{itemize}
    \item At the end of every clean stage, the embedding is exactly $0$ for every coordinate.
    \item Leveraging how often we clean, we can show that for any location $x$ it must hold $embedding[x](k) = O\left(\frac{M(x)}{\log(n)}\right)$ for all indices $k$.
    \item Moreover, using how the embedding is always updated in a Lipschitz manner, and how $\|embedding[x]\|_1$\\ $= O(M(x))$ by the previous fact, the embeddings never overestimate distances by more than a constant factor, or more formally $\|embedding[x_i]-embedding[x_j]\|_1 \le \dlcap{M}(x_i,x_j)$.
    \item By analyzing casework, we can also show how this embedding captures at least a constant fraction of the correct distance in expectation.
\end{itemize}

The main thrust of our proof will focus on a special case where all caps are strictly positive, and then we will reduce the general case to the special case. For what follows, let us assume that $M(x)>0$ at all locations $x$ in the tree. Here again, we have linearly interpolated the cap function $M$ to assign cap values at all points across any edge $(x_i, x_j)$ based on cap values at the endpoints $M(x_i)$ and $M(x_j)$.

We begin by proving a useful property for our Lipschitz cap trees, that the caps of nearby parts of the tree must have similar caps. For the simplest version of this, we claim about the similarity between points on an edge:

\begin{claim}[Caps on an edge are close by]\label{claim:edge-similar}
    For any edge $e = (u,w)$ in the decomposition with $u=s(e)$ and location $v$ along $e$, it must hold that $\frac{M(u)}{M(v)},\frac{M(v)}{M(u)} \le 1 + \frac{1}{50 \log(n)}$.
\end{claim}
\begin{proof}
    By Lipschitzness, we have that
    \begin{align}
        & |M(u)-M(v)| \le \ell(e) \le \frac{M(u)}{100 \log(n)} \\
        & \implies \left(1 - \frac{1}{100 \log(n)}\right) M(u) \le M(v) \le \left(1 + \frac{1}{100 \log(n)}\right)M(u) \\
        & \implies \frac{M(u)}{M(v)},\frac{M(v)}{M(u)} \le \max\left(1 + \frac{1}{100 \log(n)},\frac{1}{1 - \frac{1}{100 \log(n)}}\right) \\
        &=\max\left(1 + \frac{1}{100 \log(n)},1+\frac{1}{100\log(n)\left(1 - \frac{1}{100 \log(n)}\right)}\right)\\
        &= 1+\frac{1}{100\log(n)\left(1 - \frac{1}{100 \log(n)}\right)} \le 1 + \frac{1}{50 \log(n)}.
    \end{align}
    The last line holds when $\frac{1}{100 \log(n)} \le \frac{1}{2}$, which holds for $n \ge 2$.\footnote{Note how any requirement in proofs that $n$ be sufficiently large can always be handled by adding meaningless nodes to our tree with arbitrary edge weights and caps respecting Lipschitzness.}
\end{proof}

This similarly lets us prove claims about the similarity in the caps at starts of edges that are somewhat close:

\begin{claim}[Caps on a path are close by]\label{claim:path-similar}
    Consider a sequence of edges $e_1, \dots, e_k$, such that each $e_i,e_{i+1}$ is adjacent (meaning they share a vertex), and $x_1$ is on $e_1$ and $x_k$ is on $e_k$. Then, it must hold that $\frac{M(x_k)}{M(x_1)} \le e^{\frac{k}{50\log(n)}}$. As a special case, $\frac{M(s(e_k))}{M(s(e_1))} \le e^{\frac{k}{50\log(n)}}$.
\end{claim}
\begin{proof}
    Note how the adjacency condition in the sequence $e_1,\dots,e_k$ implies that it is possible to traverse from $x_1$ to $x_k$ by crossing at most $k$ edges. Thus, by \Cref{claim:edge-similar} it holds that:

    \begin{align}
        \frac{M(x_k)}{M(x_1)} &= \frac{M(x_k)}{M(s(e_k))}\cdot\frac{M(s(e_k))}{M(s(e_{k-1}))}\cdot\dots\cdot\frac{M(s(e_3))}{M(s(e_2))}\cdot\frac{M(s(e_2))}{M(x_1)} \\
        & \le \left(1 + \frac{1}{50 \log(n)}\right)^k \\
        & \le e^{\frac{k}{50 \log(n)}},
    \end{align}
    where in the last step, we used $1+x \le e^x$ for all real $x$.
\end{proof}

Intuitively, this claim will enable us to relate the lengths of almost all nearby edges, as on any root-to-leaf path, it will hold that all but $\le \log(n)$ edges satisfy $\ell(e) = \frac{M(s(e))}{100 \log(n)}$. This is because every edge $e$ on a root-to-leaf path that has length smaller than $\frac{M(s(e))}{100 \log(n)}$ must either correspond to a whole short caterpillar, or be the last piece of a caterpillar, and there are at most $\log(n)$ caterpillars on any root-to-leaf path.

Next, we will prove a useful property about the cleaning stage---after an entire clean stage, the embedding is exactly $0^{|H|}$:

\begin{claim}[Complete cleaning]\label{claim:clean}
At time $x$ immediately after processing an entire clean stage (i.e., the last $\lc$ edges were ``clean'' edges), $\embed[x] = 0^{|H|}$.
\end{claim}
\begin{proof}
    Suppose this invariant holds at the end of every complete clean stage before the clean stage we are currently considering (or at initialization, if there is none.) Then, at the start of this clean stage, any coordinate that is nonzero must be exactly $\ell(e_{\mathcal{B}})$ for one of the at most $\lb$ ``build" edges $e_{\mathcal{B}}$ in the immediately preceding build stage---namely, there are at most $\lb$ such nonzero coordinates. We aim to show that there are at least $2 \lb$ edges with weight $\ge \max_{e_{\mcB}}\frac{\ell(e_{\mathcal{B}})}{2}$ in the clean stage, as this would immediately imply that all coordinates were cleaned. To show this, note how, by definition of the modified tree decomposition and the caterpillar property, all but $\log(n)$ of the edges $e_{\mcC}$ in the clean stage must exactly satisfy $\ell(e_{\mathcal{C}}) = \frac{M(s(e_{\mathcal{C}}))}{100 \log(n)}$. Moreover, for any $e_{\mcB}$, $\ell(e_{\mathcal{B}}) \le \frac{M(s(e_{\mathcal{B}}))}{100 \log(n)}$. Then, note how any build edge $e_{\mathcal{B}}$ is within a path of length $\lb + \lc$ of each clean edge $e_{\mathcal{C}}$, and so by \Cref{claim:path-similar}, 
    $$
    \frac{M(s(e_{\mathcal{B}}))}{M(s(e_{\mathcal{C}}))} \le e^{\frac{\lb + \lc}{50 \log(n)}} \le e^{13/50} \le 2.
    $$
    Thus, at least $\lc-\log(n)=8\log(n)\ge 2\lb$ many ``clean" edges have length that is at least half of every build edge $e_{\mcB}$. %
\end{proof}

As the embedding is regularly cleaned, we use this to show that the $\ell_1$ norm of the embedding at any location $x$ is not too large: %

\begin{claim}[Embedding has small norm]
    \label{claim:norm-bound-after-clean}
    At any location $x$, $\|\embed[x]\|_1 \le \frac{M(x)}{2}.$
\end{claim}
\begin{proof}
    By \Cref{claim:clean}, any nonzero coordinate of $\embed[x]$ must correspond to one of the $\lb$ build edges among the previous $\lb+\lc$ edges processed. Any of the at most $\lb$ build edges $e$ must satisfy:
    \begin{align}
        \ell(e)
        & \le \frac{M(s(e))}{100 \log(n)} \\
        & = \frac{M(x)}{100 \log(n)} \cdot \frac{M(s(e)}{M(x)} \\
        & \le \frac{M(x)}{100 \log(n)} \cdot e^{\frac{\lb + \lc}{50 \log(n)}} \label{step:use-path} \\ 
        & \le \frac{M(x)}{100 \log(n)} \cdot e^{13/50} \\ 
        & \le \frac{M(x)}{50 \log(n)} 
    \end{align}

    Here, \Cref{step:use-path} uses \Cref{claim:path-similar}. Accordingly, each of the $\le \lb$ nonzero coordinates of $\embed[x]$ is at most $\frac{M(x)}{50 \log(n)}$, and so,
    \begin{align}
        \|\embed[x]\|_1
        & \le \lb \cdot \frac{M(x)}{50 \log(n)} \\ 
        & \le \frac{M(x)}{2}.
    \end{align}
\end{proof}

Now, we show how paths in the tree relate to distances, in terms of distances to least common ancestors. For any locations $x$ and $y$ in the tree, let $P_{x \rightarrow y}$ denote the path from location $x$ to $y$, and let $\ell(P_{x \rightarrow y})$ denote the total length of the path.

\begin{subclaim}\label{subclaim:leg}
    Fix any two locations $x$ and $y$, and let $L=$LCA$(x,y)$ be the least common ancestor of $x$ and $y$. Then,
    $$\E[\|\embed[x]-\embed[y]\|_1] \ge \Omega(1) \cdot \min(\ell(P_{L \rightarrow x}),M(x)).$$
\end{subclaim}
\begin{proof}
    If $x=L$, then $\ell(P_{L \rightarrow x})=0$, and the bound holds. Otherwise, consider the last $3\log(n)$ edges along the path $P_{L \rightarrow x}$ (or consider all the edges if there are less than $3\log(n)$). Now, consider the event that all these edges are ``build'' edges. Over the randomness of initializing $a(r_0)$, this must happen with probability at least $\frac{\lb - 3\log(n)}{\lb + \lc} = \Omega(1)$. Let us just consider the expected difference in embeddings conditioned on this constant-probability event.

    Let $\ell_P(e)$ denote the length of the edge that lies on the path $P_{L \rightarrow x}$ (this may be different than $\ell(e)$ if only a fraction of $e$ lies along $P_{L \rightarrow x}$). Observe that
    \begin{align*}
        \E[\|\embed[x]-\embed[y]\|_1] &= \sum_{i=1}^{H}\E[|\embed[x](i)-\embed[y](i)|]
    \end{align*}
    Now, for any fixed index $i \in \{1,\dots,H\}$, let $A_i$ be the event that exactly one of the last $k\le 3\log(n)$ edges $e_1,\dots,e_k$ along $P_{L \rightarrow x}$ hash to index $i$, but none of the other edges in the build stages immediately preceding both $x$ and $y$ hash to index $i$. Observe that
    \begin{align*}
        \Pr[A_i] &\ge \sum_{j=1}^{k}\frac{1}{H}\left( 1-\frac{1}{H}\right)^{2\lb}=\frac{k}{H}\left( 1-\frac{1}{H}\right)^{2\lb} = \frac{k}{H}\left(1-\frac{1}{8\log n}\right)^{8\log n}\ge \Omega(1)\cdot \frac{k}{H}.
    \end{align*}
    Furthermore, conditioned on $A_i$, observe that
    \begin{align*}
        \E\left[|\embed[x](i)-\embed[y](i)|~\bigg|~A_i\right] &= \frac{1}{k}\sum_{j=1}^k \ell_P(e_j),
    \end{align*}
    and thus,
    \begin{align*}
        \E[|\embed[x](i)-\embed[y](i)|] &\ge \Pr[A_i]\cdot \E\left[|\embed[x](i)-\embed[y](i)|~\bigg|~A_i\right] \\
        &\ge \Omega(1)\cdot \frac{k}{H}\cdot \frac{1}{k}\sum_{j=1}^k \ell_P(e_j) 
        = \Omega(1)\cdot \frac{1}{H}\sum_{j=1}^k\ell_P(e_j).
    \end{align*}
    Finally,
    \begin{align*}
        \E[\|\embed[x]-\embed[y]\|_1] &= \sum_{i=1}^{H}\E[|\embed[x](i)-\embed[y](i)|]
        \ge \Omega(1)\cdot \sum_{j=1}^k\ell_P(e_j).
    \end{align*}
    Thus, the claim is already proven by the above argument if there at most $k\le 3\log(n)$ edges on $P_{L \rightarrow x}$. If not, there are $> 3 \log(n)$ edges on $P_{L \rightarrow x}$, and if we look at the last $k=3 \log(n)$ of them, then there is at most one fractional edge (close to $x$) amongst these---all the rest are fully on $P$ i.e., $\ell_P(e)=\ell(e)$. Furthermore, by the caterpillar property, all but $\log(n)$ of these edges have length
    \begin{align}
        \ell_P(e_j) &= \frac{M(s(e))}{100 \log(n)} \\
        &= \frac{M(x)}{100 \log(n)} \cdot \frac{M(s(e))}{M(x)} \\
        & \ge \frac{M(x)}{100 \log(n)} \cdot e^{-2/50} \label{step:lb-by-path}\\
        & \ge \frac{M(x)}{200 \log(n)}
    \end{align}
    \Cref{step:lb-by-path} uses \Cref{claim:path-similar}. To conclude, 
    \begin{align*}
        \sum_{j=1}^k \ell_P(e) \ge (3\log(n)-1-\log(n)) \cdot \frac{M(x)}{200 \log(n)} = M(x)\left(\frac{1}{100} - \frac{1}{200 \log(n)}\right) = \Omega(M(x))
    \end{align*}
    for $n \ge 2$. In any case, we have shown
    \begin{align*}
        \E[\|\embed[x]-\embed[y]\|_1] \ge \Omega(1) \cdot \min(\ell(P_{L \rightarrow x}),M(x)).
    \end{align*}
\end{proof}

We are now ready to show a guarantee that in the special case that all caps are strictly positive, for every pair of nodes, the Build-Clean-Tree algorithm does not overestimate their distance, and also covers at least a constant fraction of it in expectation:
\begin{lemma}[Strictly positive caps]
    \label{lemma:tree-lipschitz-cap-lb}
    Let $embedding$ be the output of the Build-Clean-Tree algorithm, with $embedding$ of length $|H|=8\log(n)$, \underline{where all $M(x)>0$}. For every fixed pair of nodes $x_i$ and $x_j$,
    \begin{align*}
    &\text{(1)  } \Pr\left[\|\embed[i]-\embed[j]\|_1 \le \dlcap{M}(x_i, x_j)\right] =1.\\
    &\text{(2)  } \E\left[\|\embed[i]-\embed[j]\|_1\right]  \ge \Omega(1)\cdot \dlcap{M}(x_i, x_j).
    \end{align*}
\end{lemma}
\begin{proof}
     We first prove part(1). Because of the way the build-clean procedure works, the embeddings as we move continuously along the tree are coordinate-wise Lipschitz, and only ever capture distance in one coordinate along adjacent locations. This immediately gives us that $\|\embed[i]-\embed[j]\|_1 \le d(x_i,x_j)$. What remains is to show that $\|\embed[i]-\embed[j]\|_1 \le \max\left(M(x_i),M(x_j)\right)$. This follows from \Cref{claim:norm-bound-after-clean}, because \begin{align*}
        \|\embed[x] - \embed[y]\|_1 &\le \|\embed[x]\|_1 + \|\embed[y]\|_1 \\
        &\le \frac{M(x) + M(y)}{2} \le \max(M(x),M(y)).
    \end{align*}
    Next, we turn our attention towards proving part (2). Without loss of generality, let us assume $M(i) \ge M(j)$---if this is not the case, we simply swap $x_i$ and $x_j$.
    
    \textbf{Case 1:} $\frac{M(i)}{M(j)} \le 2$. In this case, the caps of $x_i$ and $x_j$ are similar. Note how then:
    \begin{align}
        \dlcap{M}(x_i, x_j) &=\min(\ell(P_{i \rightarrow L})+\ell(P_{L \rightarrow j}),\max(M(i),M(j))) \\
        &=\min(\ell(P_{i \rightarrow L})+\ell(P_{L \rightarrow j}),M(i)) \label{eqn:case1-start}\\
        & \le \min(\ell(P_{i \rightarrow L}),M(i)) + \min(\ell(P_{L \rightarrow j}),M(i)) \\
        & \le \min(\ell(P_{i \rightarrow L}), M(i)) + \min(\ell(P_{L \rightarrow j}),2 \cdot M(j)) \label{eqn:mi-mj-bound}\\
        &\le \min(\ell(P_{i \rightarrow L}), M(i)) + 2\min(\ell(P_{L \rightarrow j}), M(j)).\label{eqn:case1-end}
    \end{align}
    In \Cref{eqn:mi-mj-bound}, we used the assumption that $M(i)\le 2M(j)$ under the present case. Accordingly then, we can just show that $E[\|\embed[i]-\embed[j]\|_1] \ge \Omega(1) \cdot \min(\ell(P_{i \rightarrow L}), \cdot M(i)), \Omega(1) \cdot \min(\ell(P_{L \rightarrow j}),\cdot M(j))$. This immediately follows from \Cref{subclaim:leg}.

    \textbf{Case 2:} $\ell(P_{L \rightarrow i}) \ge \frac{M(i)}{100}$. In this case,  by \Cref{subclaim:leg}, we know 
    $$E[\|\embed[i]-\embed[j]\|_1] \ge \Omega(1) \cdot \min(\ell(P_{L \rightarrow i}), M(i)) \ge \Omega(1)\cdot\frac{M(i)}{100} \ge \Omega(1)\cdot \dlcap{M}(x_i,x_j),$$
    where in the last inequality, we used that $M(i) \ge \dlcap{M}(x_i, x_j).$
    
    \textbf{Case 3:} $\frac{M(i)}{M(j)} >2$ and $\ell(P_{L \rightarrow i}) < \frac{M(i)}{100}$. This is the only case that remains. Our proof for this case will focus on showing (i) there must be the the end of a clean stage within $P_{L \rightarrow j}$, (ii) given that there is a clean stage that ends in $P_{L\rightarrow j}$, $\E[\|\embed[i]-\embed[j]\|_1] \ge \Omega(1) \cdot \E[\|\embed[i]\|_1]$, and finally, (iii) $\E[\|\embed[i]\|_1] = \Omega(1) \cdot M(i)$---together these imply that $\E[\|\embed[i]-\embed[j]\|_1] = \Omega(1) \cdot M(i) = \Omega(1) \cdot \dlcap{M}(i,j)$, as required.

    For (i):
    \begin{subclaim}\label{subclaim:must-clean}
        $P_{L \rightarrow j}$ contains at least $\lb + \lc$ edges, and hence, the end of a clean stage.
    \end{subclaim}
    \begin{proof}
        Note how 
        \begin{align*}
        \ell(P_{L \rightarrow j}) &= \ell(P_{i \rightarrow j})-\ell(P_{L \rightarrow i}) \\
        &\ge |M(i)-M(j)| - \ell(P_{L \rightarrow i}) \\
        &\ge \frac{M(i)}{2} - \frac{M(i)}{100} = \frac{49M(i)}{100} \\
        &> \frac{49M(j)}{50}.
        \end{align*}
        Now, there \textit{must} be the end of a clean stage in $P_{L \rightarrow j}$ if $P_{L \rightarrow j}$ contains at least $\lb + \lc$ edges. We will show that if the number of edges in $P_{L \rightarrow j}$ is smaller than $\lb +\lc$, then such a small number of edges could not possibly produce such a large $\ell(P_{L \rightarrow j})$. In particular, if so, by \Cref{claim:path-similar}, each edge $e$ of the $< \lb + \lc$ edges nearest to $x_j$ in the direction of $L$ must be of length 
        \begin{align*}
            \ell(e) \le \frac{M(j)}{100 \log(n)} \cdot \frac{M(s(e))}{M(j)} \le \frac{M(j)}{100 \log(n)} \cdot e^{\frac{\lb + \lc}{50 \log(n)}} \le \frac{M(j)}{100 \log(n)} \cdot e^{13/50} \le \frac{M(j)}{50}.
        \end{align*}
        The total length of $< \lb+\lc$ such edges on $P_{L \rightarrow j}$ would then be bounded as 
        $$\ell(P_{L \rightarrow j}) < (\lb + \lc) \cdot \ell(e) \le 13 \log(n) \cdot \frac{M(j)}{50 \log(n)} \le \frac{13}{50} \cdot M(j) < \frac{49}{50} \cdot M(j) < \ell(P_{L \rightarrow j}),$$ causing a contradiction.
    \end{proof}

    For (ii):
    \begin{subclaim}\label{subclaim:clean-helps}
        If $P_{L \rightarrow j}$ contains the end of a clean stage, then $$\E[\|\embed[i]-\embed[j]\|_1] \ge \Omega(1) \cdot \E[\|\embed[i]\|_1].$$
    \end{subclaim}
    \begin{proof}
        Note how, because there is a clean stage ending in $P_{L \rightarrow j}$, the only nonzero coordinates of $\embed[j]$ must be coordinates hashed into by one of the $\le \lb$ build edges after the \textit{last }clean. Note also that these hash values are independent of $\embed[i]$, and there are at most $\lb = 4 \log(n) \le \frac{|H|}{2}$ such coordinates. Let $Z_k$ denote the event that none of these build edges hash into the $k^{\text{th}}$ coordinate. Observe that by the union bound, $\Pr[Z_k] \ge \frac12$. Conditioned on $Z_k$, observe that $\embed[i](k)-\embed[j](k)=\embed[i](k)$. Accordingly,

        \begin{align}
            \E[\|\embed[i]-\embed[j]\|_1] &= \sum_{k=1}^{H}\E[|\embed[i](k)-\embed[j](k)|] \\
            & \ge \sum_{k=1}^H \Pr(Z_k) \cdot \E\left[|\embed[i](k)-\embed[j](k)| ~\bigg|~ Z_k\right] \\
            & = \sum_{k=1}^H \Pr(Z_k) \cdot \E\left[|\embed[i](k)| ~\bigg|~ Z_k\right] \\
            & = \sum_{k=1}^H \Pr(Z_k) \cdot \E[\|\embed[i](k)\|_1] \qquad(\text{independence})\\
            & \ge \sum_{k=1}^H \frac{1}{2} \cdot E[|\embed[i](k)|] \\
            & = \frac{1}{2} \cdot \E[\|\embed[i]\|_1].
        \end{align}
    \end{proof}

    For (iii), the following lets us say that the expected norm of any node's embedding is within a constant factor of its cap. Note how this would not be true for the root $r$ if it had a nonzero cap and we had just initialized the root's embedding as all zeros. This is the reason behind our initialization with an extra $r_0$ node.
    \begin{subclaim}\label{subclaim:big-norm}
        $\E[\|\embed[i]\|_1] = \Omega(1) \cdot M(i)$.
    \end{subclaim}
    \begin{proof}
        This follows from invoking \Cref{subclaim:leg} with $x=x_i$ and $y=r_0$ (and thus, $L=LCA(x_i, r_0)=r_0$), and recalling that $\embed[r_0] = 0^{|H|}$:
        \begin{align}
            \E[\|\embed[i]\|_1] & = \E[\|\embed[i]-\embed[r_0]\|_1]\\
            & \ge \Omega(1) \cdot \min(d(r_0,x),M(x))\\
            &= \Omega(1) \cdot \min(2M(r)+d(r, x_i),M(x)) \\
            & \ge \Omega(1) \cdot \min(2 M(r) + |M(r) - M(x)|,M(x)) \qquad (\text{by Lipschitzness})\\
            & \ge \Omega(1) \cdot M(x),
        \end{align}
        where the last inequality follows because $2M(r)+|M(r)-M(x)| \ge \frac{M(x)}{2}$, which can be seen by considering two cases $M(r) \ge \frac{M(x)}{2}$ and $M(r) < \frac{M(x)}{2}$.
    \end{proof}
    Thus, as \Cref{subclaim:must-clean} implies the condition required by \Cref{subclaim:clean-helps}, combining with \Cref{subclaim:big-norm} yields $\E[\|\embed[i]-\embed[j]\|_1] \ge \Omega(1) \cdot M(i) \ge \Omega(1) \cdot \dlcap{M}(i,j)$.
    Thus, we have proven part(2) of \Cref{lemma:tree-lipschitz-cap-lb} in its entirety.
\end{proof}

Finally, we obtain the general result for the setting which allows for $M(x)$ to be zero as well.

\begin{lemma}
    \label{lemma:general-tree-lipschitz-cap-lb}
    There exists an embedding of dimension $O(\log(n))$ where for every fixed pair of nodes $x_i$ and $x_j$,
    \begin{align*}
    &\text{(1)  } \Pr\left[\|\embed[i]-\embed[j]\|_1 \le \dlcap{M}(x_i, x_j)\right] =1.\\
    &\text{(2)  } \E\left[\|\embed[i]-\embed[j]\|_1\right]  \ge \Omega(1)\cdot \dlcap{M}(x_i, x_j).
    \end{align*}
\end{lemma}
\begin{proof}
    Note how this immediately follows from \Cref{lemma:tree-lipschitz-cap-lb} if all $M(x)>0$. Otherwise, there are some locations $x$ in the tree where $M(x)=0$. We will define an embedding $\embed[i]$ as the concatenation of two embeddings $\embed_{\textrm{nonzero}}[i]$ and $\embed_{\textrm{single-cord}}[i]$. We define the embedding as follows:

    \begin{itemize}
        \item For each $x_i$ where $M(x_i)=0$ , set $\embed_{\textrm{nonzero}}[i] = 0^{|H|}$ and $\embed_{\textrm{single-cord}}[i] = 0$.
        \item Now, consider the tree where we delete all $x_i$ satisfying $M(x_i)=0$. With the remaining forest, run the special-case embedding algorithm proven in \Cref{lemma:tree-lipschitz-cap-lb} for each resulting connected component $T_k$, and use this embedding \underline{\emph{with each coordinate divided by $2$}} for each $\embed_{\textrm{nonzero}}[i]$. Also, for each connected component, define a Rademacher random variable $C_k$ (that is $+1$ with probability $\frac{1}{2}$ and $-1$ with probability $\frac{1}{2}$). For each $x_i$ in a connected component, set $\embed_{\textrm{single-cord}}[i] = C_k \cdot \frac{M(i)}{4}$.
    \end{itemize}

    To show (1), we first show that the norm of $\embed[i]$ is bounded:
    
    \begin{align}
        \|\embed[i]\|_1
        & = \|\embed_{\textrm{nonzero}}\|_1 + \|\embed_{\textrm{single-cord}}\|_1 \\
        & \le \frac{1}{2} \cdot \frac{M(i)}{2} + \frac{M(i)}{4} \label{step:use-norm}\\
        & = \frac{M(i)}{2} \label{eqn:total-embedding-norm-bound}.
    \end{align}
    \Cref{step:use-norm} uses \Cref{lemma:tree-lipschitz-cap-lb} which shows that the nonzero embedding of $x_i$ satisfies that its norm is at most $\frac{M(i)}{2}$. Accordingly, 
    $$
    \|\embed[i]-\embed[j]\|_1 \le \|\embed[i]\|_1+\|\embed[j]\|_1 \le \frac{M(i)+M(j)}{2} \le \max(M(i),M(j)).
    $$
    What remains is to show $\|\embed[i]-\embed[j]\|_1 \le d(i,j)$:

    \textbf{Case 1: $x_i,x_j$ are in the same connected component and $M(i)>0, M(j)>0$.}
    \begin{align}
        & \|\embed[i] - \embed[j]\|_1 \\
        & = \|\embed_{\textrm{nonzero}}[i]-\embed_{\textrm{nonzero}}[j]\|_1 + \|\embed_{\textrm{single-cord}}[i] - \embed_{\textrm{single-cord}}[j]\|_1 \\
        & \le \frac{1}{2} d(i,j) + \frac{1}{4} \cdot |M(i)-M(j)| \label{eqn:same-component}\\
        & \le \frac{1}{2} d(i,j) + \frac{1}{4} \cdot d(i,j) \label{eqn:same-component-lipschitzness}\\
        & \le d(i,j).
    \end{align}
    \Cref{eqn:same-component} used the fact that $\embed_{\textrm{single-cord}}[i]$ and $\embed_{\textrm{single-cord}}[j]$ have the same sign by virtue of being in the same component, and \Cref{eqn:same-component-lipschitzness} used Lipschitzness of the cap. 

    \textbf{Case 2: $x_i,x_j$ are in different connected components and $M(i)>0, M(j)>0$.}
    \begin{align}
        & \|\embed[i] - \embed[j]\|_1 \\
        & = \|\embed_{\textrm{nonzero}}[i]-\embed_{\textrm{nonzero}}[j]\|_1 + \|\embed_{\textrm{single-cord}}[i] - \embed_{\textrm{single-cord}}[j]\|_1 \\
        & \le \|\embed_{\textrm{nonzero}}[i]\|_1 +\|\embed_{\textrm{nonzero}}[j]\|_1 + \|\embed_{\textrm{single-cord}}[i]\|_1 + \|\embed_{\textrm{single-cord}}[j]\|_1 \\ 
        & \le \frac{1}{4} \cdot (M(i)+M(j)) + \frac{1}{4} \cdot (M(i)+M(j))\\
        & \le \frac{1}{4}\cdot d(i,j) + \frac{1}{4} \cdot d(i,j)  \label{step:add}\\
        & \le d(i,j).
    \end{align}
    \Cref{step:add} follows from $x_i$ and $x_j$ being in different components meaning that there must be a node $x_k$ such that $M(k)=0$ and $x_k$ is on the path from $x_i$ to $x_j$, and so $d(i,j) = d(i,k) + d(k,j) \ge |M(i)-M(k)| + |M(k)-M(j)| = M(i) + M(j)$. 
    
    \textbf{Case 3: Either $M(i)=0$ or $M(j)=0$.}
    
    If both $M(i)=0$ and $M(j)=0$, $\|\embed[i]-\embed[j]\|_1=0 \le d(i,j).$
    Else say $M(i)>0$ and $M(j)=0$. Then, 
    \begin{align*}
    \|\embed[i]-\embed[j]\|_1 &= \|\embed[i]\|_1 \\
    &\le \frac{M(i)}{2} \qquad (\text{using \eqref{eqn:total-embedding-norm-bound}})\\
    &= \frac{M(i)-M(j)}{2} \\
    &\le d(i,j) \qquad(\text{by Lipschitzness}).
    \end{align*}
    This concludes the proof of (1).

    To prove (2), we must show that the expected difference in the embeddings is at least a constant factor of the true capped distance. This immediately holds for any $x_i,x_j$ in the same component (and having nonzero caps) by \Cref{lemma:tree-lipschitz-cap-lb} just from their coordinates in $\embed_{\textrm{nonzero}}$. For any $x_i,x_j$ in different components, and having $M(i), M(j)>0$, note how we argued above that $d(i,j)\ge M(i)+M(j) \ge \max(M(i), M(j))$, and hence $\dlcap{M}(x_i, x_j) = \max(M(i),M(j))$. Thus, the expected difference 
    \begin{align*}
        \E[\|\embed[i]-\embed[j]\|_1] &\ge \E[\|\embed_{\textrm{nonzero}}[i]-\embed_{\textrm{nonzero}}[j]\|_1] \\
        &\ge \frac{1}{2}\cdot\frac{1}{4} \cdot (M(i)+M(j))) \\
        &= \Omega(1) \cdot \max(M(i),M(j)).
    \end{align*}
    If exactly one of the caps is nonzero---say $M(i)>0$ and $M(j)=0$, then again, we have $d(i,j) \ge M(i)-M(j)=M(i)=\max(M(i),M(j))$, and hence $\dlcap{M}(i,j)=M(i)$, giving
    \begin{align*}
        \E[\|\embed[i]-\embed[j]\|_1] &= \E[\|\embed[i]\|_1 \ge \E[\|\embed_{\textrm{nonzero}}[i]\|_1] = \frac{M(i)}{4}.
    \end{align*}
    Finally, if both $M(i)=0$ and $M(j)=0$, then $\dlcap{M}(i,j)=0$, but $\E[\|\embed[i]-\embed[j]\|_1]=0$ as well.
    This concludes the proof of (2).
\end{proof}
Boosting the above in-expectation guarantee of \Cref{lemma:general-tree-lipschitz-cap-lb} using logarithmically many copies then yields \Cref{thm:lipschitz-capped-tree}.

\theoremlipschitzcappedtree*
\section{General Capped \texorpdfstring{$\ell_1$}{l1} Metrics}
\label{sec:general-cap-metrics}
The Build-Clean-Tree algorithm from above has some nice properties that lend themselves to constructing embeddings for capped $\ell_1$ metrics more generally. Concretely, the task here is the following: we are given a set of $n$ points in $\R^d$ with the distance being the $\ell_1$ metric. For a fixed cap $M>0$, we want to construct an embedding of these points which captures the $\ell_1$ distance between these points capped at $M$. Concretely, we want to construct $\embed[x_1],\dots,\embed[x_n]$ such that for any $x_i, x_j$,
\begin{align*}
    \Omega(1)\cdot \min(\|x_i -x_j\|_1, M) \le \|\embed[x_i]-\embed[x_j]\|_1 \le O(1) \cdot \min(\|x_i-x_j\|_1, M).
\end{align*}
We can use the Build-Clean-Algorithm as a subroutine to construct such an embedding.
\theoremgeneralcap*

The first step towards this is interpreting each coordinate of the points as a line metric, and embedding that line metric with the cap $M$ using the Build-Clean-Algorithm.
\begin{restatable}[Coordinate-wise build-clean]{lemma}{coordinatewisebc}
    \label{lemma:coordinate-wise-build-clean}
    Let $x_1,\dots,x_n$ be points embedded in $(\R^d, \ell_1)$. Fix any coordinate $q\in \{1,\dots,d\}$, and consider the corresponding coordinates $x_{1}[q], \dots, x_n[q]$. Then, there exists an embedding $v^{(q)}_{1},\dots, v^{(q)}_{n} \in \R^{8d}$ of these coordinates, such that for any $i,j$,
    \begin{align*}
        &\text{(1)  } \Pr\left[\|v^{(q)}_{i}-v^{(q)}_{j}\|_1 \le \min(|x_{i}[q]-x_{j}[q]|, M)\right] =1.\\
        &\text{(2)  } \E\left[\|v^{(q)}_{i}-v^{(q)}_{j}\|_1\right]  \ge \Omega(1)\cdot \min(|x_{i}[q]-x_{j}[q]|, M).
    \end{align*}
    Furthermore, every $v^{(q)}_{i}$ has the property that all its coordinates are exactly $0$ or $\frac{M}{100d}$, except at most one coordinate which is contained in $\left[0, \frac{M}{100d}\right]$.
\end{restatable}
\begin{proof}[Proof sketch]
    Consider say the first coordinate of each of the $n$ points. If we map this coordinate to the line, this corresponds to a line graph metric. We will split up this line metric into segments of size $\frac{M}{100d}$ and embed it with a fixed cap of $M$ into $8d$ dimensions using the build-clean framework described in \Cref{sec:lipschitz-tree}.\footnote{Observe that we have effectively replaced ``$\log(n)$" with ``$d$" in the parameters in the analysis for Lipschitz-cap tree metrics.} The special properties of the building and cleaning processes ensure that for every point, at most one coordinate of the embedding is not either exactly 0 or $\frac{M}{100d}$. Furthermore, we are guaranteed to never overestimate the capped distance, and also capture a constant fraction of it in expectation, following similar reasoning as in the proof of \Cref{lemma:tree-lipschitz-cap-lb}. Complete details about the proof are given in \Cref{sec:proofs-coordinate-wise}.
\end{proof}
For each $x_i$, we can then concatenate the coordinate-wise embeddings given by \Cref{lemma:coordinate-wise-build-clean} independently to obtain an embedding $z_i = (v^{(1)}_i, \dots, v^{(d)}_i) \in \R^{8d^2}$. For the $z_i$ vectors thus constructed, we have the following convenient lemma:
\begin{restatable}[$z_i$'s capture distance]{lemma}{lemmazigood}
    \label{lemma:zi-good}
    For each $i \in \{1,\dots,n\}$, let $z_i$ be the concatenation of the coordinate-wise build-clean embeddings of $x_i$. Then, for any fixed $i,j$,
    \begin{align*}
        \Pr\left[\|z_i-z_j\|_1 > \Omega(1) \cdot \dcap{M}(x_i, x_j)\right] = \Omega(1).
    \end{align*}
    Additionally, $z_i$ has the special property that all its coordinates are exactly $0$ or $\frac{M}{100d}$, except at most $d$ coordinates which are contained in $\left[0, \frac{M}{100d}\right]$.
\end{restatable}
\begin{proof}[Proof sketch]
    From \Cref{lemma:coordinate-wise-build-clean}, we know that the embedding for each individual coordinate, in expectation, captures a constant fraction of the distance in that coordinate. Furthermore, the embeddings for the different coordinates are constructed independently of each other. We use these two facts and apply an argument that is morally Markov's inequality in reverse, to obtain that with at least a constant probability, the $z_i$ vectors faithfully capture a constant factor of the capped distance. Complete details about the proof are given in \Cref{sec:proofs-for-zi}.
\end{proof}

The $z_i$ vectors enjoy similar structural properties like those in \Cref{claim:sparsity-less-than-M} and \Cref{claim:sparsity-greater-than-M}. These allow us to use similar hashing+snaking techniques that we used in the proofs of \Cref{claim:fixed-cap-tree-no-overestimate}.
and \Cref{lemma:fixed-cap-tree-no-underestimate}. More concretely, let $k=C\cdot d$ for an appropriately chosen constant $C$, and consider choosing a uniformly random hash function $h:[8d^2]\to [k]$ to hash the coordinates of the $z_i$'s into $k$ buckets. That is, for each $z_i$, $H(z_i)$ is a vector of size $Cd$, whose coordinates are defined in the following:
\begin{align}
    H(z_i)[p]=\sum_{q \in [8d^2]:h(q)=p}z_i[q] \qquad \text{for } p \in \{1,\dots,k\}.
\end{align}
Now, we interpret each coordinate of $H(z_i)$ as defining a line metric, over which we will do lazy snaking (\Cref{algo:lazy-snake}). Concretely, for each $p \in [k]$, let $loc_p[i] = H(z_i)[p]$. Let $sloc_p=\textsc{LazySnake}(loc_p, M/d)$, and consider the final embedding of each point $z_i$ as 
\begin{equation}
    \label{eqn:fixed-cap-general-l1-embedding-def}
    embedding[i] = (sloc_1[i], sloc_2[i], \dots, sloc_k[i]).
\end{equation} 

We can show that the embedding constructed as above does not overestimate and underestimate distances upto constant factors.
\begin{claim}[No overestimation]
    \label{claim:fixed-cap-general-l1-no-overestimate}
    Fix any pair $x_i$ and $x_j$. Then,
     with probability 1,
     \begin{align*}
        \|embedding[i]-embedding[j]\|_1 \le \Omega(1)\cdot\dcap{M}(x_i, x_j).
     \end{align*}
\end{claim}
\begin{proof}
    We have two cases: \\
    \noindent Case 1: $\|x_i-x_j\|_1 \le M.$ \\
    In this case, $\dcap{M}(x_i, x_j) = \|x_i-x_j\|_1$. Furthermore, for any $q\in \{1,\dots,d\}$, $|x_i[q]-x_j[q]|\le M$. Thus, from \Cref{lemma:coordinate-wise-build-clean}, $\|v^{(q)}_i-v^{(q)}_j\|_1 \le |x_i[q]-x_j[q]|$ with probability 1. This means that 
    \begin{align}
        \|z_i - z_j\|_1 &= \sum_{q=1}^d \|v^{(q)}_i-v^{(q)}_j\|_1 \le \sum_{q=1}^d|x_i[q]-x_j[q]| = \|x_i-x_j\|_1. \label{eqn:zi-bound}
    \end{align}
    Observe that
    \begingroup
    \allowdisplaybreaks
    \begin{align*}
        \|embedding[i]-embedding[j]\|_1 &= \sum_{p=1}^k\left|sloc_p[i]-sloc_p[j]\right| \\
        &\le \sum_{p=1}^k\left|loc_p[i]-loc_p[j]\right| \quad \text{(snaking never overestimates distances)}\\
        &= \sum_{p=1}^k\left|\sum_{q \in [8d^2]:h(q)=p}(z_i[q] - z_j[q])\right| \\
        &\le \sum_{p=1}^k\sum_{q \in [8d^2]:h(q)=p}\left|z_i[q] - z_j[q]\right| \\
        &= \sum_{q=1}^{8d^2}\left|z_i[q] - z_j[q]\right| \\
        &=\|z_i-z_j\|_1 \\
        &\le \|x_i-x_j\|_1. \qquad (\text{\Cref{eqn:zi-bound}})
    \end{align*}
    \endgroup
    \noindent Case 2: $\|x_i-x_j\|_1 > M.$ \\
    In this case, $\dcap{M}(x_i, x_j) = M$. We have
    \begin{align*}
        \|embedding[i]-embedding[j]\|_1 &= \sum_{p=1}^k\left|sloc_p[i]-sloc_p[j]\right| \\
        &\le \sum_{p=1}^k \frac{M}{d} \qquad \left(\text{snaking width is }\frac{M}{d}\right)\\
        &\le CM.
    \end{align*}
\end{proof}

\begin{lemma}[No underestimation]
    \label{lemma:fixed-cap-general-l1-no-underestimate}
    Fix any pair $x_i$ and $x_j$, and fix $p \in [k]$. Then, we have that
    \begin{align*}
        \E\left[\left|sloc_p[i]-sloc_p[j]\right|\right] &= \E_{\bc}\E_{h}\E_{\snake}\left[\left|sloc_p[i]-sloc_p[j]\right|\right] \ge \Omega(1) \cdot \frac{\dcap{M}(x_i, x_j)}{k}.
    \end{align*}
    Thus, by linearity of expectation,
    \begin{align*}
        \E\left[\|embedding[i]-embedding[j]\|_1\right] = \E\left[\sum_{p=1}^k\left|sloc_p[i]-sloc_p[j]\right|\right] \ge \Omega(1)\cdot \dcap{M}(x_i, x_j).
    \end{align*}
\end{lemma}
\begin{proof}
    We have two cases: \\
    \noindent Case 1: $\|x_i-x_j\|_1 \le M.$ \\
    In this case, $\dcap{M}(x_i, x_j) = \|x_i-x_j\|_1$.  Furthermore, from \Cref{eqn:zi-bound}, we know that with probability 1,
    $$
    \|z_i -z_j\|_1 \le \|x_i-x_j\|_1 \le M.
    $$
    This implies that $z_i-z_j$ has at most $102d$ nonzero coordinates, and each of these coordinates is at most $\frac{M}{100d}$. To see this, observe that by the structural properties of each $z_i$ (\Cref{lemma:zi-good}), $z_j-z_j$ has at most $2d$ coordinates that are contained in $\left[0, \frac{M}{100d}\right]$---the rest of the coordinates are either exactly $0$ or $\frac{M}{100d}$. But note also that since $\|z_i-z_j\|_1 \le M$, there can only be at most $100d$ coordinates that are equal to $\frac{M}{100d}$, and the rest must be 0.
    
    Thus, under this case, we have argued that $z_i-z_j$ has a set of at most $k'=102d$ nonzero coordinates and each of these coordinates is at most $\frac{M}{100d}$ in magnitude. Let us possibly include some zero coordinates, so that we have exactly $k'=102d$ of these ``special'' coordinates $c_1,\dots, c_r, \dots,c_{k'}$. By \Cref{lemma:zi-good}, we have that with constant probability over the build-clean process, 
    \begin{align}
        \|z_i - z_j\|_1 = \sum_{r=1}^{k'}|z_i[c_r]-z_j[c_r]| > \Omega(1) \cdot \|x_i-x_j\|_1 \label{eqn:zi-good-again}
    \end{align}
    Let us first condition on this constant-probability event over the build-clean process. Next, let us condition on the realization of the hash function $h$, which is independent of the randomness in the snaking. Conditioned on this realization, for $loc_p[i]=H(z_i)[p]$ and $sloc_p=\textsc{LazySnake}(loc_p, M/d)$ we have from \Cref{claim:lazy-snake-line-fixed-cap-lb} that
    \begin{align*}
        \E_{\snake}\left[|sloc_p[i]-sloc_p[j]|\right] \ge \Omega(1) \cdot \dcap{M/d}(loc_p[i], loc_p[j]).
    \end{align*}
    Now, taking an expectation with respect to the choice of the hash function, we get
    \begin{align*}
        \E_h\E_{\snake}\left[|sloc_p[i]-sloc_p[j]|\right] \ge \Omega(1) \cdot \E_h\left[\dcap{M/d}(loc_p[i], loc_p[j])\right].
    \end{align*}
    Let $A$ be the event that only one of the $k'$ special coordinates hashes to $p$. Then, we have that $\Pr[A]=\left(1-\frac1k\right)^{k'-1} \ge \Omega(1)$, yielding
    \begin{align*}
        \E_h\left[\dcap{M/d}(loc_p[i], loc_p[j])\right] &\ge \Omega(1)\cdot\E_h\left[\dcap{M/d}(loc_p[i], loc_p[j]) ~|~ A\right]
    \end{align*}
    Furthermore, conditioned on $A$, we have
    \begin{align*}
        &\E_h\left[\dcap{M/d}(loc_p[i], loc_p[j]) ~|~ A\right] \\
        &\qquad=\sum_{r=1}^{k'} \frac{1}{k'} \cdot \E_h\left[\dcap{M/d}(loc_p[i], loc_p[j]) ~|~ \text{only $r^{\text{th}}$ special coordinate hashes to $p$}\right].
    \end{align*}
    But now, observe that if only the $r^{\text{th}}$ special coordinate hashes to $p$, we have $\dcap{M/d}(loc_p[i], loc_p[j]) =|z_i[c_r]-z_j[c_r]|$. This is because $|z_i[c_r]-z_j[c_r]|$ is at most $M/100d$, which is smaller than the cap $\frac{M}{d}$. Finally, recalling \Cref{eqn:zi-good-again},
    \begin{align*}
        \E_h\left[\dcap{M/d}(loc_p[i], loc_p[j]) ~|~ A\right] &= \sum_{r=1}^{k'}\frac{1}{k'}\left|z_i[c_r]-z_j[c_r]\right| \ge \Omega(1)\cdot \frac{1}{k'}\cdot \|x_i-x_j\|_1 \ge \Omega(1)\cdot \frac{1}{k}\cdot \|x_i-x_j\|_1. 
    \end{align*}
    Putting everything together, we get
    \begin{align*}
        \E_{\bc}\E_h\E_{\snake}\left[|sloc_p[i]-sloc_p[j]|\right] &\ge \Omega(1) \cdot \frac{\|x_i-x_j\|_1}{k}.
    \end{align*}
    
    \noindent Case 2: $\|x_i-x_j\|_1 > M.$ \\
    In this case, $\dcap{M}(x_i, x_j)=M$. First, by \Cref{lemma:zi-good}, we have that with constant probability over the build-clean process, $\|z_i - z_j\|_1 > cM$ for some constant $c$. We condition on this constant probability event. Given that this event holds, let us count the number of coordinates in $z_i-z_j$ that have magnitude smaller than $\frac{cM}{100d}$. Recall that by the properties of the build-clean algorithm (\Cref{lemma:zi-good}), there are at most $2d$ ``not-full" coordinates in $z_i-z_j$ that are neither exactly 0 or $\frac{M}{100d}$, and are in $\left[0, \frac{M}{100d}\right]$. Thus, coordinates that are smaller than $\frac{cM}{100d}$ must lie in this set of at most $2d$ not-full coordinates, and their contribution to $\|z_i-z_j\|_1$ is no more than $2d\cdot \frac{cM}{100d}=\frac{cM}{50}$. Since $\|z_i-z_j\|_1$ is at least $cM$, there is at least $\frac{49cM}{50}$ distance left, and this must be contributed to by the set of ``full" coordinates that is exactly $\frac{M}{100d}$ and the not-full coordinates that are at least $\frac{cM}{100d}$. The total number of these coordinates is at the very lest $\frac{49cM}{50}\cdot \frac{100d}{M}=98cd=\Omega(d)$, and their sum is at least $\frac{49cM}{50}=\Omega(M)$.
    
    Thus, under the conditioned event, we know that $z_i-z_j$ has a set of $k'=\Omega(d)$ ``special" coordinates $c_1,\dots,c_r,\dots,c_{k'}$ such that the sum of their absolute values is $\Omega(M)$. Fix $p \in \{1,\dots,k\}$, and let $A$ be the event that the non-special coordinates that get hashed to $p$ amount for a distance of at least $\frac{M}{C_1d}$ for a suitable constant $C_1$. Concretely, under $A$, $\left|\sum_{q \in [8d^2]:q \text{ not special},h(q)=p}(z_i[q]-z_j[q])\right| \ge \frac{M}{C_1d}$.
    We have that
    \begin{align*}
        \E\left[|sloc_p[i]-sloc_p[j]|\right] &= \Pr[A]\cdot \E\left[|sloc_p[i]-sloc_p[j]| ~|~ A\right] \\
        &\quad + \Pr[\neg A]\cdot \E\left[|sloc_p[i]-sloc_p[j]| ~|~ \neg A\right]
    \end{align*}
    Conditioned on $A$, we are happy even if none of the special coordinates hash to $p$, which happens with probability $\left(1-\frac{1}{k}\right)^{k'}  \ge \Omega(1)$. In this case, by \Cref{claim:lazy-snake-line-fixed-cap-lb}, lazy snaking will capture at least a constant fraction of $\frac{M}{C_1d}$, yielding
    \begin{align*}
        \E\left[|sloc_p[i]-sloc_p[j]| ~|~ A\right] &\ge \Omega(1) \cdot \frac{M}{d}.
    \end{align*}
    If $A$ does not occur, we have that $\left|\sum_{q \in [8d^2]:q \text{ not special},h(q)=p}(z_i[q]-z_j[q])\right| < \frac{M}{C_1d}$. In this case, we are happy if exactly one of the special coordinates hashes to $p$, which happens with probability $\left(1-\frac{1}{k}\right)^{k'-1} \ge \left(1-\frac{1}{k}\right)^{k'} \ge \Omega(1)$. This will ensure a distance of at least $\frac{cM}{100d} - \frac{M}{C_1d}=\Omega(1)\cdot\frac{M}{d}$. Lazy snaking will yet again capture at least a constant fraction of this distance in expectation, yielding
    \begin{align*}
        \E\left[|sloc_p[i]-sloc_p[j]| ~|~ \neg A\right] &\ge \Omega(1) \cdot \frac{M}{d}.
    \end{align*}
    In total, we get that
    \begin{align*}
        \E\left[|sloc_p[i]-sloc_p[j]|\right] &\ge \Pr[A]\cdot \Omega(1) \cdot \frac{M}{d} 
        + \Pr[\neg A]\cdot \Omega(1)\cdot \frac{M}{d} \\
        &\ge \Omega(1) \cdot \frac{M}{d} \\
        &= \Omega(1)\cdot \frac{\dcap{M}(x_i, x_j)}{k}.
    \end{align*}

\end{proof}

Finally, boosting with $O(\log(n))$ copies of the $O(d)$-sized embeddings gives us \Cref{thm:general-cap}.

We note that \Cref{thm:general-cap} provides a tool potentially of general interest. For example, in \Cref{thm:capped-tree}, we provided a novel proof for embedding fixed cap tree metrics into $O(\log^2(n))$ dimensions. However, note how we can also obtain the same result for this problem as a direct corollary of \Cref{thm:general-cap} and Theorem 1.1 of \cite{lee2013dimension} that embeds tree metrics into $\ell^{\log n}_1$.
\section{Discussion and Open Problems}
\label{sec:discussion}

\subsection{Embedding Probability Distributions}
Our results show that $\ell_1$ metrics that can be derived by Tree Ising Models allow for low-dimensional embeddings into $\ell_1$ with constant distortion. Recall that Tree Ising Models are the class of undirected graphical models on trees, or equivalently, the class of Bayesian networks where the underlying graph is a directed tree with all edges pointing away from a designated root. It is natural to wonder if we can obtain low-dimensional $\ell_1$ embeddings for more general classes of Bayesian networks. Tree-structured Bayesian networks are especially convenient because they allow for an efficient top-down sampling procedure. For general Bayesian networks that are not trees, the \textit{treewidth} of the underlying graph is a relevant quantity that is known to characterize the difficulty of sampling from the model \cite{peyrard2019exact}. The treewidth of a graph measures how ``tree-like" a graph is, and given that sampling from tree-structured Bayesian networks is easy, it seems natural that the difficulty of sampling from general graphical models scales with its treewidth.

Because trees have treewidth 1, our result can be stated as the following---of the class of Bayesian networks on $N$ variables whose underlying undirected graph has treewidth 1, the sub-class that has the special property that all the directed edges point away from the root allow for embedding into $\ell_1$ with $\polylog(N)$ dimensions and constant distortion. However, the question about general Bayesian networks with treewidth 1 still remains open. This class consists of directed trees where a single node can have multiple parents (also known as polytrees) \cite{dasgupta2013learning,chatterjee2022estimating,choo2023learning}, and is known to be structurally very different from trees where all edges point away from a single root. In particular, our techniques don't readily generalize to this class, because of the more complicated nature of dependencies in these models.

\begin{openproblem}
    \label{open-prob:treewidth-1}
    Let $D$ be a joint distribution on $N$ binary random variables $X_1,\dots,X_N$ given by a Bayesian network whose underlying undirected graph is a tree (has treewidth 1).\footnote{Note again that all the edges in the Bayesian network need not be pointing away from a single root node.} Consider the metric $d_D(X_i, X_j)=\Pr_D[X_i \neq X_j]$. Does this metric embed into $\ell_1$ with constant distortion and with $N^{o(1)}$ dimensions?
\end{openproblem}

The next natural direction is to see how far this agenda can be pushed---how about Bayesian networks having treewidth $> 1$? As we claimed in \Cref{claim:treewidth-3-lower-bound}, we cannot even hope to obtain a low-dimensional $\ell_1$ embedding for all metrics dictated by treewidth-3 Bayesian networks.

\claimlowerbound*
\begin{proof}[Proof sketch]
    We instantiate the lower bound instance based on the recursive diamond graph from \cite{charikar2002dimension} in the form of a graphical model. \cite{charikar2002dimension} obtain their lower bound via recursive construction in Hamming space shown in \Cref{fig:Hn}.

    \begin{figure}[H]
        \centering
        \includegraphics[scale=0.45]{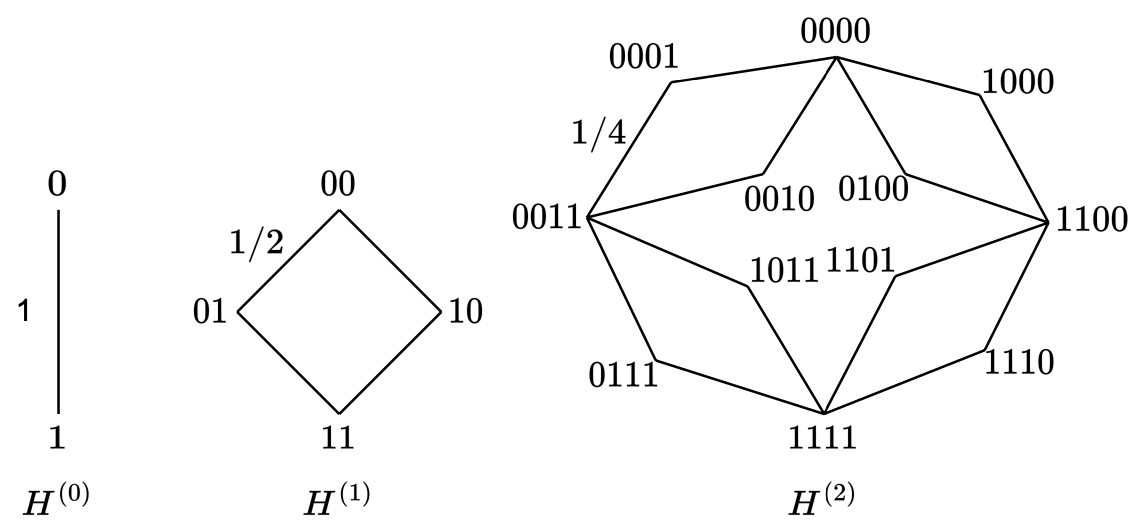}
        \caption{The recursive diamond graph in Hamming space.}
        \label{fig:Hn}
    \end{figure}

    The $n^{\text{th}}$ level of this recursive construction $H^{(n)}$ has $M=\frac{2}{3}\cdot4^n + \frac43$ many Hamming vectors. To obtain $H^{(n+1)}$ from $H^{(n)}$, any edge connecting two vectors at a (normalized) Hamming distance of $1/2^{n}$ is split up, and two new vectors that are at a distance of $1/2^{n+1}$ from each of the these are introduced. \cite[Theorem 2.1]{charikar2002dimension} asserts that embedding these vectors into $\ell_1$ with $O(1)$ distortion requires $M^{\Omega(1)}$ dimensions.
    
    We construct a Bayesian network that has the same distances as the diamond graph above. Refer to \Cref{fig:Gn}.

    \begin{figure}[H]
        \centering
        \includegraphics[scale=0.5]{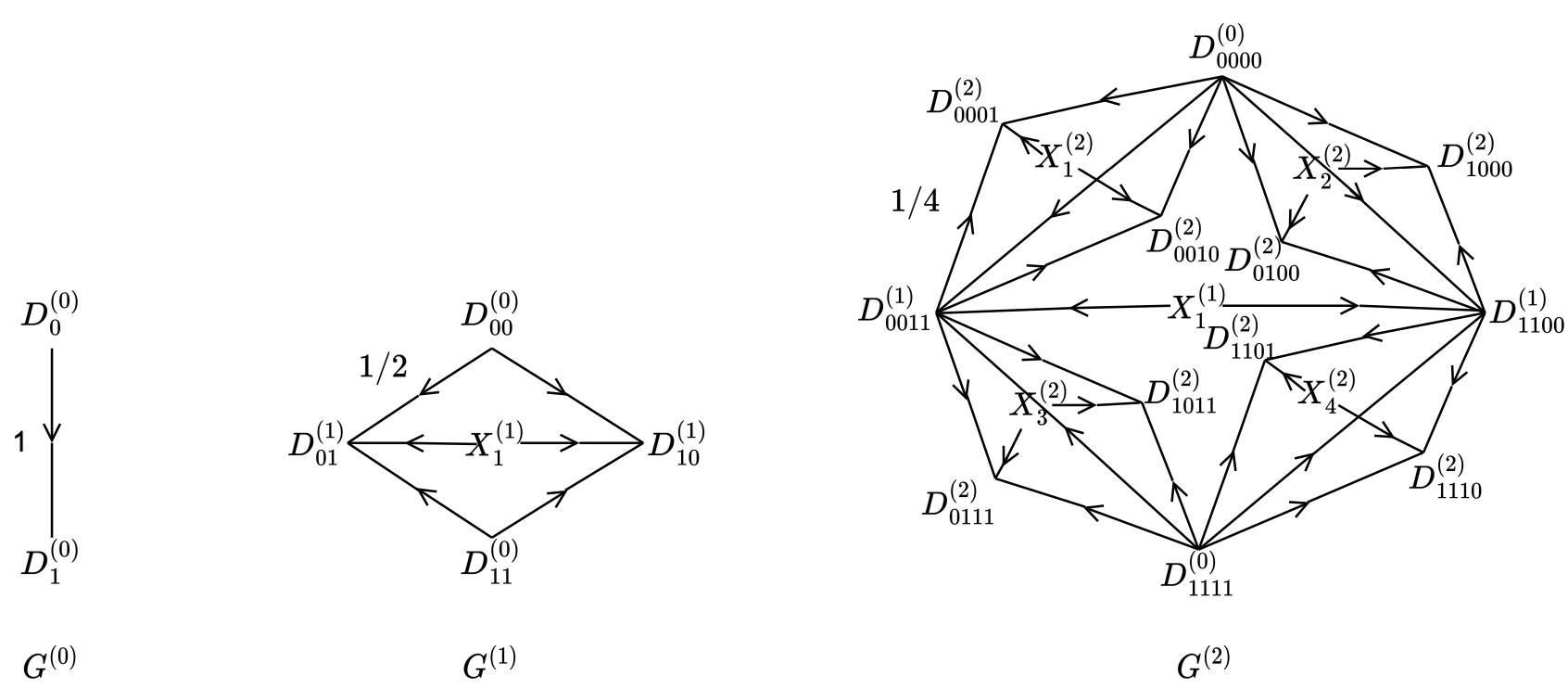}
        \caption{A Bayesian network based on the recursive diamond graph.}
        \label{fig:Gn}
    \end{figure}
    
    Each level $G^{(n)}$ of the instance consists of some $D$ nodes and some $X$ nodes. The $D$ nodes correspond to the points in Hamming space from \Cref{fig:Hn}, while the $X$ nodes simulate 50-50 coin flips. The $(n+1)^\text{th}$ level $G^{(n+1)}$ of the instance is constructed from $G^{(n)}$ by splitting every edge connecting two $D^{(n)}$ nodes at a distance of $1/2^{n}$, introducing two new $D^{(n+1)}$ nodes and an $X^{(n+1)}$ node in the process. \Cref{table:Gn} summarizes the properties of this recursive construction:

    \renewcommand{\arraystretch}{1.5}
    \begin{table}[H]
        \centering
        \begin{tabular}{ |c|c|c|c|c|c|c| } 
             \hline
             Level & $G^{(0)}$ & $G^{(1)}$ & $G^{(2)}$ & $G^{(3)}$ & \dots & $G^{(n)}$ \\ 
             \hline
             \#splittable edges & 1 & 4 & 16 & 64 & \dots & $4^n$ \\
             \#$D$ nodes & 2 & 4 & 12 & 44 & \dots & $\frac{2}{3}\cdot4^n + \frac43$ \\
             \#$X$ nodes & 0 & 1 & 5 & 21 & \dots & $\frac13\cdot4^n-\frac13$\\
             \hline
        \end{tabular}
        \caption{Properties of $G^{(n)}$.}
        \label{table:Gn}
    \end{table}

    To obtain a sample from the Bayesian network, we first independently sample a 50-50 coin flip for every $X$ node in the network. The $D$ nodes can then be sampled in a hierarchical manner. First, the $D$ nodes in the $0^{\text{th}}$ level are deterministically set: $D^{(0)}_{0\dots0}=0$ and $D^{(0)}_{1\dots1}=1$ with probability 1. Then, for $n=1,2,\dots$, we inductively realize $D^{(n)}$ nodes as follows: let $D^{(n)}_x$ and $D^{(n)}_y$ be nodes that arise from splitting an edge $(D^{(n-1)}_u, D^{(n-1)}_v)$ with node $X^{(n)}_i$, and let $x < y$ and $u < v$ lexicographically. Recall that $D^{(n-1)}_u, D^{(n-1)}_v$ and $X^{(n)}_i$ have already been realized. Then, 
    \begin{align*}
        &D^{(n)}_x \gets D^{(n-1)}_u, D^{(n)}_y \gets D^{(n-1)}_v,  \qquad\text{if } X^{(n)}_i = 0 \\
        &D^{(n)}_x \gets D^{(n-1)}_v, D^{(n)}_y \gets D^{(n-1)}_u,  \qquad\text{if } X^{(n)}_i = 1.
    \end{align*}
    In other words, $X^{(n)}_i$ acts as a ``switch" node.
    We can verify that the distances between the $D$ nodes mimic the distances between the corresponding Hamming vectors. Concretely,
    $$
        \Pr[D^{(n)}_{x} \neq D^{(n)}_{y}] = \|x-y\|_1.
    $$ 
    Furthermore, the marginal probability of a $D$ node being equal to 1 is exactly equal to the fraction of $1$'s in the Hamming vector representing it.
    
    Observe that $G^{(n)}$ has $N<2M$ many $D$+$X$ nodes in total, where $M=\frac{2}{3}\cdot4^n + \frac43$ is the number of Hamming vectors in $H^{(n)}$ above. If we were able to embed the nodes of $G^{(n)}$ into $\ell_1$ with constant distortion in $N^{o(1)}=M^{o(1)}$ dimensions, we can use the embeddings of the $D$ nodes as the embeddings of the Hamming vectors in the lower bound instance from \cite{charikar2002dimension} (i.e., \Cref{fig:Hn} above), and this will give us an $M^{o(1)}$ dimensional embedding of the recursive diamond graph, which we know is not possible. Thus, any constant-distortion embedding of $G^{(n)}$---a Bayesian network on $N$ nodes---into $\ell_1$, necessarily requires $N^{\Omega(1)}$ many dimensions. 

    We will now argue that the undirected graph underlying $G^{(n)}$ for any $n \ge 1$ has treewidth equal to 3. We will slightly abuse notation and refer to the undirected graph underlying $G^{(n)}$ by $G^{(n)}$.

    To see that the treewidth is at least 3, we note that the treewidth of a graph is at least the treewidth of any subgraph of the graph. For $n \ge 1$, observe that $G^{(n)}$ contains the subgraph underlying $G^{(1)}$, and we can verify, e.g., by brute-force, that this small subgraph has treewidth equal to 3.

    We will now show that the treewidth of $G^{(n)}$ is at most 3. For this, we appeal to the definition of the treewidth based on elimination orderings of the nodes \cite[Definition 1]{peyrard2019exact}. This definition is especially well-suited in the context of inference via variable elimination in graphical models. An equivalent form of this definition is as follows: first, we specify an ordering over the nodes in the graph. Then, we eliminate nodes in the graph in this order. Whenever we eliminate a node, we add an edge (if there isn't one already) between every pair of its neighbors that are still surviving. Before we eliminate a node, we record the count of nodes it is connected to. The maximum count that we see until we eliminate all nodes is the treewidth with respect to this ordering, and the smallest maximum count over all orderings is the graph's treewidth.

    Thus, to show that the treewidth of $G^{(n)}$ is at most 3, we only need to come up with an ordering that witnesses treewidth to be at most $3$. Given the hierarchical structure of $G^{(n)}$, this ordering is natural---we start eliminating nodes from the highest level to the lowest level. Concretely, we first eliminate all the $D^{(n)}$ nodes, in say lexicographical order. Then, we eliminate all the $X^{(n)}$ nodes (again, in say lexicographical order). We then eliminate all the $D^{(n-1)}$ nodes, followed by all the $X^{(n-1)}$ nodes, and so on, all the way until we empty the graph.

    \begin{figure}[H]
         \centering
         \begin{subfigure}[t]{0.24\textwidth}
             \centering
             \includegraphics[width=\textwidth]{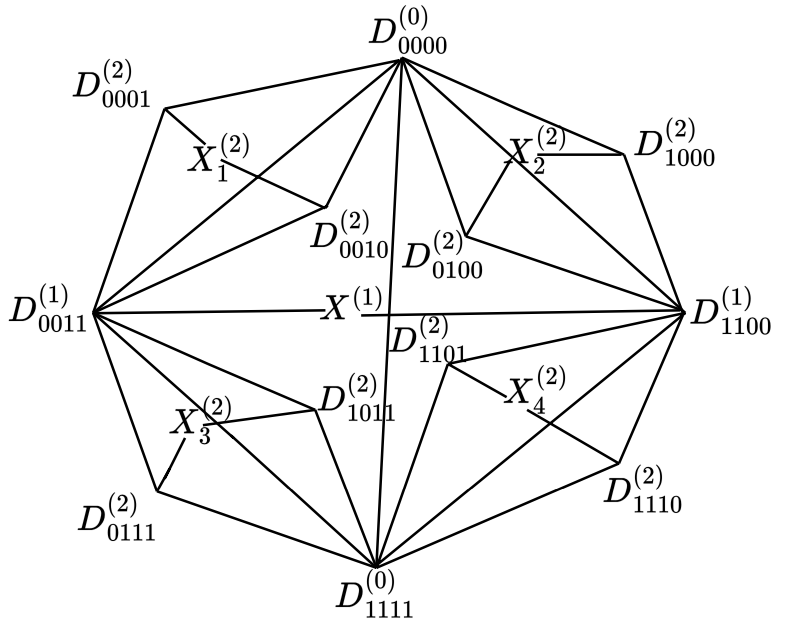}
             \caption{$G^{(2)}$}
             \label{fig:y equals x}
         \end{subfigure}
         \hfill
         \begin{subfigure}[t]{0.24\textwidth}
             \centering
             \includegraphics[width=\textwidth]{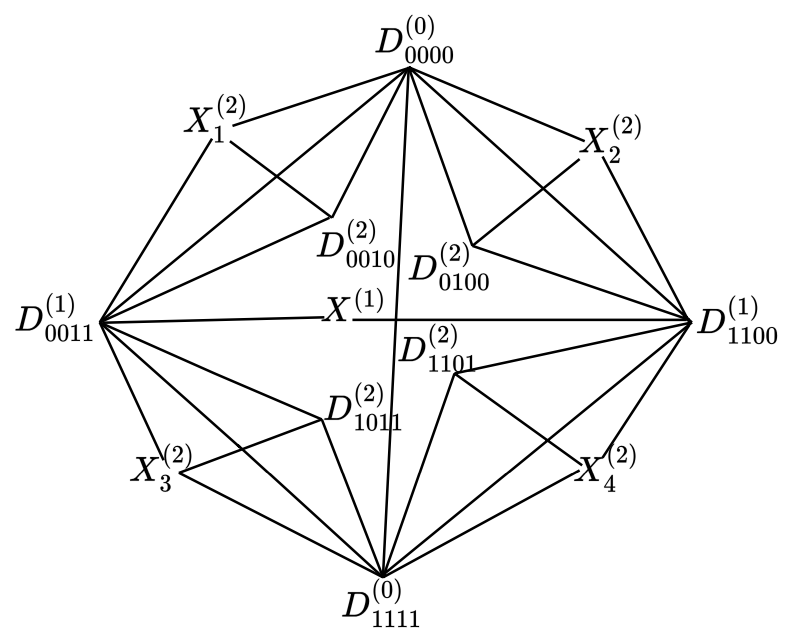}
             \caption{Half the $D^{(2)}$ nodes eliminated.}
             \label{fig:three sin x}
         \end{subfigure}
         \hfill
         \begin{subfigure}[t]{0.24\textwidth}
             \centering
             \includegraphics[width=\textwidth]{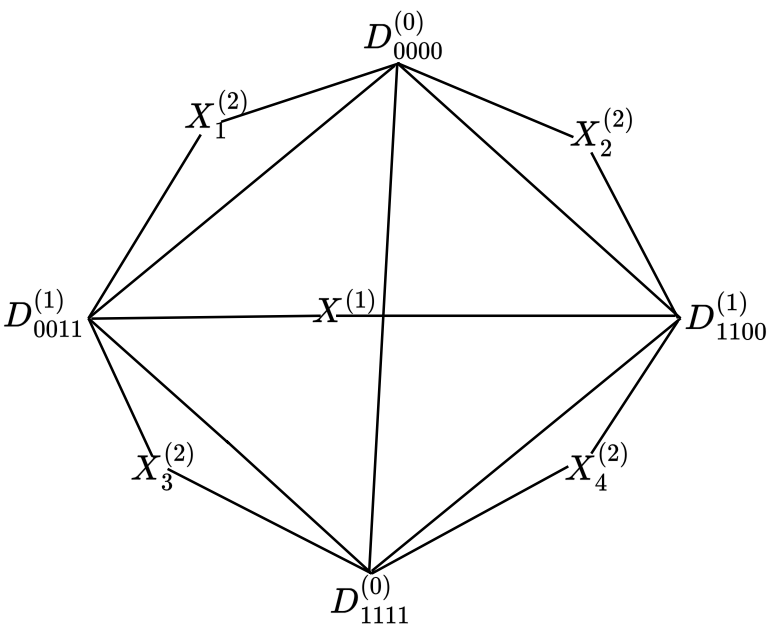}
             \caption{All $D^{(2)}$ nodes eliminated.}
             \label{fig:five over x}
         \end{subfigure}
         \hfill
         \begin{subfigure}[t]{0.24\textwidth}
             \centering
             \includegraphics[width=\textwidth]{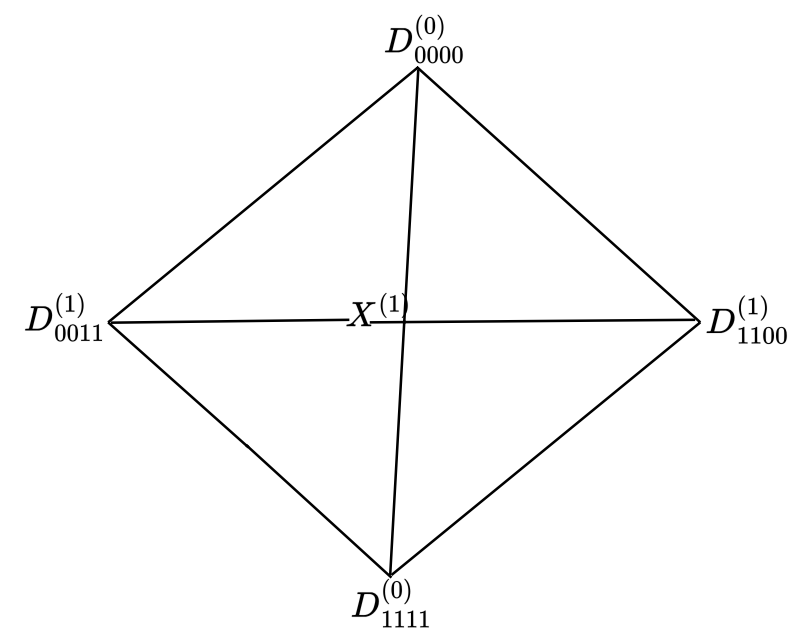}
             \caption{$X^{(2)}$ nodes eliminated. This is $G^{(1)}$.}
             \label{fig:}
         \end{subfigure}
            \caption{Level-wise elimination of nodes in $G^{(n)}$.}
            \label{fig:elimination}
    \end{figure}

    To see that the maximum degree of any node that is about to be eliminated is at most 3, consider what happens in the example in \Cref{fig:elimination}, where we start with $G^{(2)}$. First, we eliminate $D^{(2)}$ nodes, followed by $X^{(2)}$ nodes. Observe that the degree of any $D^{(2)}$ node prior to elimination is exactly 3---this corresponds to the two $D^{(1)}$ nodes and $X^{(1)}$ node generating it. After we eliminate a $D^{(2)}$ node, we must introduce pairwise edges between its neighbors---since there is already an edge between its parent $D^{(1)}$ nodes, we only need to add edges between the $X^{(1)}$ and $D^{(1)}$ nodes. Next, when it comes to eliminating the $X^{(2)}$ nodes, observe that the degree of these nodes prior to elimination is exactly 2---this corresponds to the two $D^{(1)}$ nodes that it split. Finally, when we eliminate these nodes, we don't need to introduce any new edges, because the $D^{(1)}$ nodes that they are adjacent to were already connected by an edge. More importantly, observe that after this round of eliminating $D^{(2)}$ and $X^{(2)}$ nodes, we end up exactly with the graph underlying $G^{(1)}$! Thus, the same argument carries forward. More generally, for any $n$, when we eliminate $D^{(n)}$ and $X^{(n)}$ nodes, we end up with the graph underlying $G^{(n-1)}$, and the maximum degree of any node prior to elimination is at most 3. Thus, this ordering witnesses that the treewidth of $G^{(n)}$ is at most 3, as required.
\end{proof}

Given that we have a (partial) positive answer for the treewidth-1 case, and that there isn't hope for a general result in the treewidth-3 case, this leaves the treewidth-2 case wide open (in addition to the open treewidth-1 case mentioned in \Cref{open-prob:treewidth-1}).

\begin{openproblem}
    \label{open-prob:treewidth-2}
    Let $D$ be a joint distribution on $N$ binary random variables $X_1,\dots,X_N$ given by a Bayesian network whose underlying undirected graph has treewidth 2. Consider the metric $d_D(X_i, X_j)=\Pr_D[X_i \neq X_j]$. Does this metric embed into $\ell_1$ with constant distortion and with $N^{o(1)}$ dimensions?
\end{openproblem}

Finally, our lower bound for treewidth-3 is not an Ising model, so it remains open whether more general classes of Ising models permit low dimension embeddings. For example, it is simple to embed an Ising model whose underlying graph is a cycle.\footnote{Consider conditioning on the value of an arbitrary node in the cycle. The induced distribution is a tree Ising model (with an external field). Thus, any Ising model over a cycle can be written as the mixture of two tree Ising models, and accordingly we can embed its corresponding metric in $O(\log^2(N))$ dimension with $\Theta(1)$ distortion.}

\begin{openproblem}
    \label{open-prob:more-ising}
    Let $D$ be a joint distribution on $N$ binary random variables $X_1,\dots,X_n$ given by an Ising model. Consider the metric $d_D(X_i, X_j) = \Pr_D[X_i \ne X_j]$. For what classes of graphs (e.g. bounded treewidth, general graphs) does this metric embed into $\ell_1$ with constant distortion and with $N^{o(1)}$ dimensions? 
\end{openproblem}

\subsection{Truncated Tree Metrics}

Our \Cref{thm:general-cap} shows how to truncate arbitrary $\ell_1$ metrics with $O(1)$ distortion in near-optimal dimension blowup. However, it is still open whether particular metrics can be truncated without such blowups. For example, we have two approaches that yield $O(\log^2(n))$ dimension embeddings for $O(1)$ distortion truncated tree metrics: (i) using our $O(\log(n))$ dimension blowup (\Cref{thm:general-cap}) truncation approach with the $O(\log(n))$ dimension tree metric embedding of \cite{lee2013dimension} that leverages re-randomization such as in \cite{schulman1996coding}, and (ii) using our more ad-hoc \Cref{thm:capped-tree} that leverages ideas from the $O(\log^2(n))$ dimension tree metric embedding of \cite{charikar2002dimension} without incurring dimension blowup. We know truncated tree metrics must use $\Omega(\log(n))$ dimensions, as this is the case simply for tree metrics. We do not know whether our approaches obtaining $O(\log^2(n))$ are optimal, or if there may be a better embedding.

\begin{openproblem}
    \label{open-prob:truncated-tree}
    What is the minimum dimension required to embed truncated tree metrics with $\Theta(1)$ distortion?
\end{openproblem}

\section*{Acknowledgements}
We thank Frederic Koehler for helpful discussions. This work is supported by Moses Charikar's and Gregory Valiant's Simons Investigator Awards, Tselil Schramm's NSF CAREER Grant no. 2143246, and the National Defense Science \& Engineering Graduate (NDSEG) Fellowship Program.

\bibliographystyle{alpha}
\bibliography{references}

\appendix
\section{Lazy Snaking for Lipschitz Cap Line Metrics}
\label{sec:lipschitz-line-proof}
We restate and prove \Cref{lemma:lazy-snake-line-lipschitz-cap-lb}:
\lazysnakelipschitzcapline*
\begin{proof}
    The first part of the claim is straightforward, since snaking can only reduce distance between two nodes. \\
    For the second part, we will instead show that $\E[|\embed(i)-\embed(j)|] \ge \Omega(1)\cdot \dlcap{M'}(x_i, x_j)$, where $M'(t) = \frac{47}{7447}M(t)$. This is sufficient, because $\dlcap{M'}(x_i, x_j) \ge \frac{47}{7447}\dlcap{M}(x_i, x_j)$. Fix some $i < j$. Before reaching $loc[i]$, let $t_1$ be the last time we were at the origin and decided to either rest/snake, such that $t_1 \le loc[i]$. Then, it must have been the case that either (i) we decided to snake at $t_1$, and returned back to the origin at a time $t_2>loc[i]$, or (ii) we decided to rest but there was not another decision before $t_2>loc[i]$. Namely, we have
    \begin{equation}
        \label{eqn:lipschitz-line-proof-eqn1}
        \frac{M(t_1)}{300} \le t_2 - t_1 \le \max \left(\frac{M(t_1)}{300} , \frac{2M(t_1)}{100} \right)= \frac{M(t_1)}{50}.
    \end{equation}
    \begin{equation}
        \label{eqn:lipschitz-line-proof-eqn2}
        t_2 > loc[i] \implies t_1 + \frac{M(t_1)}{50} > loc[i] \implies loc[i]-t_1 < \frac{M(t_1)}{50}.
    \end{equation}
    From Lipschitzness, we have,
    \begin{align}
        &|M(loc[i])-M(t_1))| \le |loc[i]-t_1| \le |t_2-t_1| \le \frac{M(t_1)}{50} \nonumber \\
        \implies \qquad & \frac{49}{50}M(t_1) \le M(loc[i]) \le \frac{51}{50}M(t_1).
    \end{align}

    Combining with \Cref{eqn:lipschitz-line-proof-eqn2}, we get
    \begin{equation}
        \label{eqn:lipschitz-line-proof-eqn3}
        loc[i]-t_1 < \frac{M(loc[i])}{49}.
    \end{equation}
    Observe also that by Lipschitzness, for every $t \in [t_1,t_2]$, we have
    \begin{align}
        &|M(loc[i])-M(t)| \le |loc[i]-t| \le |t_2-t_1| \le \frac{M(t_1)}{50} \le \frac{M(loc[i])}{49}\\
        \implies \qquad & \frac{48M(loc[i])}{49} \le M(t) \le \frac{50M(loc[i])}{49}.
    \end{align}
    Let $E$ be the event that $snake[t_1]=0$ for some $t_1 \in \left[loc[i]-\frac{M(loc[i])}{49}, loc[i]\right]$, such that we are ready to flip a coin at $t_1$. Then, \Cref{eqn:lipschitz-line-proof-eqn3} implies that $\Pr[E]=1$. Repeating the same argument as above, but from the perspective of $t_1$, we get that there also exists a $t_0\in \left[loc[i]-\frac{M(loc[i])}{49}-\frac{50}{49}\frac{M(loc[i])}{49},loc[i]-\left(\frac{50}{51}\right)^2\frac{M(loc[i])}{300}\right]$ such that with probability 1, $snake[t_0]=0$ and we about to flip at $t_0$. We note here that $t$ was initialized to $-2M(0)$ in Step 1 of \Cref{algo:lipschitz-lazy-snake} precisely to ensure that such a $t_0$ always exists. Initializing in this manner ensures that there is always a distance of $\frac{M(loc[i])}{49}+\frac{50}{49}\frac{M(loc[i])}{49}$ before every $loc[i]$---if there wasn't, Lipschitzness would be violated.

    We will first do the analysis for the setting where $M(loc[i]) \ge M(loc[j])$, so that $\max(M(loc[i]), M(loc[j]))$ $ = M(loc[i])$. The analysis for $M(loc[i]) < M(loc[j])$ is very similar, and is presented afterwards.
    We have three cases: \\
    
    \noindent \textbf{Case 1a:} $0 < |loc[i]-loc[j]|\le \frac{47}{7447}M(loc[i]).$ \\
    Here, $\dlcap{M'}(x_i, x_j)=|loc[i]-loc[j]|$. First, we invoke \Cref{eqn:lipschitz-line-proof-eqn3} and condition on arriving at $t_1 \in \left[ loc[i]-\frac{M(loc[i])}{49},loc[i]\right]$, where $snake[t_1]=0$ and we are about to flip a coin. Starting at $t_1$, observe that with at most $\frac{M(loc[i])}{49}\cdot \frac{49\cdot 300}{48M(loc[i])}=\Theta(1)$ ``rest" flips, we get to a $t \in \left[loc[i]-\frac{50M(loc[i])}{49\cdot300}, loc[i]\right]$ such that $snake[t]=0$. This is because $t_1$ is at most $\frac{M(loc[i])}{49}$ away from $loc[i]$, and the each resting period is at least $\frac{48M(loc[i])}{49\cdot300}$ and at most $\frac{50M(loc[i])}{49\cdot300}$. Now, at $t$, we can flip a ``snake", for which the snaking width is at least $\frac{48M(loc[i])}{49\cdot100}$. Observe that 1) $loc[i]$ is at most $\frac{50M(loc[i])}{49\cdot300}$ from $t$, 2) $loc[j]$ is at most $\frac{47}{7447}M(loc[i])$ away from $loc[i]$ and 3) $\frac{50M(loc[i])}{49\cdot300} + \frac{47}{7447}M(loc[i]) \le \frac{48M(loc[i])}{49\cdot100}$. Together, this ensures that both $loc[i]$ and $loc[j]$ are encountered going forward from $t$ within one snaking width, i.e., $snake[loc[j]]-snake[loc[i]]=loc[j]-loc[i]=\dlcap{M'}(x_i, x_j)$. \\

    \noindent \textbf{Case 2a:} $\frac{47}{7447}M(loc[i]) < |loc[i]-loc[j]|\le \frac{99}{2401}M(loc[i]).$ \\
    In this case, $\dlcap{M'}(x_i, x_j)=\frac{47}{7447}M(loc[i])$. Repeating the same reasoning as above, we start at $t_1$, and then get to a $t \in \left[loc[i]-\frac{50M(loc[i])}{49\cdot300}, loc[i]\right]$ such that $snake[t]=0$ with a constant number of ``rest" flips. Now, we flip another ``rest", which ensures that $snake[loc[i]]=0$, and gets us to a $t \in \left[loc[i], loc[i]+\frac{50M(loc[i])}{49\cdot300}\right]$. Note that since $|loc[i]-loc[j]| > \frac{47}{7447}M(loc[i]) > \frac{50M(loc[i])}{49\cdot300}$ we know $t < loc[j]$ and that there is still some distance left to get to $loc[j]$. Now, for any $t \in [loc[i], loc[j]]$, by Lipschitzness, we have 
    $$
        |M(t)-M(loc[i])| \le |loc[i]-loc[j]| \le \frac{99}{2401}M(loc[i]),
    $$
    which means that $M(t) \in \left[\frac{2392}{2401}M(loc[i]), \frac{2500}{2401}M(loc[i])\right]$. Therefore, with at most $\frac{99M(loc[i])}{2401}\cdot \frac{2401\cdot300}{2392 M(loc[i])}=\Theta(1)$ ``rest" flips, we get to a $t$ which is in $\left[loc[j]-\frac{50M(loc[i])}{49\cdot300\cdot2}-\frac{2500M(loc[i])}{2401\cdot300}, loc[j]-\frac{50M(loc[i])}{49\cdot300\cdot2}\right]$. Once here, we can flip a ``snake"---noting that 1) the snaking width is at least $\frac{2392}{2401\cdot100}M(loc[i])$ 2) the distance to $loc[j]$ is at most $\frac{50M(loc[i])}{49\cdot300\cdot2}+\frac{2500M(loc[i])}{2401\cdot300}$ and 3) the quantity in 2) is smaller than the quantity in 1), we get that $snake[loc[j]]$ is at least $\frac{50M(loc[i])}{49\cdot300\cdot2}$, which is at least a constant fraction of $\dlcap{M'}(x_i, x_j)$. \\

    \noindent \textbf{Case 3a:} $\frac{99}{2401}M(loc[i]) < |loc[i]-loc[j]|.$ \\
    In this case, $\dlcap{M'}(x_i, x_j) = \frac{47}{7447}M(loc[i])$. We argued above that with probability 1, there exists a $t_0 \in \left[loc[i]-\frac{M(t_1)}{50} - \frac{M(t_0)}{50}, loc[i] - \frac{M(t_0)}{300} \right] \in \left[loc[i]-\frac{99M(loc[i])}{2401},loc[i]-\left(\frac{50}{51}\right)^2\frac{M(loc[i])}{300}\right]$ such that $snake[t_0]=0$ and we are about to flip a coin. Observe that by Lipschitzness, for any $t \in [t_0, loc[i]]$, 
    $$M(t) \in \left[\frac{2392}{2401}M(loc[i]), \frac{2500}{2401}M(loc[i])\right].$$
    Therefore, with at most $\frac{99M(loc[i])}{2401}\cdot \frac{2401\cdot300}{2392M(loc[i])}=\Theta(1)$ ``rest" flips, we get to a  $t$ which is in
    $$\left[loc[i]-\left(\frac{50}{51}\right)^2\frac{M(loc[i])}{300}-\frac{2500M(loc[i])}{2401\cdot300}, loc[i]-\left(\frac{50}{51}\right)^2\frac{M(loc[i])}{300}\right].$$
    Now, we may flip a ``snake"---noting that 1) the snaking width is at least $\frac{2392M(loc[i])}{2401\cdot100}$ and 2) the distance to $loc[i]$ is at least $\left(\frac{50}{51}\right)^2\frac{M(loc[i])}{300}$ and at most $\left(\frac{50}{51}\right)^2\frac{M(loc[i])}{300}+\frac{2500M(loc[i])}{2401\cdot300}<\frac{2392M(loc[i])}{2401\cdot100}$, we can conclude that by the time we get to $loc[i]$, $snake[loc[i]]$ will at least be $\left(\frac{50}{51}\right)^2\frac{M(loc[i])}{300}=\Omega(1)\cdot\dlcap{M'}(x_i, x_j)$. Next, we know that with probability 1, there exists a $t'\in \left[loc[j]-\frac{M(loc[j])}{49}, loc[j]\right]$ such that $snake[t']=0$ and we are about flip a coin. Further, this $t'$ is larger than $loc[i]$, because $loc[j]-\frac{M(loc[j])}{49}$ needs to be larger than $loc[i]$ under the case we are studying. We catch hold of this $t'$, and with a constant number of ``rest" flips, we ensure that $snake[loc[j]]=0$. Thus, we obtain that $|snake[loc[i]]-snake[loc[j]]| \ge \Omega(1)\cdot\dlcap{M'}(x_i, x_j)$. \\
    
    Now, we will do the analysis for when $M(loc[i]) < M(loc[j])$, so that $\max(M(loc[i]), M(loc[j]))$ $ = M(loc[j])$. Again, we have three cases: \\

    \noindent \textbf{Case 1b:} $0 < |loc[i]-loc[j]|\le \frac{47}{7447}M(loc[j]).$ \\
    Here, $\dlcap{M'}(x_i, x_j)=|loc[i]-loc[j]|$. By Lipschitzness, we have $M(loc[i])\in \left[\frac{7400}{7447}M(loc[j]),\frac{7494}{7447}M(loc[j])\right]$. First, we condition on arriving at $t_1$, where $snake[t_1]=0$ and we are about to flip a coin. Starting at $t_1$, observe that with at most $\frac{M(loc[i])}{49}\cdot \frac{49\cdot 300}{48M(loc[i])}=\Theta(1)$ ``rest" flips, we get to a $t \in \left[loc[i]-\frac{50M(loc[i])}{49\cdot300}, loc[i]\right]$ such that $snake[t]=0$ and we are about to flip a coin. Now, at $t$, we can flip a ``snake", for which the snaking width is at least $\frac{48M(loc[i])}{49\cdot100} \ge \frac{7400\cdot48M(loc[j])}{7447\cdot49\cdot100}$. Observe that 1) $loc[i]$ is at most $\frac{50M(loc[i])}{49\cdot300} \le \frac{7494\cdot50M(loc[j])}{7447\cdot49\cdot300}$ from $t$, 2) $loc[j]$ is at most $\frac{47}{7447}M(loc[j])$ away from $loc[i]$ and 3) $\frac{7494\cdot50M(loc[j])}{7447\cdot49\cdot300} + \frac{47}{7447}M(loc[j]) \le \frac{7400\cdot48M(loc[j])}{7447\cdot49\cdot100}$. Together, this ensures that both $loc[i]$ and $loc[j]$ are encountered going forward from $t$ within one snaking width, i.e., $snake[loc[j]]-snake[loc[i]]=loc[j]-loc[i]=\dlcap{M'}(x_i, x_j)$. \\

    \noindent \textbf{Case 2b:} $\frac{47}{7447}M(loc[j]) < |loc[i]-loc[j]|\le \frac{99}{2401}M(loc[j]).$ \\
    In this case, $\dlcap{M'}(x_i, x_j)=\frac{47}{7447}M(loc[j])$. Repeating the same reasoning as above, we start at $t_1$, and then get to a $t \in \left[loc[i]-\frac{50M(loc[i])}{49\cdot300}, loc[i]\right]$ such that $snake[t]=0$ with a constant number of ``rest" flips. Now, we flip another ``rest", which ensures that $snake[loc[i]]=0$, and gets us to a $t \in \left[loc[i], loc[i]+\frac{50M(loc[i])}{49\cdot300}\right]$. Since, $M[loc[i]] < M[loc[j]]$, $t \in \left[loc[i], loc[i]+\frac{50M(loc[j])}{49\cdot300}\right]$. Note that since $|loc[i]-loc[j]| > \frac{47}{7447}M(loc[j])> \frac{50M(loc[j])}{49\cdot300}$ we know $t < loc[j]$ and there is still some distance left to reach $loc[j]$. Now, for any $t \in [loc[i], loc[j]]$, by Lipschitzness, we have 
    $$
        |M(t)-M(loc[j])| \le |loc[i]-loc[j]| \le \frac{99}{2401}M(loc[j]),
    $$
    which means that $M(t) \in \left[\frac{2392}{2401}M(loc[j]), \frac{2500}{2401}M(loc[j])\right]$. Therefore, with at most $\frac{99M(loc[j])}{2401}\cdot \frac{2401\cdot300}{2392 M(loc[j])}=\Theta(1)$ ``rest" flips, we get to a $t$ which is in $\left[loc[j]-\frac{50M(loc[j])}{49\cdot300\cdot2}-\frac{2500M(loc[j])}{2401\cdot300}, loc[j]-\frac{50M(loc[j])}{49\cdot300\cdot2}\right]$. Once here, we can flip a ``snake"---noting that 1) the snaking width is at least $\frac{2392}{2401\cdot100}M(loc[j])$ and 2) the distance to $loc[j]$ is at most $\frac{50M(loc[j])}{49\cdot300\cdot2}+\frac{2500M(loc[j])}{2401\cdot300}$, and 3) the quantity in 2) is smaller than the quantity in 1), we get that $snake[loc[j]]$ is at least $\frac{50M(loc[j])}{49\cdot300\cdot2}$, which is at least a constant fraction of $\dlcap{M'}(x_i, x_j)$. \\

    \noindent \textbf{Case 3b:} $\frac{99}{2401}M(loc[j]) < |loc[i]-loc[j]|.$ \\
    In this case, $\dlcap{M'}(x_i, x_j) = \frac{47}{7447}M(loc[j])$. We know that with probability 1, there exists a 
    $$t'\in \left[loc[i]-\frac{M(loc[i])}{49}, loc[i]\right]$$
    such that $snake[t']=0$ and we are about flip a coin. We catch hold of $t'$, and with a constant number of ``rest" flips, we ensure that $snake[loc[i]]=0$. Next, we argued above that with probability 1, there exists a $t_0 \in \left[ loc[j]-\frac{M(t_0)}{50} - \frac{M(t_1)}{50}\right] \in \left[loc[j]-\frac{99M(loc[j])}{2401},loc[j]-\left(\frac{50}{51}\right)^2\frac{M(loc[j])}{300}\right]$ such that $snake[t_0]=0$ and we are about to flip a coin. Observe that by Lipschitzness, for any $t \in [t_0, loc[j]]$, $M(t_0) \in \left[\frac{2392}{2401}M(loc[j]), \frac{2500}{2401}M(loc[j])\right]$. Therefore, with at most $\frac{99M(loc[j])}{2401}\cdot \frac{2401\cdot300}{2392M(loc[j])}=\Theta(1)$ ``rest" flips, we get to a  $t$ which is in $\left[loc[j]-\left(\frac{50}{51}\right)^2\frac{M(loc[j])}{300}-\frac{2500M(loc[j])}{2401\cdot300}, loc[j]-\left(\frac{50}{51}\right)^2\frac{M(loc[j])}{300}\right]$. Now, we may flip a ``snake"---noting that 1) the snaking width is at least $\frac{2392M(loc[j])}{2401\cdot100}$ and 2) the distance to $loc[j]$ is at least $\left(\frac{50}{51}\right)^2\frac{M(loc[j])}{300}$ and at most $\left(\frac{50}{51}\right)^2\frac{M(loc[j])}{300}+\frac{2500M(loc[j])}{2401\cdot300}<\frac{2392M(loc[j])}{2401\cdot100}$, we can conclude that by the time we get to $loc[j]$, $snake[loc[j]]$ will at least be $\left(\frac{50}{51}\right)^2\frac{M(loc[j])}{300}=\Omega(1)\cdot\dlcap{M'}(x_i, x_j)$.

    Finally, note that in every case, we showed that, conditioned on the outcomes of at most a constant number of coin flips, $|snake[loc[i]]-snake[loc[j]]| = |\embed[i]-\embed[j]| \ge \Omega(1)\cdot \dlcap{M'}(x_i, x_j)$. This immediately gives that $\E[|\embed[i]-\embed[j]|] \ge \Omega(1)\cdot \dlcap{M'}(x_i, x_j)$, completing the proof.
\end{proof}

\section{Proofs for General Capped \texorpdfstring{$\ell_1$}{\l} Metrics}
\label{sec:proofs-general-l1-metrics}

\subsection{Coordinate-wise Embeddings are Good}
\label{sec:proofs-coordinate-wise}
In this section, we will essentially reuse the analysis of the build-clean procedure from \Cref{sec:lipschitz-tree} in order to prove \Cref{lemma:coordinate-wise-build-clean}. We restate the lemma here for convenience.
\coordinatewisebc*
\begin{proof}
    The proof is constructive. Without loss of generality, consider say the first coordinate of each of the $n$ points. We overload notation here and identify this coordinate of each point with $x_1,\dots,x_n$. If we map these coordinates to the line, we obtain a line graph metric. Our goal is to embed this line metric with a fixed cap $M$---we will do so using the build-clean framework from \Cref{sec:lipschitz-tree}. We will root the line at the leftmost point on the line (say this is $r$). For technical reasons again, we add an additional node $r_0$ at a distance $2M$ to the left of $r$. We will split up the line into edges of length $\frac{M}{100d}$---these edges correspond to the ``shortened" caterpillar edges from the tree decomposition in \Cref{sec:lipschitz-tree}. We group edges into build and clean stages. As before, each edge $e$ has auxiliary information $a(e)$ associated with it which determines whether it is processed in the build/clean stage. We set $\lb=4d$ and $\lc=9d$. The auxiliary information for the leftmost/first edge is chosen uniformly at random from $\{0,\dots,\lb+\lc-1\}$. Thereafter, for every edge $e$, $a(e)=a(par(e))+1 \mod (\lb + \lc)$, where $par(e)$ is simply the adjacent edge to the left of $e$. If $a(e) \in \{0,\dots,\lb-1\}$, $e$ is processed in the build stage, else it is processed in the clean stage. After processing all the edges, all the points on the line will be mapped to an embedding of size $|H|=8d$.  We construct this $|H|$-sized embedding for every location on the line metric in the same way as described in \Cref{sec:lipschitz-tree} for Lipschitz-cap trees.  In the language of the analysis there, let $\embed[1],\dots,\embed[n]$ denote the embeddings thus obtained (these correspond to the vectors $v^{(1)}_1,\dots, v^{(1)}_n$ in the language of \Cref{lemma:coordinate-wise-build-clean}). To avoid simply rewriting the entire analysis from \Cref{sec:lipschitz-tree}, we will walk through the claims therein in a stepwise manner, and only specify the parts relevant to the special case we are in, i.e., a line graph with a fixed cap.
    
    Since the cap at every location is equal to $M$, we do not need to relate the caps of points on nearby edges, as we did in \Cref{claim:edge-similar} and \Cref{claim:path-similar}.
    
    Next, \Cref{claim:clean} holds too---at any location $x$ immediately after processing an entire clean stage, all the coordinates of $\embed[x]$ are identically 0. While the same proof in \Cref{claim:bad-edges-timwef} works, we can also see this more directly here. The build stage, by virtue of processing only at most $4d$ edges, fills in only at most $4d$ coordinates to a value $\frac{M}{100d}$, while a complete clean stage comprises of $9d$ edges of length $\frac{M}{100d}$, and hence definitely zeros out all the nonzero coordinates.
    
    Because of complete cleaning, the proof in \Cref{claim:norm-bound-after-clean} would similarly give us that for any location $x$, $\|\embed[x]\|_1 \le \frac{M}{2}$ (actually, we would obtain a stronger bound of $\frac{M}{25}$, but the weaker bound suffices).
    
    At this point, we have already proved part (1) of \Cref{lemma:coordinate-wise-build-clean}. Because of the way the build-clean procedure works, the embeddings as we move continuously along the line are coordinate-wise Lipschitz, and only ever capture distance in one coordinate along adjacent locations. This already gives us that $\|\embed[i]-\embed[j]\|_1 \le |x_i -x_j|$. From the norm bound on the embedding, we also have that
    $$
    \|\embed[i]-\embed[j]\|_1 \le \|\embed[i]\|_1 + \|\embed[j]\|_1 \le M.
    $$
    Hence, we have that $\|\embed[i]-\embed[j]\|_1 \le \min(|x_i-x_j|, M)$.
    
    Next, we turn our attention towards proving part (2) of the lemma. Consider the guarantee given by \Cref{subclaim:leg}. Let $x_i$ be to the left of $x_j$ without loss of generality. Then, the LCA of $x_i$ and $x_j$ on the line would simply be $x_i$. The same proof, with ``$\log(n)$" replaced with ``$d$" would yield that 
    \begin{align}
        \E\left[\|\embed[x_i]-\embed[x_j]\|_1\right] \ge \Omega(1) \cdot \min(|x_i-x_j|, M),
    \end{align}
    and this is precisely what we wanted to prove for the purposes of part(2) in \Cref{lemma:coordinate-wise-build-clean}.
\end{proof}

\subsection{Concatenations of Coordinate-wise Embeddings are Good}
\label{sec:proofs-for-zi}

First, we state a reverse application of Markov's inequality:
\begin{proposition}[Reverse Markov]
    \label{prop:reverse-markov}
    Let $F$ be a random variable such that $F \in [0, M]$ with probability 1, and $\E[F] \ge c_1 M$, for $c_1 \le 1$. Then, for any $C_2 \ge 1$, 
    \begin{align*}
        \Pr\left[F \ge \frac{c_1}{C_2}\cdot M\right] \ge 1-\frac{1-c_1}{1-\frac{c_1}{C_2}}.
    \end{align*}
\end{proposition}
\begin{proof}
    Define $G = M-F$, so that $G \in [0, M]$. Then $\E[G] = M - \E[F] \le M-c_1M$. Applying Markov's inequality,
    \begin{align*}
        &\Pr\left[F < \frac{c_1}{C_2}\cdot M \right] = \Pr\left[G > \left(1-\frac{c_1}{C_2} \right)\cdot M \right] \le \frac{1-c_1}{1-\frac{c_1}{C_2}} \\
        \implies \qquad& \Pr\left[F \ge \frac{c_1}{C_2}\cdot M\right] \ge 1-\frac{1-c_1}{1-\frac{c_1}{C_2}}.
    \end{align*}
\end{proof}

We can now use \Cref{prop:reverse-markov} to prove the following general result which we can directly use to prove that the $z_i$ vectors constructed in \Cref{sec:general-cap-metrics} capture the capped distance well.
\begin{lemma}
    \label{lemma:reverse-markovish}
    Suppose $R_1,\dots,R_d$ are independent random variables such that each $R_i \in [0, M]$ with probability 1, and $\E[\sum_{i=1}^dR_i] \ge c \cdot M$ for a constant $c \le 1$. Then,
    \begin{align*}
        \Pr\left[\sum_{i=1}^dR_i \ge \frac{c^2}{100}\cdot M\right] \ge \Omega(1).
    \end{align*}
\end{lemma}
\begin{proof}
    Assume for the sake of contradiction that 
    $$\Pr\left[\sum_{i=1}^dR_i \ge \frac{c^2}{100}\cdot M\right] < \min\left(\frac{1}{100}, 1- \frac{1-\left(\frac{99c}{100}\right)^2}{1-\frac{c^2}{100}}\right).$$
    In particular, since the $R_i$s are non-negative, this implies that 
    \begin{equation}
    \label{eqn:reverse-markovish-contradiction}
        \Pr\left[\sum_{i=1}^jR_i \ge \frac{c^2}{100}\cdot M\right] < \min\left(\frac{1}{100}, 1- \frac{1-\left(\frac{99c}{100}\right)^2}{1-\frac{c^2}{100}}\right) \quad\text{for all $j\in\{1,\dots,d\}$}.
    \end{equation}
    We will reason about the random variable
    $$
        \min\left(\sum_{i=1}^dR_i, c M\right) = \sum_{i=1}^d\min\left[R_i, \max\left(0, c M - \sum_{j=1}^{i-1} R_j\right)\right].
    $$
    Let $A_i$ be the event that $\sum_{j=1}^iR_j < \frac{c}{100}\cdot M$ and $B_i$ be the event that $\sum_{j=1}^iR_j < \frac{c^2}{100}\cdot M$. Further, let $\Ind[A_i]$ and $\Ind[B_i]$ be their indicators. Note that $\Pr[B_i] \le \Pr[A_i]$, since $c \le 1$. Then, observe that
    \begin{align*}
        \min\left[R_i, \max\left(0, c M - \sum_{j=1}^{i-1} R_j\right)\right] &\ge \Ind[A_{i-1}]\cdot \min\left(R_i, \frac{99c}{100}\cdot M\right) \\
        &\ge \Ind[A_{i-1}]\cdot \frac{99c}{100}\cdot R_i.
    \end{align*}
    Here, the last inequality follows because $\frac{99c}{100}\cdot R_i \le \frac{99c}{100}\cdot M$ as well as $\frac{99c}{100}\cdot R_i \le R_i$, since $c \le 1$. Thus, we obtain that
    \begingroup
    \allowdisplaybreaks
        \begin{align*}
            \E\left[\min\left(\sum_{i=1}^dR_i, c M\right)\right] &= \E\left[ \sum_{i=1}^d\min\left[R_i, \max\left(0, c M - \sum_{j=1}^{i-1} R_j\right)\right]\right] \\
            &\ge \E\left[ \sum_{i=1}^d \Ind[A_{i-1}]\cdot \frac{99c}{100}\cdot R_i \right]\\
            &= \sum_{i=1}^d \E\left[\Ind[A_{i-1}]\cdot \frac{99c}{100}\cdot R_i \right] \qquad(\text{linearity of expectation}) \\
            &= \frac{99c}{100}\cdot\sum_{i=1}^d \E\left[\Ind[A_{i-1}]\right]\cdot \E\left[R_i \right] \qquad \text{(independence of $R_i$s)} \\
            &= \frac{99c}{100}\cdot\sum_{i=1}^d \Pr\left[A_{i-1}\right]\cdot \E\left[R_i \right]\\
            &\ge \frac{99c}{100}\cdot\sum_{i=1}^d \Pr\left[B_{i-1}\right]\cdot \E\left[R_i \right] \qquad \left(\text{since $\Pr[B_{i-1}] \le \Pr[A_{i-1}]$}\right)\\
            &\ge \frac{99c}{100}\cdot\frac{99}{100}\cdot \sum_{i=1}^d\E\left[R_i \right] \qquad \text{(\Cref{eqn:reverse-markovish-contradiction})} \\
            &\ge \left(\frac{99c}{100}\right)^2 M.
        \end{align*}
    \endgroup
    Now, we invoke \Cref{prop:reverse-markov} with $F = \min\left(\sum_{i=1}^dR_i, c M\right)$. Note that $F \in [0, M]$ with probability 1, and $\E[F] \ge \left(\frac{99c}{100}\right)^2 \cdot M$ --- thus, for $C_2=100\cdot \left(\frac{99}{100}\right)^2$, we get
    \begin{align*}
        \Pr\left[\min\left(\sum_{i=1}^dR_i, c M\right) \ge \frac{c^2}{100}\cdot M\right] &\ge 1-\frac{1-\left(\frac{99c}{100}\right)^2}{1-\frac{c^2}{100}}. 
    \end{align*}
    Notice that we have ended up showing 
    \begin{align*}
        \Pr\left[\sum_{i=1}^d R_i \ge \frac{c^2}{100}\cdot M\right] \ge \Pr\left[\min\left(\sum_{i=1}^dR_i, c M\right) \ge \frac{c^2}{100}\cdot M\right] &\ge 1-\frac{1-\left(\frac{99c}{100}\right)^2}{1-\frac{c^2}{100}},
    \end{align*}
    which contradicts \Cref{eqn:reverse-markovish-contradiction}. Hence, it must be true that 
    $$
    \Pr\left[\sum_{i=1}^dR_i \ge \frac{c^2}{100}\cdot M\right] \ge \min\left(\frac{1}{100}, 1- \frac{1-\left(\frac{99c}{100}\right)^2}{1-\frac{c^2}{100}}\right) = \Omega(1).
    $$
\end{proof}
We are now ready to prove \Cref{lemma:zi-good}, which we resatate here for convenience.
\lemmazigood*
\begin{proof}
    Recall that $z_i$ is the concatenation of independently constructed vectors $v^{(q)}_i$ for $q \in \{1,\dots,d\}$. Let $R_q$ be the random variable $\|v^{q}_i-v^{(q)}_j\|_1$. Observe that the $R_q$'s are independent of each other, and that $\|z_i-z_j\|_1 = \sum_{q=1}^dR_q$. \\
    \noindent Case 1: $\|x_i-x_j\|_1 \le M.$ \\
    In this case, $\dcap{M}(x_i, x_j)=\|x_i-x_j\|_1$. Furthermore, for every $q$, $|x_i[q]-x_j[q]|\le \|x_i-x_j\|_1$. From part (1) of \Cref{lemma:coordinate-wise-build-clean}, this means that for every $q$, $R_q=\|v^{q}_i-v^{(q)}_j\|_1\in [0,\|x_i-x_j\|_1]$ with probability 1. Furthermore, from part (2) of \Cref{lemma:coordinate-wise-build-clean}, $\E[R_q] \ge \Omega(1)\cdot |x_i[q]-x_j[q]|$. Thus, 
    $$
        \E\left[\sum_{q=1}^dR_q\right] \ge \Omega(1)\cdot \sum_{q=1}^d|x_i[q]-x_j[q]| = \Omega(1)\cdot \|x_i-x_j\|_1.
    $$
    Instantiating \Cref{lemma:reverse-markovish} with these $R_q$ variables, we get
    \begin{align*}
        &\Pr\left[\sum_{q=1}^dR_q \ge \Omega(1)\cdot \|x_i-x_j\|_1\right] \ge \Omega(1) \\
        \implies \qquad &\Pr\left[\|z_i-z_j\|_1 \ge \Omega(1)\cdot \|x_i-x_j\|_1\right] \ge \Omega(1).
    \end{align*}
    \noindent Case 2: $\|x_i-x_j\|_1 > M.$ \\
    In this case, $\dcap{M}(x_i, x_j)=M$. Observe first that for any $q\in \{1,\dots,d\}$, by part 1 of \Cref{lemma:coordinate-wise-build-clean}, $R_q \in [0, M]$ with probability 1. Next, observe also that $\E\left[\sum_{q=1}^dR_q\right] \ge \Omega(1)\cdot M$. To see this, consider two cases. In one case, there exists a coordinate $q$ such that $|x_i[q]-x_j[q]| > M$. In this case, part (2) of \Cref{lemma:coordinate-wise-build-clean} already promises us that for this $q$, $\E[R_q] \ge \Omega(1)\cdot M$, and since the $R_q$'s are non-negative, we have the same lower bound on $\E\left[\sum_{q=1}^dR_q\right]$. In the remaining case, $|x_i[q]-x_j[q]| \le M$ for all $q$. In this case, we can still apply part (2) of \Cref{lemma:coordinate-wise-build-clean}, to say that
    \begin{align*}
        \E\left[\sum_{q=1}^dR_q\right] &= \sum_{q=1}^d\E\left[R_q\right] \\
        &\ge \sum_{q=1}^d\Omega(1)\cdot |x_i[q]-x_j[q]| \\
        &= \Omega(1)\cdot \|x_i-x_j\|_1 \\
        &\ge \Omega(1) \cdot M,
    \end{align*}
    where the last inequality follows because $\|x_i-x_j\|_1 > M$ in the present case. Thus, we have shown that every $R_q$ is in $[0, M]$, and $\E\left[\sum_{q=1}^dR_q\right] \ge \Omega(1)\cdot M$. Instantiating \Cref{lemma:reverse-markovish} gives
    \begin{align*}
        &\Pr\left[\sum_{q=1}^dR_q \ge \Omega(1)\cdot M\right] \ge \Omega(1) \\
        \implies \qquad &\Pr\left[\|z_i-z_j\|_1 \ge \Omega(1)\cdot M\right] \ge \Omega(1).
    \end{align*}
\end{proof}

\end{document}